\documentclass[a4paper,onecolumn,11pt,accepted=2023-11-26]{quantumarticle}
\pdfoutput=1
\usepackage[utf8]{inputenc}
\usepackage[english]{babel}
\usepackage[T1]{fontenc}
\usepackage{amsmath}
\usepackage{hyperref}

\usepackage{tikz}
\usepackage{lipsum}

\usepackage{braket}
\usepackage{enumerate}

\usepackage{mathtools}
\usepackage{graphicx}
\usepackage{amsmath}
\usepackage{amsfonts}
\usepackage{amssymb}
\usepackage{mathrsfs}
\usepackage{enumerate}
\usepackage{lscape}

\usepackage{tabu}%
\usepackage{longtable}

\usepackage{rotating}
\usepackage{placeins}
\usepackage{multirow}
\usepackage{color}
\usepackage{xcolor}
\usepackage{subfig}
\usepackage{blkarray,lipsum}
\usepackage{listings}
\usepackage{fancyhdr}
\usepackage{makecell}
\usepackage{mathtools}
\newcommand\scalemath[2]{\scalebox{#1}{\mbox{\ensuremath{\displaystyle #2}}}}

\usepackage{url}
\usepackage{multirow}
\usepackage{arydshln}
\usepackage[]{authblk}
    \usepackage{rotating}
    \usepackage{enumitem}
\usepackage{hyperref}
\usepackage{caption}
\usepackage{tikz}
\hypersetup{
    colorlinks=true, 
    linktoc=all,     
    linkcolor=blue,  
}
\newtheorem{theorem}{Theorem}[section]

\usepackage[ruled,vlined,noline,linesnumbered]{algorithm2e}

\newtheorem{lemma}{Lemma}[section]

\newtheorem{assumption}{Assumption}[section]

\newenvironment{proof}[1][Proof]{\textbf{#1.} }{\ \rule{0.5em}{0.5em} \vspace{1ex}}

\newcommand{\R}{\mathbb{R}}

\newcommand{\C}{\mathbb{C}}

\usepackage{dsfont}

\newcommand{\bs}{\boldsymbol}
\newcommand{\Herm}{\mathrm{Herm}}
\newcommand{\im}{\mathrm{i}}

\usepackage{algorithmic}

\DeclareMathOperator{\argmin}{argmin}

\DeclareMathOperator{\dist}{dist}

\begin{document}

\title{Iteration Complexity of Variational Quantum Algorithms}

\author{Vyacheslav Kungurtsev}
\affiliation{Department of Computer Science, Czech Technical University in Prague, Karlovo nam. 13, Prague 2, Czech Republic}
\orcid{0000-0003-2229-8824}
\author{Georgios Korpas}
\affiliation{Department of Computer Science, Czech Technical University in Prague, Karlovo nam. 13, Prague 2, Czech Republic}
\affiliation{Archimedes Research Unit on AI, Data Science and Algorithms,  Athena Research and Innovation Center, 15125 Marousi, Greece}
\orcid{0000-0003-3850-4979}
\author{Jakub Marecek}
\affiliation{Department of Computer Science, Czech Technical University in Prague, Karlovo nam. 13, Prague 2, Czech Republic}
\orcid{0000-0003-0839-0691}
\author{Elton Yechao Zhu}
\affiliation{Fidelity Center for Applied Technology, FMR LLC, Boston, MA 02210, USA}
\orcid{0000-0002-4497-2093}


\maketitle


\begin{abstract}
  There has been much recent interest in near-term applications of quantum computers, i.e., using quantum circuits that have short decoherence times due to hardware limitations. 
Variational quantum algorithms (VQA), wherein an optimization algorithm implemented on a classical computer evaluates a parametrized quantum circuit as an objective function, are a leading framework in this space. 
An enormous breadth of algorithms in this framework have been proposed for solving a range of problems in machine learning, forecasting, applied physics, and combinatorial optimization, among others.

In this paper, we analyze the iteration complexity of VQA, that is, the number of steps that VQA requires until its iterates satisfy a surrogate measure of optimality. We argue that although VQA procedures incorporate algorithms that can, in the idealized case, be modeled as classic procedures in the optimization literature, the particular nature of noise in near-term devices invalidates the claim of applicability of off-the-shelf analyses of these algorithms. 
Specifically, noise makes the evaluations of the objective function via quantum circuits \emph{biased}.
Commonly used optimization procedures, such as SPSA and the parameter shift rule, can thus be seen as derivative-free optimization algorithms with biased function evaluations, for which there are currently no iteration complexity guarantees in the literature.
We derive the missing guarantees and find that the rate of convergence is unaffected. However, the level of bias contributes unfavorably to both the constant therein, and the asymptotic distance to stationarity, i.e., the more bias, the farther one is guaranteed, at best, to reach a stationary point of the VQA objective.
\end{abstract}



\clearpage 
\section{Introduction}\label{sec:intro}
There has been a great deal of recent interest in near-term applications of quantum computers. 
Variational quantum algorithms \cite{McClean_2016,yuan2019theory,cerezo2021variational,2023arXiv231202279A} are the most popular approach by a large margin, as evidenced by the quantity of papers and citations thereof studying their performance on a wide range of practical and theoretical problems.
Recent studies show that they can potentially outperform existing approaches in a number of important tasks, such as combinatorial optimization~\cite{Egger2021warmstartingquantum}, quantum chemistry \cite{bauer2020quantum,lin2022heisenberg},
particle physics~\cite{humble2022snowmass,blance2021quantum} and many others.

In variational quantum algorithms (VQAs), a parametric quantum circuit is run and ``evaluated", or endowed with measurement that collapses the state to some noisy eigenvalue of an observable operator. This measurement is in turn used to classically guide the estimate of the parameters, which are adjusted according to the evaluation, and then used to prepare the (updated) parametric quantum circuit again. Thus, it is a scheme to simultaneously use classical and quantum computing in an alternating manner, developed as a means of utilizing the potential speedup properties of quantum computing despite the inability to obtain useful output from circuits beyond a relatively shallow depth in the so-called near-term devices \cite{cerezo2021variational}. 

The classical computation component can be viewed as an unconstrained optimization problem, wherein a function of the so-called variational parameters is minimized with respect to these variables. This objective function has certain specific properties, which we shall discuss in detail in the following. However, one important and immediate point is that the evaluations of the function (and its gradients and any higher-order information) are ultimately noisy, meaning that one evaluation with a certain set of parameters may yield a different value than another evaluation with the same set of variational parameters. 

We can distinguish two classes of algorithms used to compute the steps of the classical optimization procedure in VQA.
One popular option \cite{HarrowNapp}, inspired by automatic differentiation, evaluates closed-form expressions for the gradient \cite{ekert2002direct,fischer2022ancilla}.
Alternatively, one can use ``zero order" or ``derivative free optimization" methods, which only make use of function evaluations rather than seeking to compute gradients. 
A popular choice is the \emph{Simultaneous Perturbation Stochastic Approximation} (SPSA) algorithm, a classic method introduced by Spall et al.~\cite{spall1992multivariate} to solve stochastic root-finding problems using noisy function evaluations to approximate a gradient. There it was found that a careful evaluation of the problem function at distinct points yielded a more reliable estimate of the gradient than finite difference methods, relative to the number of evaluations required, and it became a standard tool of zeroth-order stochastic approximation. {Despite being a comparatively unused method in classical settings, known to perform substandard compared to other methods, in the case of VQA, practical numerical experience has yielded SPSA to be the most popular zero order optimization method, see, e.g.~~\cite{Bonet2023}. In studying the convergence properties under biased noise, we hope to elucidate on the potential unique value of SPSA to this class of problems.

More commonly in the general optimization literature, an approach described (ambiguously, since SPSA also used two evaluations) as the ``two-point function approximation" rule is the most extensively analyzed. As such, as a point of both theoretical and numerical comparison, we consider the application of both SPSA vis-a-vis the two point function approximation rule for optimizing VQA objectives.

The third popular procedure, which is unique to the functional form of parametrized quantum circuits, is called the ``parameter shift rule". It can be understood as an intermediate step between zero'th- and first-order methods. In particular, although by strict definition the procedure is zero-order in that it estimates a gradient using function evaluations, it is based on the original observation that for certain classes of gates and choices of perturbation, the approximation is exact. Based on select choices of the direction and magnitude of the perturbation around the base point, the parameter shift rule can be calculated to produce the exact gradient for certain circuits~\cite{Mitarai_2018, schuld2019evaluating}. The class of circuits for this applicability, however, is limited, especially when one considers the entire pipeline through an ultimate evaluation of an observable. Recent work has expanded the scope of procedures that appear equivalent to analytical expressions as well as the estimation accuracy of those that are not~\cite{Banchi_2021}. As a consequence, given that the parameter shift rule computes analytically exact in certain contexts, and approximations in many others, it is best considered as a method in between zero- and first-order, and is best analyzed in a manner attuned to the problem at hand. 

{A recent numerical comparison of these solvers is given~\cite{Bonet2023}. Extensive testing and hyperparameter tuning was used to obtain the best settings for computational performance on a set of representative problems. SPSA is compared with a custom solver and a few classic second-order optimization methods. There was no clear winner, with different methods being competitive for different problem types and parameters.}

In regards to claiming analytical results on these methods, while there are plenty of studies on the iteration complexity of classical optimization algorithms, their applicability to VQA is suspect due to the presence and nature of the errors arising in near-term circuits. That is: beyond the potential stochastic nature of a quantum circuit evaluation (i.e., for a number of problems, the theoretical output can be one of several eigenstates of an observable operator), the gates in near-term devices are far from implemented with perfect fidelity. The errors or unintended operations in the quantum state~\cite{benenti2019principles} are difficult to model. Common models consider rotation (\emph{dephasing}) and dampening of an amplitude (\emph{dissipation}). While there is active research in developing hardware and error-correcting codes to mitigate the scale and impact of these errors, for the time being, it is clear that such systematic errors in the evaluation of quantum circuits shall remain endemic for some time. 

Whereas there is thorough theoretical understanding of the performance of many variants of SGD and zero-order optimization methods, the analysis is typically inspired by statistical and machine learning applications. In this case, the function being minimized is a loss function in the form of a large sum of discrepancies from actual to predicted estimates across a large sum of data samples. The stochasticity comes from the necessity of sampling one or multiple (minibatch) samples in evaluation of the function or gradient during the course of the optimization procedure, due to the frequently large sample size and unrealizability of storing an entire function or gradient in memory. A key feature of this class of problems is that a function or gradient evaluation can be performed in such a way that it is unbiased, i.e., while it is noisy, the statistical expectation of the evaluation is equal to a desired deterministic quantity.

By contrast, as we shall see below, the errors that appear in the evaluation of quantum circuits in VQA tend to exhibit ``bias". To formalize this notion, consider an optimization problem of the form,
\begin{equation}\label{eq:optprob}
    \min\limits_{\bs \theta} \, f(\bs\theta):= \mathbb{E}_{\xi\sim \Xi(\bs\theta)}[F(\bs \theta,\xi)],
\end{equation}
where $\bs \theta \in \mathbb{R}^p$ are the deterministic VQA parameters and $\xi \in \mathbb{R}$ is the stochastic sample. We consider that $\Xi$ is the underlying distribution on which we want to minimize this first moment quantity. However, in practice, we may not have access to unbiased samples of the quantity in~\eqref{eq:optprob}, but instead each physical evaluation of $F(\bs\theta,\cdot)$ satisfies a difference distribution, formally,
\begin{equation}\label{eq:noisyeval}
    F(\bs\theta,\cdot)\sim F(\bs\theta,\xi),\,\xi\sim \chi (\bs\theta),
\end{equation}
where in general the state-dependent distributions over $\xi$ are not the same nor do they have the same moments, i.e.,
\[
\mathbb{E}_{\xi\sim\Xi(\bs\theta)} [F(\bs\theta,\xi)]\neq \mathbb{E}_{\xi\sim\chi(\bs\theta)} [F(\bs\theta,\xi)].
\]
Thus, in this work the desired minimization is with respect to the distribution denoted by $\Xi(\bs\theta)$, but only estimates associated with the distribution $\chi(\bs\theta)$ are available. We assume that there is a reasonable amount of resemblance in the sense that $\text{supp}\left(\cup_{\bs\theta}\Xi(\bs\theta)\right)\subseteq \text{supp}\left(\cup_{\bs\theta}\chi(\bs\theta)\right)$. In the sequel, for parsimony we will drop the state-dependence of the distributions, i.e., write $\Xi(\bs\theta)$ as just $\Xi$, observing that the notation limitation does not hinder the clarity of the presented results.

In standard Stochastic Approximation (SA) Theory, one iteratively computes,
\begin{equation}\label{eq:biasedsg}
\bs g_k \approx \nabla F(\bs \theta_k,\xi_k) = \nabla f(\bs \theta_k)+\bs n(\bs \theta_k,\xi_k)
\end{equation}
i.e., forming some possibly approximate estimate of a sample of $\nabla F$ from a distribution, written formally as the exact gradient of the (hidden) function $f$ plus an approximation error denoted as $ \bs n(\bs \theta_k,\xi_k)$. This error represents the deviation from the exact gradient and can be due to systemic stochastic noise and algorithmic inexact approximation of the gradient. We consider the general case of state-dependent noise, where $\bs \theta_k$ is the state.

SA then involves iterating,
\begin{equation}\label{eq:saiter}
    \bs \theta_{k+1}=\bs \theta_k-\alpha_k \bs g_k
\end{equation}
where $\alpha_k$ is a \emph{gain} sequence satisfying
\begin{equation}\label{eq:sagain}
    \sum\limits_{k=0}^\infty \alpha_k=\infty,\qquad\sum\limits_{k=0}^\infty (\alpha_k)^2<\infty.
\end{equation}

Standard SGD methods assume that such an evaluation satisfies $\mathbb{E}_\xi[\bs n(\bs \theta_k,\xi)]=0$, that is, they assume the noise is unbiased. The convergence properties of SA under asymptotically vanishing biased noise are considered throughout many classic texts~\cite{borkar2009stochastic}. Several recent works do consider systemic bias that remains throughout the iterative process, i.e., $\mathbb{E}_k [\bs n(\bs \theta_k,\xi_k)]\neq 0,\,\forall k$, and, even stronger, a stationary $\limsup_k\left\|\mathbb{E}_k [\bs n(\bs \theta_k,\xi_k)]\right\|= b\neq 0,\,\forall k$, has been considered in works such as~\cite{tadic2017asymptotic,pmlr-v99-karimi19a} and most recently~\cite{ajalloeian2020convergence}. Of course, with systematic bias, one cannot hope to reach, even asymptotically, even a local minimizer. Instead, guarantees an iteration is in some neighborhood of a stationary point.

Analyses of biased zeroth-order optimization methods in the literature are scant, to the best of the authors' knowledge.
In~\cite{chen2015stochastic} this problem is claimed to be uninteresting since one can simply minimize $f(\bs \theta_k)-\mathbb{E}_\xi[n(\bs \theta_k,\xi_k)]$. However, we believe this is misguided, \textcolor{black}{at least in the case of VQA,} as one  does not have a closed-form expression for the noise, and its estimation may add significant computational expense. That said, we are also unaware of any other applications, other than VQA, wherein there is a known (noisy but with known distribution) function to minimize, with the actual observations subject to biased noise with unknown distribution and even moments thereof. As such, the novelty of the work is appropriate for the highlight, and to the knowledge of the authors, sole relevant application of VQA.

The presence of biased noise indicates that one cannot hope to achieve, even asymptotically, a stationary point of $\min_{\bs \theta} \mathbb{E}_{\xi}[F(\bs \theta,\xi)]$. Instead, there are notions of approximation and robustness that become relevant: given a bound for the noise, for a zero-order algorithm, 1) how close can we guarantee that iterates get to stationary points and 2) how can we guarantee an estimate from the algorithm that is the best choice under consideration of the worst case outcome of the noise?
We should also mention that there exists one work studying the asymptotic guarantees of SPSA under bounded noise~\cite{granichin2014simultaneous} with subsequent application to Gaussian mixture models~\cite{boiarov2017simultaneous}. In this work the authors established a quantitative bound on the asymptotic bias associated with the mean squared error of the gradient for SPSA and other zero-order optimization procedures.

\paragraph{Our contribution}
In this paper, we are interested in studying the properties of classical optimization algorithms as used for variational quantum algorithms (VQA)
under non-trivial noise in the evaluation of a quantum circuit.
The bias in zeroth-order methods would consist of \emph{function evaluations} yielding an error, i.e.,
\begin{equation}\label{eq:feval}
f \sim F(\bs \theta,\xi) = f(\bs \theta)+ n(\bs \theta,\xi),\quad \xi\in\chi 
\end{equation}
with $\mathbb{E}_{\xi\sim\chi}[ n(\bs \theta,\xi)]\neq 0$. 

The rest of this paper proceeds as follows. 
Section~\ref{sec:back} presents the background on VQA and applicable optimization methods. We also present our argument in regards to the necessity of considering systematic bias in the analysis of the convergence of optimization algorithms for VQA and review the literature on convergence guarantees in the presence of bias. 
Section~\ref{sec:newconv} presents our novel contribution in providing nonasymptotic guarantees of neighborhood convergence. 
Section~\ref{sec:num} presents the results of numerical performance on a set of quantum optimization problems. 
Finally, we discuss the implications of our work and insights for future research in Section~\ref{sec:conc}. Appendix~\ref{app:proofs} provides the proofs of our main results. 
In summary, this work represents a more faithful-to-reality convergence analysis for quantum variational solvers, providing first-principles insight into the intuitive practice of eschewing attempts to evaluate a gradient and consider a zero-order method instead.

\begin{table}[]
\centering
\resizebox{\columnwidth}{!}{
\begin{tabular}{ccc}  \Xhline{5\arrayrulewidth} 
Quantity & Description           & Reference   \\ \hline
$F,C$      & cost functions         & Eq. \eqref{eq:optprob} and Eq. \eqref{eq:cost}        \\
$\bs \theta$    & variational parameter in $\mathbb{R}^p$ & Eq.  \eqref{eq:optprob} and Sec. \ref{sec:vqas}      \\
$\bs \theta^*$    & optimal variational parameter & Sec. \ref{sec:vqas}      \\
$\xi$       & stochastic sample     & Eq. \eqref{eq:optprob}    \\
$\Xi(\bs \theta)$       & ideal (noise) distribution &   Eq. \eqref{eq:optprob}           \\
$\chi(\bs \theta)$       & physical (noise) distribution &   Eq. \eqref{eq:noisyeval}          \\
$n(\bs \theta,\xi)$ & random variable describing mean zero and noisy bias & Eq. \eqref{eq:feval}\\
$\bs g_k$       & (stochastic) gradient at step $k$ &  Eq. \eqref{eq:biasedsg}           \\
$\alpha_k$       & gain in the negative gradient direction $k$ &  Eq. \eqref{eq:saiter} and Eq. \eqref{eq:sagain}          \\ 
$\ket{\psi_0},\, \rho_0$       &    initial quantum state  & Eq.\eqref{eq:cost} and Sec. \ref{sec:vqas}\\
$\ket{\psi(\bs \theta)},\, \rho(\bs \theta)$       &    parametrized quantum state  & Eq. \eqref{eq:cost} and Sec. \ref{sec:vqas} \\
$U(\bs \theta)$       &   parametrized variational ansatz  & Eq. \eqref{eq:U} and Eq. \eqref{eq:U} \\
$O, \, B$       &   observable that defines the cost function  & Eq. \eqref{eq:cost2}, Sec. \ref{sec:vqas2} and Fig. \ref{fig:PSR} resp. \\
$\mathcal{E}(\rho)$ & quantum noise channel & Eq.\eqref{eq:Kraus} and Sec. \ref{sec:bias} \\
$\mathrm{Herm}(\mathbb{C}^n)$ & space of Hermitian operators over $\mathbb{C}^n$ & Before Eq. \eqref{eq:cost} \\
$\mathrm{Hilb}(\mathbb{C}^n)$ & Hilbert space over $\mathbb{C}^n$ & Eq.\eqref{eq:Kraus} and Sec. \ref{sec:bias} \\ 
$\bs{\Delta}_k$ & $k$-th simultaneous perturbation vector & Sec. \ref{sec:derfree} and Eq. \eqref{eq:spsa}\\
$\mu$ & smoothing parameter $>0$ & Eq. \eqref{eq:smoothedfunc} \\
$f_\mu$ &   Gaussian smoothing function & Eq. \eqref{eq:smoothedfunc} \\ 
$N$ & number of SA iterations & Sec. \ref{sec:reviewguarantees} \\
$L$ &  modulus of Lipschitz continuity & Eqs. \eqref{eq:litdiff} \\
$\gamma$ & step size parameter & Eq. \eqref{eq:step} \\
\hdashline[0.5pt/5pt]
$C_n$       & constant $>0$    & Assumption \ref{ass:31}      \\
$c,\, c_k$     &  set of scalars (errors) in the gradient approximations   & Eq. \eqref{eq:spsa} and. Eq. \eqref{eq:twopointNesterov}         \\
${b},\, {b}_k$     &       upper bound of bias          &    Eq. \eqref{eq:biasas}     \\ \hdashline[0.5pt/5pt]
$p,L_V$ & number of layers in the VQA & Eq. \eqref{eq:U} and Eq. \eqref{eq:bLv} resp. \\ \hline
\end{tabular}  
}
\caption{Summary of frequent notation used throughout the paper.}
\end{table}


\section{Background}\label{sec:back}

\subsection{Variational Quantum Algorithms}\label{sec:vqas}

\textcolor{black}{The advent of \emph{variational quantum algorithms} (VQAs) \cite{Peruzzo_2014, Farhi2014, Wecker2015,McClean_2016, yuan2019theory,cerezo2021variational} came as a practical response to the challenges of \emph{noisy intermediate scale quantum} (NISQ) devices \cite{Preskill2018quantumcomputingin}. NISQ devices suffer from coherent and incoherent noise, short coherence times, and limitations on qubit connectivity. VQA attempt to achieve quantum advantage under those practical limitations.}
Of course, VQAs also have many disadvantages. By incorporating a hybrid classical-quantum framework, a significant overhead appears in preparing the initial quantum state to feed as input to the circuit, and significant computing potential is lost by frequent wave collapse by observation. 
On the other hand, one of their attractive features is that they provide a general framework that fits a range of problems, from quantum simulation \cite{Peruzzo_2014} to solving differential equations \cite{Kyriienko2021} and to designing quantum autoencoders \cite{Shaikh2022}, to name a few. 
No matter what class of problem one is interested in solving using VQAs, the basic structure of these algorithms is invariant, despite the fact that certain applications might benefit from different algorithmic subroutines. 

Specifically, VQAs are built around a quantum-classical loop. First, assuming that one contains data that can be encoded in a parametrized quantum circuit, one defines a cost function $C(\bs \theta)$,  $\bs \theta \in \mathbb{R}^p$, whose true minimum, obtained by the optimizer $\bs \theta^* \coloneqq \argmin_{\bs \theta}C(\bs \theta)$, corresponds to the solution of the problem. Usually, $C(\bs \theta)$ is given as the expectation value of some operator $O \in \Herm(\mathbb{C}^d)$,
\begin{equation}
\begin{aligned}\label{eq:cost}
    C(\bs \theta) &= \braket{\psi(\bs \theta) | O \psi(\bs \theta)}\\
                  &= \braket{\psi_0|U^\dagger(\bs \theta)O U(\bs \theta)|\psi_0},
\end{aligned} 
\end{equation}
where $\ket{\psi_0} \in \mathbb{C}^d$ denotes the initially prepared state and $U \in U(d)$. Furthermore, it is assumed that the minimum of $C(\bs \theta)$ is not efficiently computable using only a classical computer. (This requires, among others, the use of gates outside of the Clifford group and limits the applicability of common quantum error correction mechanisms.) Therefore, the quantum part of the information loop mentioned above amounts to:
\begin{enumerate}[font=\itshape]
    \item Encoding (loading) the data into a parametrized quantum circuit.
    \item Initializing the parameters $\bs \theta \in \mathbb{R}^p$ of the quantum circuit.
    \item Measuring the parametrized output of the circuit $C(\bs \theta)$ or its gradient.
\end{enumerate}
Once an iteration of the quantum part of the loop has been performed, a classical optimizer is used to minimize the measured cost function with the aim of updating the parameters $\bs \theta \to \bs \theta'$, $ \bs \theta' \in \mathbb{R}^p$ such that $C(\bs \theta') \leq C(\bs \theta)$.
In Eq. \eqref{eq:cost}, the unitary $U = U(\bs \theta)$ is referred to as the \emph{variational ansatz}, which explicitly contains precise information about the variational parameters $\bs \theta \in \R^p$ and corresponds to the actual quantity that is \emph{classically trainable}.
In the literature, one encounters problem-inspired or problem-specific
ans{\" a}tze, suited for certain problems, as well as generic ansatz architectures that are problem-agnostic \cite{cerezo2021variational}. Generically, in the simplest form, $U(\bs \theta)$ is decomposed as the sequential application of unitaries parameterized by the parameters: 
\begin{align}\label{eq:U}
    U(\bs \theta) = U_p(\theta_p) \ldots U_1(\theta_1),
\end{align}
where each parameterized unitary $U_i\in U(d)$, where $d$ is the level of the quantum system, takes the form
\begin{align}\label{eq:Ui}
    U_i(\theta_i) \coloneqq W_i\exp(-\im \theta_i H_i),
\end{align}
for $H_i \in \Herm(\C^d)$, $W_i \in U(d)$ and $\theta_i \in \R$, for any $i \in [p]$. The unitaries $W_i$ are not subject to variation by the classical optimizer, while Hermitian operators $H_i$, the \emph{Hamiltonians}, are optimally chosen according to the problem at hand. For example, for a certain type of VQAs, the so-called \emph{quantum approximate optimization algorithm} (QAOA) \cite{Farhi2014}, as well as generalization thereof such as the quantum alternating operator ansatz \cite{Hadfield_2019}, best suited for combinatorial optimization problems, one considers a repeating sequence of two types of Hamiltonians where $H_{2i}$ corresponds to a sum of the Pauli $\sigma_x$ matrices and $H_{2i+1}$ corresponds to the Ising model Hamiltonian. Generally, there are many different types of ansatze, including ansatze for thermal states \cite{Verdon2019}.

An important aspect of VQAs is the presence of systemic errors (bias). As a motivating example, consider the depolarizing channel of a density matrix
\begin{align}
    \mathcal{E}(\rho) =p\frac{\mathds{1}}{d}+(1-p)\rho,
\end{align}
$\mathcal{E} \in {\rm End}({\rm Herm}(\mathbb{C}^d))$, which replaces each qubit (or the whole system) in state $\rho$ with a totally mixed state with some probability $p > 0$. Consider that subsequently, an observable $O$ is applied to compute the energy, entropy or entanglement, for instance. It is clear that if in all circuit runs the state $\rho$ is input into the operator $\mathcal{E}$ any evaluation will not have an expectation $\braket{O} \coloneqq {\rm Tr}(\rho O)$, rather $\braket{\tilde O} \coloneqq {\rm Tr}(\mathcal{E}(\rho) O)$, and no matter how many trials are performed, the sample average will not equal this desired value, that is, 
\begin{align}
    \braket{O} \neq \mathbb{E}(\tilde{O})\coloneqq \frac{1}{N}\sum\limits_{i=1}^N f_i,\quad f_i\sim \mathbb{E}_i(\tilde{O}),
\end{align}
where $f_i$ corresponds to a measurement. Said differently, the estimator $\tilde{O}$ of some observable $O$ is biased in the sense $\mathbb{E}[O-\tilde{O}] \neq 0$. It is known that, for example, on IBM machines, there exists an asymmetric read-out error \cite{Nachman2020}.
See Appendix \ref{sec:bias} for a numerical illustration. 

Thus, it is clear that the noise arising from a classical measurement applied to a quantum circuit, as is commonly done in, e.g., VQA, is not unbiased, as far as the expectation of the (parameterized) estimator is concerned, and thus any analysis of algorithms with function evaluations of this sort must incorporate this for faithful modeling of the problem. 

\subsection{Parameter Shift Rule as a Gradient Estimate}\label{sec:vqas2}

Let us consider a VQA unitary in the form of Eq. \eqref{eq:U}. The cost function to be optimized takes the form
\begin{align}\label{eq:cost2}
    C(\bs\theta) = \braket{\psi_0 | U^\dagger(\bs\theta) O U(\bs\theta) |\psi_0},
\end{align}
with $U(\bs\theta)$ as in Eq. \eqref{eq:U}. Typically, the unitaries executed in the context of VQAs are of the form of Eq. \eqref{eq:Ui} where $H_i$ are rotation-like \textcolor{black}{operators such} that $H_i^2 =\mathds{1}$. This is the case for all operators acting on single qubits; however, $H_i$ can be any generic operator, and a viable strategy is to decompose them into Pauli strings \cite{Crooks2019}, for example.  In what follows we introduce the backbone of what is generically referred to as the \emph{parameter shift rule} (PSR).

The gradient of Eq. \eqref{eq:cost2} is of particular interest in this article. Precisely for the reasons mentioned in the previous paragraph, it suffices to analyze the case where the set of Hermitian operators $\{ H_i\}$ are rotation-like operators, that is, given in terms of Pauli strings. Recall that any gate can be represented as a unitary operator 
\begin{align}
   U_i(\theta_i) = e^{-i \theta_i H_i} = \cos(\theta_i)\mathds{1}- i \sin(\theta_i)H_i.  
\end{align}
Any operator of the form $O(\theta_i)\coloneqq U_i^\dagger(\theta_i)OU_i(\theta_i)$, as in Eq. \eqref{eq:cost2}, can be written as 
\begin{align}
    A + B \cos(\theta_i)+ C \sin(\theta_i),
\end{align}
with $A,B,C$ operators independent of the variational parameter $\theta_i$ and only dependent on the unitary $H_i$. Additionally, using standard trigonometric rules, 
\begin{equation}
\begin{aligned}
\frac{d \cos (x)}{d x} &=\frac{\cos (x+s)-\cos (x-s)}{2 \sin (s)} \\
\frac{d \sin (x)}{d x} &=\frac{\sin (x+s)-\sin (x-s)}{2 \sin (s)},
\end{aligned}
\end{equation}
for $s\neq n\pi$ where $n \in \mathbb{Z}$, we can easily compute the derivative of $O(\theta_i)$ as
\begin{align}\label{eq:derivative}
    \frac{dO(\theta_i)}{d\theta_i}= \frac{O(\theta_i+s)-O(\theta_i-s)}{2\sin(s)}.
\end{align}
Therefore, computing the derivative with respect to one of the variational parameters amounts to computing the actual observable twice in symmetrically shifted points. The parameter-shift rule, as defined in \cite{Mitarai_2018,Li_2017,schuld2019evaluating} corresponds to $s=\pi/2$. Note that for any rotation-like operator $H_i$, the derivative \eqref{eq:derivative} is exact (except for integer multiples of $\pi)$. The gradient can be easily evaluated using the same exact rule by redefining the expectation value. That is, to compute the derivative with respect to $\theta_1$, we redefine $\tilde{O}_1\coloneqq U_2^\dagger(\theta_2)\ldots U_L^\dagger(\theta_L)O U_L(\theta_L) \ldots U_2(\theta_2)$ and $\tilde{O}_1(\theta_1) \coloneqq U_1^\dagger(\theta_1)\tilde{O}_1U_1(\theta_1)$, while for $1< j < L$ we define $\tilde{O}_j\coloneqq U_{j+1}^\dagger(\theta_{j+1})\ldots U_L^\dagger(\theta_L)O U_L(\theta_L) \ldots U_{j+1}(\theta_{j+1})$ and $\ket{\tilde{\psi}_j} = U_{j-1}(\theta_{j-1})\ket{\psi_0}$ such that the derivative with respect to $\theta_j$ is taken for $\tilde{O}_j(\theta_j) = U_j^\dagger(\theta_j) \tilde{O}_j U_j(\theta_j)$. The gradient is then given as
\begin{align}\label{eq:nabla}
    \nabla_{\bs \theta} O(\bs \theta) = \sum_{i=1}^L \frac{d\tilde{O}_i(\theta_i)}{d\theta_i}.
\end{align}
Higher-order derivatives can also be derived by iterations of the same rule \cite{Mari_2021}. 

For quantum variational algorithms, the parameter shift rule has emerged as a standard means by which to analytically compute the gradients of a parametrized circuit (see, e.g.,~\cite{Banchi_2021}). The rule became popular because for noisy near-term quantum computers, it is quite difficult to implement accurate finite difference estimates since the two required shifted evaluations, of difference of order $O(10^{-7})$, cannot be effectively distinguished from the noise. However, the parameter-shift rules are not only exact, unlike the finite-difference methods, but also involve a relatively large shift away from the noise threshold. However, the expectation values required to compute Eq. \eqref{eq:nabla} are still subject to the limitation that only a finite number of shots will be available to construct empirical estimates, and since the ultimate objective evaluation is an expectation, a sampling error will appear. In addition, though, as we shall see in the discussion of bias in Sec. \ref{sec:bias}, the evaluations will not be of the ideal desired sampled distribution but, as a result of systemic noise, biased. We shall study the effect of this on the convergence properties in the following section. 

\begin{table}[h!]
\centering
\resizebox{\columnwidth}{!}{
\begin{tabular}{l|ccccc}  \Xhline{5\arrayrulewidth} 
\multirow{2}{*}{ Description }          & \multirow{2}{*}{ Year}     &\multirow{2}{*}{ Name }         &  \multirow{2}{*}{Reference}    &  \multirow{2}{*}{Context} &  \multirow{2}{*}{ Exactness}   \\ &&&&&\\\hline
\multicolumn{6}{c}{\multirow{2}{*}{\bf A Small Sample of Classical Precursors}} \\ 
\multicolumn{6}{c}{}  \\ \hline
\multirow{2}{*}{\begin{tabular}[c]{@{}l@{}}Random Optimization  \\ method \end{tabular}}   &\multirow{2}{*}{1977} &\multirow{2}{*}{ Dorea }  &\multirow{2}{*}{\cite{Dorea1983}} & \multirow{2}{*}{OPT}  &\multirow{2}{*}{ No, stochastic in nature } \\&&&&&\\ \hdashline[0.5pt/5pt]
\multirow{2}{*}{\begin{tabular}[c]{@{}l@{}}Automatic  \\ Differentiation\end{tabular}}   &\multirow{2}{*}{1983} &\multirow{2}{*}{ Wengert }  &\multirow{2}{*}{\cite{Wengert1964ASA}} & \multirow{2}{*}{OPT}  &\multirow{2}{*}{ Yes } \\&&&&&\\ \hdashline[0.5pt/5pt]
\multirow{2}{*}{\begin{tabular}[c]{@{}l@{}}SPSA\end{tabular}}   &\multirow{2}{*}{1992} &\multirow{2}{*}{ Spall }  &\multirow{2}{*}{\cite{spall1992multivariate}} & \multirow{2}{*}{OPT}  &\multirow{2}{*}{ No, stochastic in nature } \\&&&&&\\ \hdashline[0.5pt/5pt]
\multirow{2}{*}{\begin{tabular}[c]{@{}l@{}}Nesterov Smoothing \end{tabular}}   &\multirow{2}{*}{2005} &\multirow{2}{*}{ Nesterov}  &\multirow{2}{*}{\cite{nesterov2005smooth}} & \multirow{2}{*}{OPT}  &\multirow{2}{*}{ Yes  } \\&&&&&\\ \hdashline[0.5pt/5pt]
\multirow{2}{*}{\begin{tabular}[c]{@{}l@{}}Two-point \\ Approximation\end{tabular}}   &\multirow{2}{*}{2015} &\multirow{2}{*}{ Duchi \emph{et. al.} }  &\multirow{2}{*}{\cite{duchi2015optimal}} & \multirow{2}{*}{OPT}  &\multirow{2}{*}{ No, stochastic in nature } \\&&&&&\\ \hdashline[0.5pt/5pt]
\multirow{2}{*}{\begin{tabular}[c]{@{}l@{}}Gaussian Smoothing  \end{tabular}}   &\multirow{2}{*}{2017} &\multirow{2}{*}{ Nesterov \& Spokoiny}  &\multirow{2}{*}{\cite{Nesterov2015}} & \multirow{2}{*}{OPT}  &\multirow{2}{*}{ No, stochastic in nature  } \\&&&&&\\ \hline

\multicolumn{6}{c}{\multirow{2}{*}{\bf Parameter Shift Rule (PSR)}} \\
\multicolumn{6}{c}{}  \\ \hline
\multirow{2}{*}{\begin{tabular}[c]{@{}l@{}}Gradient of expectation\\ values of observables\end{tabular}}   &\multirow{2}{*}{2005} &\multirow{2}{*}{ Khaneja \emph{et. al.} }  &\multirow{2}{*}{\cite{Khaneja2005}} & \multirow{2}{*}{QOC}  &\multirow{2}{*}{ First-order approximation } \\&&&&&\\ \hdashline[0.5pt/5pt]
\multirow{2}{*}{\begin{tabular}[c]{@{}l@{}}Gradient evalutation\\with two parameters\end{tabular}} & \multirow{2}{*}{2017} &  \multirow{2}{*}{ Li \emph{et. al.} }& \multirow{2}{*}{\cite{Li_2017}} & \multirow{2}{*}{QOC} & \multirow{2}{*}{ First-order approximation } \\&&&&&\\ \hdashline
\multirow{2}{*}{\begin{tabular}[l]{@{}l@{}}First paper on PSR\\with two parameters\end{tabular}} & \multirow{2}{*}{2018} & \multirow{2}{*}{ Mitarai \emph{et. al.} }&\multirow{2}{*}{\cite{Mitarai_2018}} &  \multirow{2}{*}{QML} &  \multirow{2}{*}{Yes} \\&&&&&\\ \hdashline[0.5pt/5pt]
\multirow{2}{*}{\begin{tabular}[l]{@{}l@{}}Extension of \cite{Mitarai_2018} for\\ more generic circuits\end{tabular}} & \multirow{2}{*}{2018} & \multirow{2}{*}{Vidal \& Theis}&\multirow{2}{*}{\cite{Theis2018}} &  \multirow{2}{*}{QML} &  \multirow{2}{*}{Yes} \\&&&&&\\ \hdashline[0.5pt/5pt]
\multirow{2}{*}{\begin{tabular}[c]{@{}l@{}}Follow up of \cite{Mitarai_2018}\\and slight generalization\end{tabular}} & \multirow{2}{*}{2019} & \multirow{2}{*}{ Schuld \emph{et. al.} }&\multirow{2}{*}{\cite{schuld2019evaluating}} & \multirow{2}{*}{QML} &\multirow{2}{*}{\begin{tabular}[c]{@{}c@{}}Yes for Pauli generators\\and no for  non-Pauli generators\end{tabular}} \\&&&&&\\ \hdashline[0.5pt/5pt]
\multirow{2}{*}{\begin{tabular}[c]{@{}l@{}}Follow up of \cite{schuld2019evaluating} allowing\\more generic gates\end{tabular}} & \multirow{2}{*}{2019} & \multirow{2}{*}{ Crooks }&\multirow{2}{*}{\cite{Crooks2019}} & \multirow{2}{*}{QML} &\multirow{2}{*}{\begin{tabular}[c]{@{}c@{}}Yes due to the Pauli decomposition\\of  non-Pauli generators\end{tabular}}\\&&&&&\\  \hdashline
\multirow{2}{*}{\begin{tabular}[l]{@{}l@{}}Extension of PSR for generators\\ with general eigenspectrum\end{tabular}} & \multirow{2}{*}{2021} & \multirow{2}{*}{ Mari \emph{et. al.} }&\multirow{2}{*}{\cite{Mari_2021}} & \multirow{2}{*}{QML}  & \multirow{2}{*}{Yes, iterative in nature}\\&&&&&\\\hdashline[0.5pt/5pt]
\multirow{2}{*}{\begin{tabular}[l]{@{}l@{}}Generalization of the PSR\\for higher-order derivatives\end{tabular}} & \multirow{2}{*}{2021} & \multirow{2}{*}{ Izmaylov \emph{et. al.} }&\multirow{2}{*}{\cite{izmaylov2021analytic}} & \multirow{2}{*}{QML}  & \multirow{2}{*}{Yes, iterative in nature}\\&&&&&\\\hdashline[0.5pt/5pt]
\multirow{2}{*}{\begin{tabular}[l]{@{}l@{}}Stochastic generalization of\\the original PSR\end{tabular}} & \multirow{2}{*}{2021} & \multirow{2}{*}{ Banchi \& Crooks }&\multirow{2}{*}{\cite{Banchi_2021}} & \multirow{2}{*}{QML}  & \multirow{2}{*}{No, stochastic in nature}\\&&&&&\\ \hdashline[0.5pt/5pt]
\multirow{2}{*}{\begin{tabular}[l]{@{}l@{}}Stochastic generalization for gates\\with more than two eigenvalues\end{tabular}} & \multirow{2}{*}{2021} & \multirow{2}{*}{ Wierichs \emph{et. al.} }&\multirow{2}{*}{\cite{Wierichs2021}} & \multirow{2}{*}{QML}  & \multirow{2}{*}{No, stochastic in nature}\\&&&&&\\ \hdashline[0.5pt/5pt]
\multirow{2}{*}{\begin{tabular}[l]{@{}l@{}}Fourier analysis for feasibility \\ and characterization of PSRs\end{tabular}} & \multirow{2}{*}{2021} & \multirow{2}{*}{ Theis }&\multirow{2}{*}{\cite{Theis2021}} & \multirow{2}{*}{OPT}  & \multirow{2}{*}{N/A}\\&&&&&\\ \hdashline[0.5pt/5pt]
\multirow{2}{*}{\begin{tabular}[l]{@{}l@{}}Four-term PSR for \\ fermionic Hamiltonians\end{tabular}} & \multirow{2}{*}{2021} & \multirow{2}{*}{ Kottmann \emph{et. al.} }&\multirow{2}{*}{\cite{Kottmann2021}} & \multirow{2}{*}{OPT}  & \multirow{2}{*}{Yes}\\&&&&&\\ \hdashline
\multirow{2}{*}{\begin{tabular}[l]{@{}l@{}}Quantum gradient computation\\with classical resources\end{tabular}} & \multirow{2}{*}{2022} & \multirow{2}{*}{ Koczor \&  Benjamin}&\multirow{2}{*}{\cite{koczor2022quantum}} & \multirow{2}{*}{QML}  & \multirow{2}{*}{No, classical computation}\\&&&&&\\ \hline
\end{tabular}
}
\caption{The evolution of parameter shift rule (PSR) approaches from the ancient (top) to the most recent (bottom). Acronyms used in context include QOC: Quantum Optimal Control, QML: Quantum Machine Learning, OPT: Optimization.}
\label{t:table1}
\end{table}

Several extensions to the ``original" parameter shift rules have been discovered, including stochastic versions, all of which we summarize in Table \ref{t:table1}. 
Let us review some of the landmark literature regarding the parameter shift rule. Parameter shift rules originate from the work on quantum optimal control \cite{Khaneja2005, Li_2017}. The ``original" parameter shift rules, as they are currently understood in the context of VQAs, were first derived in \cite{Mitarai_2018,Theis2018,schuld2019evaluating}. In Ref. \cite{Crooks2019}, generalized the ``original" work for operators that are not given as Pauli strings by using gate decomposition to simpler gates. Similar generalizations were studied in Ref. \cite{Mari_2021} while in Ref. \cite{izmaylov2021analytic} parameter shift rules were derived for higher-order derivatives (Hessian, etc.). Ref. \cite{Banchi_2021} produced an stochastic generalization of the parameter shift rules that was claimed to be unbiased due to the inherent randomness of the quantum measurements. Using the quantum Fourier transform, Wierichs \emph{et. al.} proposed a generalized parameter shift rule (abbreviated as PSR in Fig. \ref{fig:PSR}) that is capable of significantly reducing the number of circuit evaluations required for the computation of the gradients with respect to gate parameters \cite{Wierichs2021}. Ref. \cite{Theis2021} discussed optimization/feasibility aspects of parameter shift rules using Fourier analysis. Finally, Ref. \cite{koczor2022quantum} proposed a classical algorithm for computing quantum gradients based on a conservative number of circuit evaluations at nearby points to a reference point.

\begin{figure}[tb!]
\centering
\begin{tikzpicture}[scale=0.8, every node/.style={scale=0.8}]
\draw (0,9.5) node[draw,text width=4cm,align=center, very thin] {CLASSICAL PRECURSORS};
\node[rectangle,
    draw = black,
    text = black,
    minimum width = 18cm, 
    minimum height = 3cm,very thin] (r) at (0,7) {};
    
\draw[very thin,dashed, gray] (0,5.5) -- (0,8.5) node[right] {};

\draw (-4,7) node {Autodiff (Wengert,~1964)};
\draw (4,7.5) node[align=center,text=black] {Random optimization method (Dorea,~1983)};
\draw (4.4,6.5) node[align=center,text=black] {Gaussian~smoothing (Nesterov~\&~Spokoiny,~2017)};

\draw[] (-5,9.5) node[text width=4cm,align=center, very thin, gray] {Determenistic algorithms};
\draw[] (5,9.5) node[text width=4cm,align=center, very thin, gray] {Stochastic algorithms};
\draw [-stealth,gray](-5,9.2) -- (-5,8.6);
\draw [-stealth,gray](5,9.2) -- (5,8.6);
\draw [-stealth,gray](-5,5.2) -- (-5,4.0);
\draw [-stealth,gray](5,5.2) -- (5,4.0);
\draw (0,4.5) node[draw,text width=4cm,align=center, very thin] {PARAMETER SHIFT RULES (PSRs)};
\node[rectangle,
    draw = black,
    text = black,
    minimum width = 18cm, 
    minimum height = 8cm,very thin] (r) at (0,-0.5) {};

\draw[line width=0.01mm,dashed,gray] (0,-4.5) -- (0,3.5) node[right] {};
\draw (-4,2.5) node[text width=6cm,align=center,text=black] {Original Deterministic PSR \cite{Mitarai_2018}\\[0.5em] $\sim \braket{B(x)}-\braket{B(-x)}$};
\draw (4,2.5) node[text width=6cm,align=center,text=black] {Original Stochastic PSR \cite{Banchi_2021} \\[0.5em] $\sim  \int_0^1(\braket{B(x)}-\braket{B(-x)}) dx $}; \node[rectangle,
    fill = green,
    opacity = 0.1,
    draw = green,
    text = purple,
    minimum width = 18cm, 
    minimum height = 2.0cm] (r) at (0,2.5) {}; 

\draw (-4,0.5) node[text width=8cm,align=center,text=black] {Four-term PSR (fermionic $H$) \cite[Eq. (22)]{Kottmann2021} \\[0.5em] $\sim \sum_{a \in \{+,- \}} c_\alpha\Big(\braket{B(x_\alpha)}-\braket{B(-x_\alpha)}\Big)$};
\draw (-4,-1.5) node[text width=8cm,align=center,text=black] {General PSR (equidistant) \cite{Wierichs2021} \\[0.5em] $\sim \sum_{\mu=1}^{2R} \mathbb{E}(x_\mu)f(\mu)$};
\node[text width=8cm,align=center,text=black] at (-4,-3.2) {General PSR \cite{Wierichs2021} \\[0.5em] $\sim \sum_{\mu=1}^Ry_\mu\Big( {\mathbb{E}(x_\mu)}-\braket{\mathbb{E}(-x_\mu)} \Big) $ };  
General PSR
\node[text width=8cm,align=center,text=black] at (4,-0.5) {General Stochastic PSR \cite{Wierichs2021}\\[0.8em] $ \sim \int_{0}^{1} \sum_{\mu=1}^{R} y_{\mu}\left[\mathbb{E}_{\mu}\left(x_{0}, t\right)-\mathbb{E}_{-\mu}\left(x_{0}, t\right)\right] \mathrm{d} t$\\[0.8em]  $E_{\pm \mu}\left(x_{0}, t\right):=\langle B\rangle_{U_{F}\left(t x_{0}\right) U\left(\pm x_{\mu}\right) U_{F}\left((1-t) x_{0}\right)|\psi\rangle}$ }; 
\end{tikzpicture}
\caption{The evolution of parameter shift rule (PSR) approaches loosely based on \cite[Fig. 1]{Wierichs2021}, from the ancient (top) to the most recent (bottom). Our methods apply to the PSR approaches highlighted in green. Following \cite{Wierichs2021}, for an observable $B$, we denote by $\braket{B(x)}$, that is, the expectation value at $x\in \mathbb{R}^n$. In the formula for the equidistant PSR, $x_\mu = \tfrac{2\mu-1}{2R}\pi$ and $f(\mu) = \tfrac{(-1)^{\mu-1}}{4R\sin(x_\mu/2)}$ (see Ref. \cite[App. A]{Wierichs2021}). In the equidistant General PSR, the General PSR and the General Stochastic PSR $\mathbb{E}_\mu$ is a function of $\braket{B(x_{\mu})}$. Specifically,  
$\mathbb{E}(x) \coloneqq \left\langle\psi\left|U^{\dagger}(x) B U(x)\right| \psi\right\rangle$ while for the General Stochastic PSR, the actual definition is provided below the corresponding PSR 
(for the precise definitions see Ref. \cite{Wierichs2021}). Finally, in the above, $U_F$ denotes the stochastic ``unitary" $\exp{i(H+F)}$, with $H$ the Hermitian generator with eigenvalues in the set $\{+1,-1\}$, $F$ any other Hermitian (that does not have to commute with $G$), and $R \in \mathbb{N}$.
    }
    \label{fig:PSR}
\end{figure}

\subsection{Derivative-Free (Zero-Order) Optimization}\label{sec:derfree}

The parameter-shift approach specific to VQA, as well as other popular algorithms for minimizing VQA objectives such as SPSA, fall under the class of derivative-free, or zero-order, optimization algorithms, meaning that they only use function evaluations, rather than analytical gradient and higher derivative computations, in the iterative procedure of solving the problem. 

In the classical optimization literature \cite{conn2009introduction}, this consists of two classes of strategies. The first class is mesh-based methods that have a set of predefined directions to consider at each iteration, with the function evaluated at some number of such directions and compared to each other and the original point. Alternatively, the so-called model-based methods use function evaluation information to approximate derivatives and higher-order information, and subsequently use some algorithm similar to classical derivative-based optimization procedures, such as trust region and gradient descent. There are a number of possible techniques to approximate a gradient using function evaluations, many of which are compared in~\cite{berahas2021theoretical}. The two we consider in this paper, the two-point function approximation and SPSA, are classic in their widespread use, in particular for the problems of interest in this paper, despite, however, their overall weaker precision and performance compared to alternatives as found in this work. 

The SPSA method, introduced in~\cite{spall1992multivariate} is defined to compute a stochastic gradient approximation at $x$ by evaluating the noisy objective function at two random points. In particular, letting $\{\bs\Delta_k\}$ be an independent set of random variables with $\mathbb{E}\left[|\bs\Delta_k|_i^{-1}\right]\le D$ and $c_k$ a set of scalars with $c_k\to 0$, compute the gradient estimate with each component defined as,
\begin{equation}\label{eq:spsa}
    [\bs g]_i \sim \frac{F(\bs \theta_k+c_k\bs\Delta_k,\xi_k^{1})-F(\bs \theta_k-c_k\bs\Delta_k,\xi_k^{2})}{2c_k [\bs\Delta_k]_i}  
\end{equation}

This was compared to schemes based on simple finite difference estimations of the gradient and it was found both theoretically and numerically that SPSA could achieve similar performance in convergence speed by SA iteration, while requiring fewer function evaluations for the approximation at each iteration~\cite{chin1997comparative}.

We note the resemblance of this procedure to the notion of two-point function approximation popularized in~\cite{duchi2015optimal} inspired by ``Nesterov smoothing"~\cite{nesterov2005smooth}. See also~\cite{Nesterov2015} for a convergence analysis. In this case we evaluate,
\begin{equation}\label{eq:twopointNesterov}
    \bs g_k\sim \frac{F(\bs \theta_k+c_k\bs\Delta_k,\xi_k^{1})-F(\bs\theta_k,\xi_k^{2})}{c_k} \bs\Delta_k
\end{equation}
We shall analyze the performance of this method alongside SPSA because of its strong resemblance to relying on two function evaluations to estimate the gradient. We note that ultimately this latter method has received more attention in the literature in the past decade, and thus is more known. However, SPSA tends to be the algorithm of choice for quantum applications. With these conceptual trade-offs, we shall cover both approaches in the subsequent analysis and experiments.
A detailed overview of the available SPSA analyses, culminating in \cite{granichin2014simultaneous}, is presented in the Appendix. 

The two-point approximation~\cite{duchi2015optimal} suggests that the gradient can be approximated by~\eqref{eq:twopointNesterov}.
The technique is based on the Gaussian smoothing technique \cite{Dorea1983,Nesterov2015}, which defines a function for a smoothing parameter corresponding to the perturbation for the function evaluation, here $c_k$. Considering a generic $c>0$, we write:
\begin{equation}\label{eq:smoothedfunc}
    f_{c}(\bs\theta) = \frac{1}{(2\pi)^{\frac{n}{2}}}
    \int f(\bs\theta+c \bs \Delta)e^{-\frac{1}{2}\|\bs \Delta\|^2} d\bs \Delta = \mathbb{E}_{\bs \Delta\sim \mathcal{N}(0,I)}[f(\bs\theta+c \bs \Delta)]
\end{equation}
and satisfies some nice properties that shall be used in the sequel (see, e.g.,~\cite{nesterov2005smooth,ghadimi2013stochastic}) for $L$-Lipschitz smooth $f(\bs\theta)$,
\begin{subequations}\label{eq:litdiff}
\begin{align}
\nabla f_{c}(\bs\theta) &= \frac{1}{(2\pi)^{\frac{n}{2}}}\int \frac{f(\bs\theta+c \bs \Delta)-f(\bs\theta)}{c} \bs\Delta e^{-\frac{1}{2}\|\bs \Delta\|^2} d\bs\Delta \\
&|f_{c}(\bs\theta)-f(\bs\theta)| \le \frac{c^2}{2} Lp \\
& \|\nabla f_{c}(\bs\theta)-\nabla f(\bs\theta)\|\le \frac{c}{2}L(p+3)^{\frac{3}{2}}
\end{align}
\end{subequations}

Subsequently,~\cite{balasubramanian2018zeroth} contributed to supplying the convergence theory for the case that $f(\bs\theta)$ is generally smooth but nonconvex, and 1) the problem is constrained-- $\bs\theta$ is required to be in some compact set $\mathcal{X}$, or 2) the gradient $\nabla F(\bs\theta,\xi)$ is almost surely sparse. Furthermore, it is assumed that there is a global bound on the gradient norm.
In this setting, function evaluations are assumed to be unbiased, that is, for $\hat F\sim F(\bs\theta,\xi)$, we have $\mathbb{E}[\hat F]=f(\bs\theta)$.

\subsection{A Review of Asymptotic Results for SPSA}
We present the single work in the literature that we were able to find with guarantees for SPSA under conditions of biased but bounded noise on the function evaluations, namely~\cite{granichin2014simultaneous}. In this work, several assumptions are made regarding the noise and the algorithm, and an \emph{asymptotic} result is given, indicating the eventual convergence to an approximately optimal point. We note that this work considers a time-varying instance where the objective function to minimize is changing over time, and as such, we simplify the original results for the purposes of our work.

Now we review the assumptions in the paper, considering that an evaluation of the noise function $f$ generates a value in the form of~\eqref{eq:feval}, which for reference we repeat below:
\[
f' \sim F(\bs \theta,\xi) = f(\bs \theta)+n(\bs \theta,\xi),\,\xi\in\chi 
\]
\begin{assumption}\label{ass:31}\cite{granichin2014simultaneous}
It holds that
\begin{enumerate}
    \item It holds that $\mathbb{E}\left[\|(n(\bs \theta_k,\xi_k)-n(\bs \theta_{k-1},\xi_{k-1})\|^2\right]\le C_n$ for some $C_n>0$.
    \item The function $f(\bs \theta)$ is $\mu$-strongly convex, i.e., 
    \[
    \langle \bs\theta-\bs \theta^*,\mathbb{E}\left[\nabla F(\bs\theta,\xi)\right]\rangle \ge \mu \|\bs\theta-\bs \theta^*\|^2
    \]
    for some unique minimum $\bs \theta^*$, for all $\bs\theta$.
    \item It holds that for the unique minimum,
    \[
    \mathbb{E}\left[\left\|\nabla F(\bs \theta^*,\xi)\right\|^2\right] \le G
    \]
    for some $G>0$
    \item The gradient is $L$-Lipschitz for all $\xi$, that is,
    \[
    \|\nabla F(\bs\theta,\xi)-\nabla F(\bs\theta',\xi)\|\le L \|\bs\theta-\bs\theta'\|
    \]
\end{enumerate}
\end{assumption}
Among these assumptions, given that the objective function corresponds to the realization of a quantum circuit, and the parameters $\bs\theta$ rotate an element of a Lie group (the quantum gates), and each source of potential noise is bounded (see Section~\ref{sec:noise} above), assumptions related to boundedness of the variance clearly holds. In addition, from Section~\ref{sec:vqas} it is clear that all the objective function is continuously differentiable with respect to $\bs\theta$, being matrix exponentials, and thus possibly expressed as trigonometric functions, which are all globally Lipschitz continuous.

However, the second assumption, requiring strong convexity, is one that is rarely satisfied in our running example of minimizing functions that evaluate runs of quantum circuits. The imaginary exponentials, equivalently trigonometric functions, that are characteristic of VQA objectives are oscillatory in nature and thus exhibit both positive and negative curvature, thus violating this condition. 

The main Theorem in that work can be summarized as,
\begin{theorem}\cite[Theorem 1]{granichin2014simultaneous}
It holds that for all $\epsilon>0$ there exists some $K$ such that for $k\ge K$,
\[
\sqrt{\mathbb{E}\left\|\bs \theta_k-\bs \theta^*\right\|^2} \le M+\epsilon
\]
where $M$ depends on problem and algorithmic constants. 
\end{theorem}
Thus for this case there is a mean squared error result asymptotically.

In this paper we seek to analyze the method to derive a non-asymptotic result, i.e., one that describes the rate of convergence, and do so applicably for nonconvex functions. For insight we turn to other related work.

\subsection{A Review of Guarantees for Two Point Function Approximation}\label{sec:reviewguarantees}
We consider now the SA iteration with the gradient approximation $\bs g_k$ given by~\eqref{eq:twopointNesterov}. 
The work~\cite{duchi2015optimal} presents convergence guarantees for the smooth convex case as well as well the convex but possibly nonsmooth set of functions.

Subsequently,~\cite{balasubramanian2018zeroth} contributed to supplying convergence theory for the case that $f(\bs\theta)$ is generally smooth but nonconvex, and 1) the problem is constrained-- $\bs\theta$ is required to be in some compact set $\mathcal{X}$, or 2) the gradient $\nabla F(\bs\theta,\xi)$ is almost surely sparse. Furthermore, it is assumed that there is a global bound on the gradient norm. In our context, the presence of $\mathcal{X}$ could correspond to some a priori knowledge in regards to bounds on the permissible $\bs\theta$, which would not be standard in practice, i.e., typically the $\bs\theta$ are taken to be unconstrained and arbitrary in $\mathbb{R}^p$.

With the Frank-Wolfe gap defined as $\bs g^k_G:=\langle \nabla f(\bs \theta_k),\bs \theta_k-\hat{\bs\theta}\rangle,\,\hat{\bs\theta}:=\argmin\limits_{\bs u\in\mathcal{X}}\langle \nabla f(\bs \theta_k),\bs u\rangle$, the work demonstrates in~\cite[Theorem 2.1]{balasubramanian2018zeroth} that $\mathbb{E}[\bs g^R_G]\le \frac{C}{\sqrt{N}}$, with $N$ the number of iterations and the constant $C$ depending on the noise, the gradient bound, and the size of the feasible set. The superscript $R$ indicates a randomly chosen combination of iterations. 

In this setting, function evaluations are assumed to be unbiased, i.e., for $\hat F\sim F(\bs\theta,\xi)$, it holds that $\mathbb{E}[\hat F]=f(\bs\theta)$. 

Finally, closest to the analysis framework we use here, the seminal work~\cite{ghadimi2013stochastic} considers both gradient and zero-order schemes for stochastic optimization (still unbiased, in both cases). In this case, one considers a fixed budget of $N$ iterations, and with $c_k\equiv c= O(1/\sqrt{N})$ and $\alpha_k\equiv \alpha = O(1/\sqrt{N})$ one obtains that the gradient at a particular randomly chosen combination of iterates is bounded by a term that is itself $O(1/\sqrt{N})$ for generally nonconvex functions plus an additive term indicating the bias due to the approximation accuracy. This is the same asymptotic rate as in the case of unbiased SGD, however, now with only an approximate stationarity guarantee, and larger constant terms appearing in the expressions for the convergence rate.

\subsection{A Review of Guarantees for Biased SGD}
Ref.~\cite{tadic2017asymptotic} explores the asymptotic guarantees of biased stochastic approximation under the model that in~\eqref{eq:biasedsg} we have $\bs n(\bs \theta_k,\xi_k)=\bs\zeta_k+\bs\eta_k$ with $\mathbb{E}[\bs\zeta_k]=0$ and $\limsup\|\eta_k\|<\infty$ and prove, depending on the regularity of the function, various bounds with regard to $\limsup\limits_{k\to\infty} \|\nabla f(\bs \theta_k)\|$ and $\dist(\bs \theta_k,\bs \Theta^*)$ with $\bs \Theta^*$ the set of stationary points. 

More recently, with quantitative bounds, the work in~\cite{ajalloeian2020convergence} considers the model for biased SGD~\eqref{eq:biasedsg} by:
\begin{equation}\label{eq:biasedsgassume}
\bs n(\bs\theta,\xi) = \bs b(\bs\theta)+\bs r(\bs\theta,\xi),\,\|\bs b(\bs\theta)\|\le m\|\nabla f(\bs\theta)\|^2+\zeta^2,\,\mathbb{E}_{\xi}\|\bs r(\bs\theta,\xi)\|^2\le M\|\nabla f(\bs\theta)+\bs b(\bs\theta)\|^2+\sigma^2
\end{equation}
they obtain:
\begin{lemma}\cite[Lemma 3]{ajalloeian2020convergence}
If $\alpha_k\equiv\alpha\le \frac{1}{(1+M)L}$ then it holds that after $K$ iterations,
\begin{align}\label{eq:step}
    \frac{1}{K}\sum\limits_{k=1}^K \mathbb{E}\|\nabla f(\bs \theta_k)\|^2
\le \frac{2F}{K\alpha(1-m)}+\frac{\alpha L\sigma^2}{1-m}+\frac{\zeta^2}{1-m},
\end{align}
where $F$ is the distance between $f(\bs \theta_0)$ and the local minimum attained.
\end{lemma}

\begin{theorem}\cite[Theorem 4]{ajalloeian2020convergence} \label{th:32}
For $\epsilon>0$, if $\alpha=\min\left\{\frac{1}{(1+M)L},\frac{\epsilon(1-m)+\zeta^2}{2L\sigma^2}\right\}$, then with $K=\mathcal{O}\left(\frac{M+1}{\epsilon(1-m)+\zeta^2}+\frac{\sigma^2}{\epsilon^2(1-m)^2+\zeta^4}\right)LF$ iterations, it holds that
\[
\frac{1}{K}\sum\limits_{k=1}^K\mathbb{E} \|\nabla f(\bs \theta_k)\|^2
=\mathcal{O}\left(\epsilon+\frac{\zeta^2}{1-m}\right).
\]
\end{theorem}
For completeness, we shall also mention~\cite{pmlr-v99-karimi19a} which presents a (complementary) non-asymptotic analysis of the convergence guarantees associated with biased SGD under a similar model. 

Thus for circuits in which the gradient is computed analytically these results indicate the ultimate performance and guarantees in regards to the algorithm's output. As an insight into the sequel, we can expect that using zero-order methods to approximately estimate the gradient shall result in at least as similar perturbations, as in relative inaccuracy to stationarity.


\section{New Iteration Complexity Results} \label{sec:newconv}

Our main result is the analysis of iteration complexity of 
the two-point approximation and the SPSA algorithm under the setting of biased function evaluations. We remind the reader that we consider the estimation as based on the expansion given in~\eqref{eq:feval}. That is: the biased function evaluation $f'$ is the true function $f(\bs \theta)$ corrupted by noise $n(\bs \theta,\xi)$: 
\[
f' \sim F(\bs \theta,\xi) = f(\bs \theta)+n(\bs \theta,\xi),\,\xi\in\chi(\bs\theta),
\]
 wherein $\bs \theta \in \mathbb{R}^p$ are the deterministic VQA parameters and $\xi \in \mathbb{R}$ is the sample generated from the random field, as in Eqs. \eqref{eq:optprob} and \eqref{eq:feval}.

We make two assumptions. 
First, we assume the noise $n(\bs \theta_k,\xi_k)$  associated with a function evaluation
is a mixture
of a mean-zero component $r_k$ and a biased component $b_k$, wherein the bias is bounded from above $\|b_k\|\le b$, almost surely (a.s.). 

\begin{assumption}\label{as:boundbias}
The function evaluation~\eqref{eq:feval} satisfies,
\begin{equation}\label{eq:biasas}
    n(\bs \theta_k,\xi_k) = r_k+b_k,\text{ with }\mathbb{E}[r_k]= 0,\,\mathbb{E}\|r_k\|^2\le \sigma^2,\text{ and }\|b_k\|\le b\,\,\, {\rm a.s.}
\end{equation}
\end{assumption}
{Indeed, recalling the form for the objective function for VQA, starting from the discussion from~\eqref{eq:U},
\[
C(\bs\theta) = \braket{\psi_0 | U^\dagger(\bs\theta) O U(\bs\theta) |\psi_0},\text{ with }U(\bs \theta) = U_p(\theta_p) \ldots U_1(\theta_1)\text{ and }U_i(\theta_i) = W_i\exp(-\im \theta_i H_i)
\]
In considering noise as an error that takes place at one of the gates, this would correspond to a change of $H_i$. Note that regardless, the modulus of the exponential term is one, $\vert \exp(-\im \theta_i \tilde H_i)\vert=1$ for any perturbed $\tilde H_i$. For illustration, let us compute the error corresponding to an evaluation of a three layer circuit, with a gate error in the middle layer:
\begin{tiny}
\[
\begin{array}{l}
\braket{\psi_0 | W_3\exp(\im \theta_3 H_3)W_2\exp(\im \theta_2 \tilde H_2)W_1\exp(\im \theta_1 H_1) O W_3\exp(-\im \theta_3 H_3)W_2\exp(-\im \theta_2 \tilde H_2)W_1\exp(-\im \theta_1 H_1) |\psi_0} \\
- \braket{\psi_0 | W_3\exp(\im \theta_3 H_3)W_2\exp(\im \theta_2  H_2)W_1\exp(\im \theta_1 H_1) O W_3\exp(-\im \theta_3 H_3)W_2\exp(-\im \theta_2  H_2)W_1\exp(-\im \theta_1 H_1) |\psi_0} \\ 
= \braket{\psi_0 | W_3\exp(\im \theta_3 H_3)W_2\exp(\im \theta_2 \tilde H_2)W_1\exp(\im \theta_1 H_1) O W_3\exp(-\im \theta_3 H_3)W_2\exp(-\im \theta_2 \tilde H_2)W_1\exp(-\im \theta_1 H_1) |\psi_0} \\
-  \braket{\psi_0 | W_3\exp(\im \theta_3 H_3)W_2\exp(\im \theta_2 H_2)W_1\exp(\im \theta_1 H_1) O W_3\exp(-\im \theta_3 H_3)W_2\exp(-\im \theta_2 \tilde H_2)W_1\exp(-\im \theta_1 H_1) |\psi_0} \\
+ \braket{\psi_0 | W_3\exp(\im \theta_3 H_3)W_2\exp(\im \theta_2 H_2)W_1\exp(\im \theta_1 H_1) O W_3\exp(-\im \theta_3 H_3)W_2\exp(-\im \theta_2 \tilde H_2)W_1\exp(-\im \theta_1 H_1) |\psi_0} \\ 
- \braket{\psi_0 | W_3\exp(\im \theta_3 H_3)W_2\exp(\im \theta_2 H_2)W_1\exp(\im \theta_1 H_1) O W_3\exp(-\im \theta_3 H_3)W_2\exp(-\im \theta_2 H_2)W_1\exp(-\im \theta_1 H_1) |\psi_0} \\
\le \braket{\psi_0 | W_3\exp(\im \theta_3 H_3)W_2(\exp(\im \theta_2 \tilde H_2)-\exp(\im \theta_2 H_2))W_1\exp(\im \theta_1 H_1) O W_3\exp(-\im \theta_3 H_3)W_2\exp(-\im \theta_2 \tilde H_2)W_1\exp(-\im \theta_1 H_1) |\psi_0} \\
+ \braket{\psi_0 | W_3\exp(\im \theta_3 H_3)W_2\exp(\im \theta_2 H_2)W_1\exp(\im \theta_1 H_1) O W_3\exp(-\im \theta_3 H_3)W_2(\exp(-\im \theta_2 \tilde H_2)-\exp(-\im \theta_2  H_2))W_1\exp(-\im \theta_1 H_1) |\psi_0} 
\end{array}
\]
\end{tiny}
}

{
Recalling that $W_i\in U(d)$, and operations $\exp(i\theta H)$ are norm preserving, we can obtain a bound as based on the operation error. Namely, the observation operator $O$ is the one component of the function that can conscript or enlarge the norm of the input. Thus, an explicit bound can be characterized as,
\[
b= 2\|\exp(i\theta H)-\exp(i\theta\tilde H)\|\sigma^M(O)
\]
where $\sigma^M(O)$ is the largest eigenvalue of $O$, which is of course constant for a given problem, producing the upper bound. Any statistically symmetric additional noise can be considered as incorporated into $r_k$.
}

Second, we have the standard assumption on the Lipschitz continuity of the gradient. {From the equations above explicitly detailing the function evaluations, it is clear that the expressions are entirely analytic in the parameters $\theta$:}
\begin{assumption}\label{as:lipgrad}
The gradient is $L$-Lipschitz for all $\xi$, that is,
    \begin{align}
    \label{eq:lipgrad}
    \|\nabla F(\bs\theta,\xi)-\nabla F(\bs\theta',\xi)\|\le L \|\bs\theta-\bs\theta'\|
    \end{align}
\end{assumption}

\textcolor{black}{To comment on the applicability to VQA: a Lipschitz constant on the gradient of a function corresponds to the upper bound on the largest Hessian singular value. Note that the use of $\xi$ in the definition is ideal noise. In this case, recalling that each gate is norm preserving, in between a map from and onto the Boch sphere, the Lipschitz bound would be $L=\sigma^M(O)$.}

\paragraph{Non-Asymptotic Convergence Guarantees for Biased Two-Point Approximation}

We now state our first main result on the convergence rate of the two-point approximation with bias. 
In particular, we consider two regimes:
\begin{enumerate}
    \item[\emph{(1)}] one in which we analyze the performance after a given number $K$ of iterations,
    \item[\emph{(2)}] one in which we analyze the performance of a mixture $\bs \theta_R$ of the iterates with a diminishing step size. That is, $\bs \theta_R$ is a random variable where each iterate (out of $K$ iterates, with $K$ countable) is sampled with a probability that depends on the diminishing step size. 
\end{enumerate}
Across both regimes, we bound the expectation of a norm of the gradient, which could be seen as a measure of the distance to the nearest critical point. 
Across both regimes, the bound involves three terms: 
\begin{itemize}
    \item one that is independent of the noise and depends only on the modulus $L$ of Lipschitz continuity of the gradient \eqref{eq:lipgrad} 
    and a certain measure of the objective range $4(f_c(\bs \theta_k)-f^*_{c})$, where $f^*_{c}$ is a lower bound on the objective
    \item one that depends on the variance $\sigma^2$ \eqref{eq:biasas} of the expectation-zero noise \eqref{eq:biasas} and on the magnitude $b$ of the bias \eqref{eq:biasas}, 
    \item one that depends on the magnitude $b$ of the bias \eqref{eq:biasas}.
\end{itemize}
The first two are seen to asymptotically decrease to zero with the iterations, while the last presents a lower bound on the measure of optimality. Let us recall the smoothed function $f_c$ from~\eqref{eq:litdiff}. 

\begin{theorem}\label{th:twopoint}
Consider the two-point function approximation algorithm with iteration~\eqref{eq:saiter} for a problem that meets the Assumptions~\ref{as:boundbias} and~\ref{as:lipgrad}. 

\begin{enumerate}
    \item 
Consider a fixed budget of $K$ iterations. Let $\alpha_k\equiv\alpha/\sqrt{K}$ with $\alpha\le L/4$. It holds that the iterates satisfy:
\begin{equation}\label{eq:totalconvfixed0}
   \min\limits_{k\in[K]} \mathbb{E}\|\nabla f(\bs \theta_k)\|^2 \le \frac{4(f_c(\bs \theta_k)-f^*_{c})}{\alpha\sqrt{K}} +\frac{2L\alpha(\sigma^2+\frac{8b^2}{c^2})}{\alpha\sqrt{K}} +\left[\frac{2b^2}{c^2}+\frac{c^2}{2}L(p+3)^3\right]
\end{equation}

\item Consider a diminishing stepsize of $\alpha_k=\alpha/k^{\gamma}$ with $\gamma\in (1/2,1]$ and $\alpha\le L/4$. Then it holds that, taking,
\[
\bs \theta_R\sim \bs \theta_k\text{ w.p. } \alpha_k/\sum\limits_{l=1}^K \alpha_l
\]
we have,
\begin{equation}\label{eq:totalconvdim0}
   \mathbb{E}\|\nabla f(\bs \theta_R)\|^2 \le \frac{4(f_c(\bs \theta_k)-f^*_{c})}{\alpha K^\gamma} +\frac{2L\alpha(\sigma^2+\frac{8b^2}{c^2})}{\alpha K^{\gamma}} +\left[\frac{2b^2}{c^2}+\frac{c^2}{2}L(p+3)^3\right]
\end{equation}
\end{enumerate}
\end{theorem}

The proof is given in Appendix \ref{app:proofs} on page \pageref{app:proofs}.
Note that this implies that regardless of the number of iterations performed, the measure of stationarity will always have a lower bound given by the right-hand side (``effect of bias''), which is:
\begin{align}
\label{eq:effbias}
    \frac{2b^2}{c^2}+\frac{c^2}{2}L(p+3)^3.
\end{align}
When we minimize the effect of bias \eqref{eq:effbias}, as a function of $c$, we obtain the following:
\begin{align}
    c^* =\left[\frac{4b^2}{L(p+3)^3}\right]^{1/4}.
\end{align}

\begin{figure}[]
    \centering
    \includegraphics[scale=0.25]{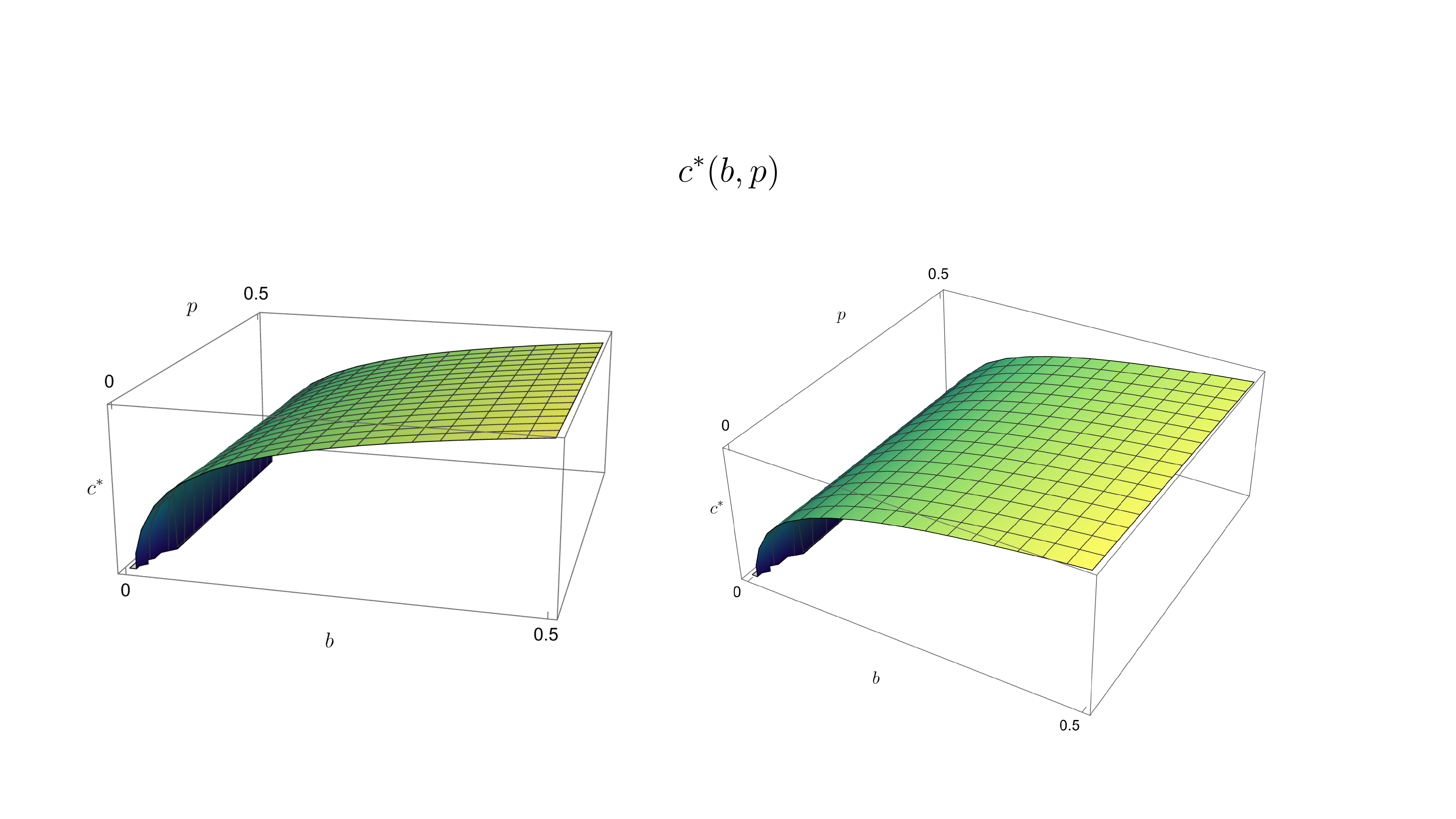}
    \caption{The (logarithm of the) bound surface of the optimal perturbation $c$ as a function of $b,p$.}
    \label{fig:cbp}
\end{figure}

\paragraph{Non-Asymptotic Guarantees for SPSA}

Now, we consider applying a similar line of analysis to SPSA. Recall that the noisy gradient estimate is of the form~\eqref{eq:spsa}
\begin{align*}
    [\bs g_k]_i \sim \frac{F(\bs \theta_k+c_k\bs\Delta_k,\xi_k^{1})-F(\bs \theta_k-c_k\bs\Delta_k,\xi_k^{2})}{2c_k [\bs\Delta_k]_i}. 
\end{align*}

We shall make use of a Lemma in the literature characterizing the bias associated with this estimate in the case of being able to compute unbiased samples of the noisy objective function itself.
\begin{lemma}\label{lem:spsaacc}\cite[Lemma 1]{spall1992multivariate}
Suppose that $\{\mathbf\Delta_k\}$, and its components, are i.i.d. and symmetrically distributed around zero with $[\mathbf\Delta_k]_i\le B_0$ almost surely and $\mathbb{E}[[\mathbf\Delta_k]_i^{-1}]\le B_1$, with $B_0,B_1>0$. Furthermore, assume that for some $K$, for all $k\ge K$, there exists $\delta>0,B_2>0$ such that $\left\|\frac{\partial^3}{\partial \bs\theta^3}f(\bs\theta_k)\right\|\le B_2$ for all $\bs\theta$ in a $\delta$-neighborhood around $\bs\theta_k$. Then it holds that, for such $\bs \theta_k$, the SPSA estimate~\eqref{eq:spsa} without bias in function evaluations yields
\begin{equation}\label{eq:spsabiasnobias}
    \bs g_k-\nabla f(\bs\theta_k)=\bs{r}(\xi^1_k,\xi^2_k,\bs\Delta_k)+\bs b_0(\bs\theta_k),\,\mathbb{E}[\bs{r}(\xi^1_k,\xi^2_k,\bs\Delta_k)]=0,\,\|\bs b_0(\bs\theta_k)\|\le b_0 c^2_k
\end{equation}
for some $b_0$ depending on $B_0$, $B_1$, $B_2$, and $p$.
\end{lemma}

We will also need an additional assumption:

\begin{assumption}\label{as:gradbound}
The perturbation is bounded away from zero, i.e., 
 there exists $B_3>0$ such that $[\bs \Delta_k]_i^{-1}\le B_3$ almost surely for all $k$.
\end{assumption}

With this, we are able to prove the convergence rate of the SPSA procedure with bias. 
Again, we consider two regimes. 
One, where we analyze the performance after a given number $K$ of iterations.
Second, we analyze the performance of mixture of the outcomes of iterations with a diminishing step size.
Across both regimes, we bound the expectation of a norm of the gradient as a measure of the distance to the nearest stationary point. 
Again, the bounds involve three terms, the first one independent of the noise and approaching zero with the iterations, the second one depending on the noise and bias but also approaching zero, and the third the asymptotic bias. 

We incorporate the previous Lemma in the following sense in the proof of the Theorem, which is also in the Appendix: we consider the SPSA procedure in the case wherein there is no bias in function evaluations, then consider a perturbation of this. Finally: 

\begin{theorem}\label{th:spsa}
Consider the SPSA algorithm with the iteration~\eqref{eq:saiter} and gradient estimate~\eqref{eq:spsa} for a problem satisfying Assumptions~\ref{as:boundbias} and~\ref{as:lipgrad}. Also, let the SPSA algorithm satisfy the conditions of Lemma~\ref{lem:spsaacc} and Assumption~\ref{as:gradbound}. 
Then it holds that

\begin{enumerate}
    \item 
Consider a fixed budget of $K$ iterations. Let $\alpha_k\equiv\alpha/\sqrt{K}$ with $\alpha\le L/2$. It holds that the iterates satisfy:
\begin{equation}\label{eq:totalconvspsafixed}
   \min\limits_{k\in[K]} \mathbb{E}\|\nabla f(\bs \theta_k)\|^2 \le \frac{4(f_c(\bs \theta_k)-f^*_{c})}{\alpha\sqrt{K}} +\frac{2L\alpha\left[\frac{\sigma^2 B_3^2}{c_k^2}+2b_0 c_k^2+\frac{2b^2 B_3^2}{c^2_k}\right]}{\sqrt{K}} +\left[b_0 c_k^2+\frac{b^2 B_3^2}{c^2_k}\right]
\end{equation}

\item Consider a diminishing stepsize of $\alpha_k=\alpha/k^{\gamma}$ with $\gamma\in (1/2,1]$ and $\alpha\le L/2$. Then it holds that, taking,
\[
\bs \theta_R\sim \bs \theta_k\text{ w.p. } \alpha_k/\sum\limits_{l=1}^K \alpha_l
\]
we have,
\begin{equation}\label{eq:totalconvspsadim}
   \mathbb{E}\|\nabla f(\bs \theta_R)\|^2 \le \frac{4(f_c(\bs \theta_k)-f^*_{c})}{\alpha K^\gamma} +\frac{2L\alpha\left[\frac{\sigma^2 B_3^2}{c_k^2}+2b_0 c_k^2+\frac{2b^2 B_3^2}{c^2_k}\right]}{ K^{\gamma}} +\left[b_0 c_k^2+\frac{b^2 B_3^2}{c^2_k}\right]
\end{equation}
\end{enumerate}
\end{theorem}

The proof is in Appendix \ref{app:proofs} on page \pageref{secondproof}.

\begin{figure}[t!]
\centering
\includegraphics[scale=0.22]{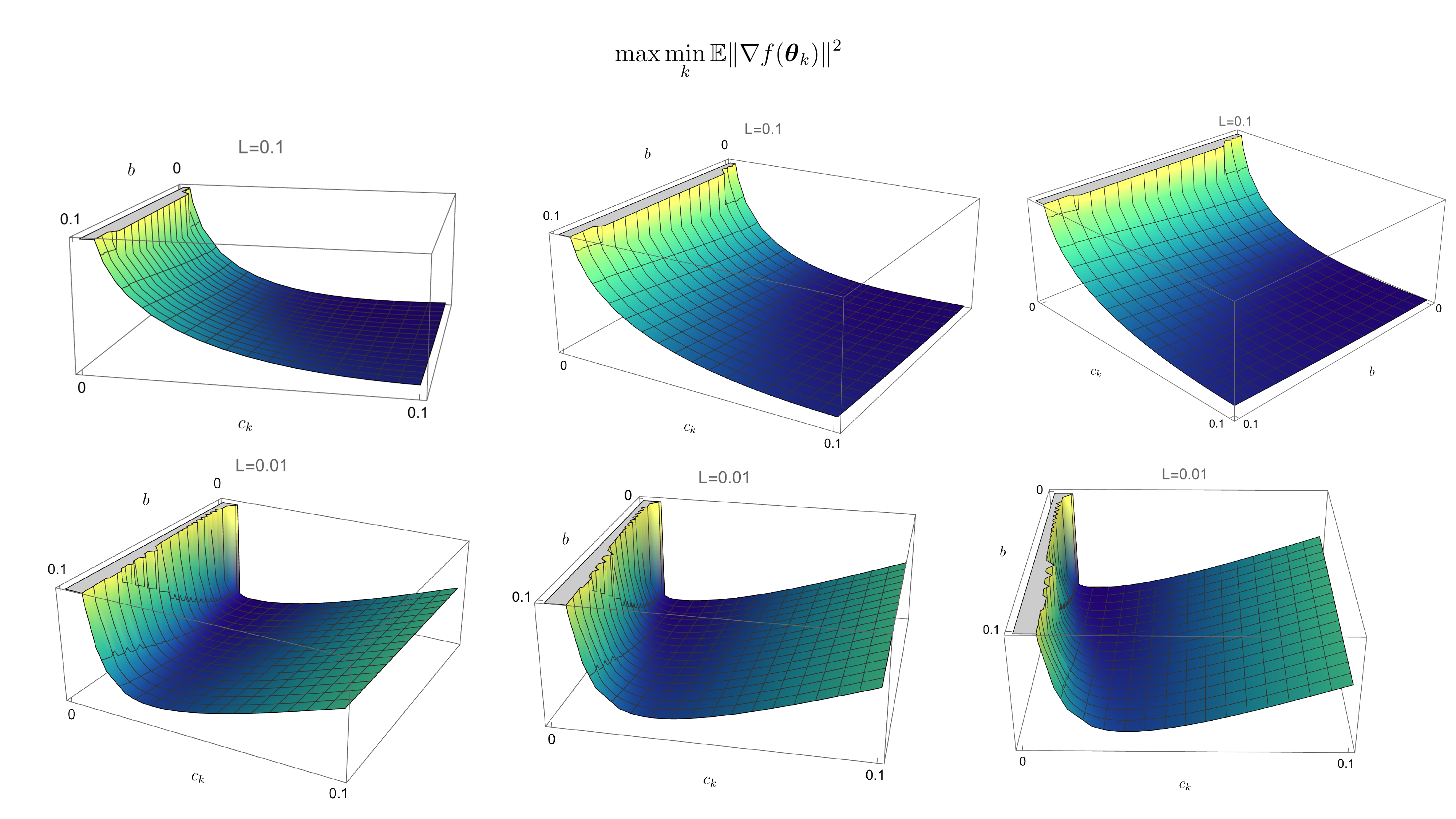}
    \caption{The bound surface of Eq. \eqref{eq:totalconvspsafixed} with $L=8$, $\alpha=4$, $\sigma=1,K=10^5$ for $0 \leq b,c \leq 1$. The $z$-axis is in logarithmic scale.}
    \label{fig:boundsurface1}
\end{figure}

\paragraph{Guarantees for the Parameter Shift Rule}
Now, we consider the original Parameter Shift rule as given by~\eqref{eq:derivative}, rewritten for convenience below with $s=\frac{\pi}{2}$,
\[
    \frac{dO_i(\theta_i)}{d\theta_i}= \frac{O_i\left(\theta_i+\frac{\pi}{2}\right)-O_i\left(\theta_i-\frac{\pi}{2}\right)}{2},\quad \nabla O(\bs \theta) = \sum_{i=1}^L \frac{d\tilde{O}_i(\theta_i)}{d\theta_i}
    \]
    Now, we have that, without bias in the function evaluations, this is an \emph{exact} gradient computation. For noisy operations under the ideal probability distribution $\Xi$, this would be an unbiased stochastic gradient. For the real noise distribution $\chi$, however, biased noise will appear from the evaluations of $O\left(\theta_i+\frac{\pi}{2}\right)$ and $O\left(\theta_i-\frac{\pi}{2}\right)$. By our biased noise model given in Assumption~\ref{as:boundbias} we have that,
    \[
    n(\bs\theta,\xi) = O(\bs\theta,\xi\sim\Xi)+\left(O(\bs\theta,\xi\sim\chi)-O(\bs\theta,\xi\sim\Xi)\right)
    \]
    and 
    \begin{align}\label{eq:bLv}
    b=L_V\max\limits_i\sup_{\xi\in \chi} \left|O_i\left(\theta_i \pm \frac{\pi}{2},\xi\sim\chi\right)-O_i\left(\theta_i+\frac{\pi}{2}\right)\right|
    \end{align}
    where $L_V$ are the number of layers in the VQA, to distinguish it from the Lipschitz constant.
    Now, given that the ideal PS computation is the exact stochastic gradient, we are in the biased SGD regime, with the noise setting~\eqref{eq:biasedsgassume} with $\zeta^2 = b$, and it can be seen that we can take $m=M=0$. From~\eqref{eq:step} we have,
    \begin{equation}\label{eq:psbound}
        \frac{1}{K}\sum\limits_{k=1}^K \|\nabla f(\bs \theta_k)\|^2
\le \frac{2F}{K\alpha}+\alpha L\sigma^2+\zeta^2 = \frac{2F}{K\alpha}+\alpha L\sigma^2+b
    \end{equation}

\paragraph{Discussion and Comparison}

Let us compare the convergence guarantees for the two point function approximation, SPSA and Parameter Shift rule by~\eqref{eq:totalconvfixed0},~\eqref{eq:totalconvspsafixed} and~\eqref{eq:psbound}, respectively. 
    
\begin{enumerate}
    \item[\emph{(1)}] It appears that all algorithms have the same overall convergence rate.
    \item[\emph{(2)}] The error grows at least with the square of the maximum eigenvalue of the observable $\sigma_M(O)$ through the influence of this factor on $L$ and $b$.
    \item[\emph{(3)}] The Parameter Shift Rule (PSR) exhibits the smallest effect of bias, in general, as well as the smallest constants for the convergence, which is reasonable given it is an exact gradient computation
    \item[\emph{(4)}] SPSA, however, seems to be the most tuneable, i.e., it is clear that $B_3$ and $c_k$ are freely chosen by the user, and as such can be adjusted so as to significantly mitigate the bias. Thus, in cases wherein significant bias mitigation is called for in order to obtain sensible result, aggressive parameter choices can be made in SPSA to achieve this. \textcolor{black}{Note that the influence of the circuit depth, corresponding to $p$, does not influence the SPSA convergence except via its influence of $b_c$, in contrast to the two point function approximation.}
    \item[\emph{(5)}] The two-point function approximation would perform better than SPSA for generic derivative-free optimization using classical algorithmic parameter choices, but is clearly the option of last resort in variational quantum algorithms, given that there is only one parameter $c$ to adjust to mitigate the noise (compared to SPSA's two: $c$ and $B_3$), while at the same time underperforming PSR, if the circuit is amenable to gradient computations using PSR. \textcolor{black}{Furthermore, the convergence error scales with the cube of the circuit depth, or number of variational parameters, $q$, due to the additional bias incurred with greater dimensional problems in the two point function approximation to the gradient.}

    \item[\emph{(6)}] Taking more samples in each function evaluation would reduce the symmetric noise, and thus $\sigma^2$, which would also have a favorable effect on both the rate of convergence and quality of the final iterate.
    \item[\emph{(7)}] The final bias scales with the error profile as defined by the worst norm difference that occurs in the case of a modified operation, $\|\exp(i\theta H)-\exp(i\theta\tilde H)\|$.
\end{enumerate}


\section{Numerical Experiments}\label{sec:num}
In this section, we illustrate the iteration complexity of VQAs in the context of a standard QAOA test cast, the MAX-CUT problem on a 4-regular graph.

he Quantum Approximate Optimization Algorithm \cite{Farhi2014} is a VQA that can tackle a variety of quadratic unconstrained bnary optimization problems (QUBOs) by finding a approximate solutions. The way it works, as a VQA, is to maximize the cost function of output bit strings. It is a (quantum) random search algorithm where, e.g., for the MAX-CUT problem, it can be viewed as a sequence or quantum random walks on the hypercube defined by the problem. At each iteration, the algorithm provides an improved cost function by optimizing it and feeding it back to the quantum computer for the next iteration. In the standard QAOA one starts with the uniform quantum superposition state and then two different Hamiltonians are encoded, the mixer Hamiltonian $H_M$, and the problem Hamiltonian $H_C$ which is equivalent to the $d=1$ Ising model Hamiltonian. The corresponding circuit is shown in Figure \ref{fig:qaoa}.

The standard quantum circuit is shown in Fig. \ref{fig:qaoa}, 
\begin{figure}[tb]
    \centering
    \includegraphics[scale=0.70]{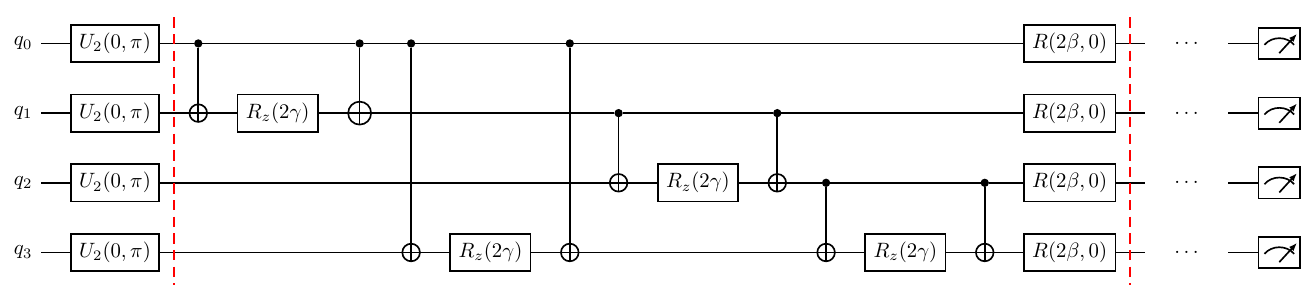}
    \caption{The QAOA circuit for 4 qubits. Within the red crossed lines we denote one layer of the QAOA and with dots we denote the repetition of the layers to the desired depth.}
    \label{fig:qaoa}
\end{figure}
where the gates are defined as usual (cf. Appendix \ref{sec:gatedef}).

In order to perform tests with varying amounts of noise and varying bounds on the bias therein, 
we utilize the Qiskit \texttt{Aer} simulator \cite{Qiskit} with 
the native gateset comprised of identity, $R_z$, $\sqrt{X}$, and CNOT gates,
 to which we transpile the $U_2$ and $R(\theta,\phi)$ gates of Fig. \ref{fig:qaoa}, 
and a custom noise model. The custom noise model
captures three sources of error:
\begin{itemize}
    \item imperfect fidelity of 1-qubit gates: we associate each application of a 1-qubit gate (after transpiling to the native gateset) with a depolarizing channel (cf. p. \pageref{eq:depolarizing}) acting upon $\rho$ as:
    \begin{align}
    \mathcal{E}_{1}(\rho) := 0.999 \rho + 0.001 \text{Tr}(\rho) \frac{\mathds{1}}{2}
    \end{align}
    \item imperfect fidelity of 2-qubit gates: we associate each application of a CNOT gate (i.e., after transpiling $R(\theta,\phi)$ gates of Fig. \ref{fig:qaoa} to the native gateset with CNOT only) with a depolarizing channel acting on $\rho$ as:
    \begin{align}
    \mathcal{E}_{2}(\rho) :=0.99 \rho + 0.01 \text{Tr}(\rho) \frac{\mathds{1}}{4}
    \end{align}
    \item bias in the read-out: 
    we consider a matrix $P$ defining the so-called assignment probabilities, i.e., conditional probability of recording a noisy measurement outcome as $\alpha$ given that the measurement outcome should have been $\beta$, for all representations $\alpha, \beta$ of bit-strings on $4$ qubits.
    We detail two such matrices $P$ below. 
\end{itemize}

In what follows, we first test the optimization performance for a random choice of assignment probabilities within the read-out, while afterwards we define a custom read-out noise model based on the average 1-qubit local read-out errors of some standard freely available low-qubit \texttt{IBMQ} devices; see Table \ref{table:ibmq}. The need of such a custom noise model stems from the fact that it becomes practically impossible to distinguish between the gate errors and the read-out error for anything more than one qubit.

\begin{table}[!htb]
\centering
\begin{tabular}{cccccccc} \hline \hline
\texttt{ibmq\_}    & \texttt{armonk} & \texttt{lima}   & \texttt{santiago} & \texttt{bogota} & \texttt{belem}  & \texttt{quito} & \texttt{manila}   \\  \hline
qubits & 1      & 5      & 5        & 5      & 5      & 5  &   5 \\
$p(1|0)$    & 0.0322 & 0.0078 & 0.0274   & 0.0496 & 0.0052 & 0.0186 &0.0114 \\
$p(0|1)$    & 0.0358 & 0.0260  & 0.0448   & 0.1916 & 0.0352 & 0.0548 &0.034 \\
bias $\tilde b$    & 0.00359  & 0.0182  &  0.0174   & 0.1402 &  0.0300 & 0.0362  &0.0226  \\ \hline
\end{tabular} 
\caption{Assignment probabilities for qubit-flipped measurements across various IBM devices. Observe that all entries of the last row are larger than those of the row before for all devices, clearly showcasing the read-out asymmetry.}
\label{table:ibmq}
\end{table}

\newpage
\subsection*{A random noise model}

First, we revisit the convergence of VQA utilizing SPSA and 
a randomly (and highly non-physical) probability transition matrix $P_{\rm rand}$:

\[ 
P_{\rm rand}=\scalemath{0.84}{\left(
\begin{matrix*}[r]
 0.55 & 0.27 & 0.17 & 0.00 & 0.00 & 0.00 & 0.00 & 0.00 & 0.00 & 0.00 & 0.00 & 0.00 & 0.00 & 0.00 & 0.00 & 0.00 \\
 0.58 & 0.14 & 0.16 & 0.04 & 0.00 & 0.02 & 0.04 & 0.00 & 0.00 & 0.00 & 0.00 & 0.00 & 0.00 & 0.00 & 0.00 & 0.00 \\
 0.10 & 0.10 & 0.60 & 0.13 & 0.01 & 0.01 & 0.02 & 0.00 & 0.00 & 0.01 & 0.01 & 0.00 & 0.00 & 0.00 & 0.00 & 0.00 \\
 0.32 & 0.64 & 0.01 & 0.01 & 0.02 & 0.00 & 0.00 & 0.00 & 0.00 & 0.00 & 0.00 & 0.00 & 0.00 & 0.00 & 0.00 & 0.00 \\
 0.30 & 0.68 & 0.01 & 0.00 & 0.01 & 0.00 & 0.00 & 0.00 & 0.00 & 0.00 & 0.00 & 0.00 & 0.00 & 0.00 & 0.00 & 0.00 \\
 0.02 & 0.52 & 0.44 & 0.00 & 0.00 & 0.01 & 0.00 & 0.00 & 0.00 & 0.00 & 0.00 & 0.00 & 0.00 & 0.00 & 0.00 & 0.00 \\
 0.90 & 0.06 & 0.01 & 0.01 & 0.00 & 0.02 & 0.00 & 0.00 & 0.00 & 0.00 & 0.00 & 0.00 & 0.00 & 0.00 & 0.00 & 0.00 \\
 0.94 & 0.02 & 0.02 & 0.01 & 0.00 & 0.00 & 0.00 & 0.00 & 0.00 & 0.00 & 0.00 & 0.00 & 0.00 & 0.00 & 0.00 & 0.00 \\
 0.25 & 0.29 & 0.13 & 0.28 & 0.00 & 0.00 & 0.00 & 0.02 & 0.02 & 0.00 & 0.00 & 0.00 & 0.00 & 0.00 & 0.00 & 0.00 \\
 0.35 & 0.05 & 0.38 & 0.19 & 0.01 & 0.01 & 0.01 & 0.00 & 0.00 & 0.00 & 0.00 & 0.00 & 0.00 & 0.00 & 0.00 & 0.00 \\
 0.32 & 0.27 & 0.09 & 0.07 & 0.19 & 0.03 & 0.03 & 0.00 & 0.00 & 0.00 & 0.00 & 0.00 & 0.00 & 0.00 & 0.00 & 0.00 \\
 0.81 & 0.10 & 0.06 & 0.01 & 0.00 & 0.00 & 0.00 & 0.00 & 0.00 & 0.00 & 0.00 & 0.00 & 0.00 & 0.00 & 0.00 & 0.00 \\
 0.41 & 0.02 & 0.24 & 0.01 & 0.30 & 0.01 & 0.00 & 0.01 & 0.00 & 0.00 & 0.00 & 0.00 & 0.00 & 0.00 & 0.00 & 0.00 \\
 0.42 & 0.46 & 0.07 & 0.02 & 0.01 & 0.00 & 0.02 & 0.00 & 0.00 & 0.00 & 0.00 & 0.00 & 0.00 & 0.00 & 0.00 & 0.00 \\
 0.15 & 0.18 & 0.64 & 0.03 & 0.00 & 0.00 & 0.00 & 0.00 & 0.00 & 0.00 & 0.00 & 0.00 & 0.00 & 0.00 & 0.00 & 0.00 \\
 0.10 & 0.75 & 0.08 & 0.03 & 0.01 & 0.00 & 0.01 & 0.01 & 0.00 & 0.00 & 0.00 & 0.00 & 0.00 & 0.00 & 0.00 & 0.00 \\
\end{matrix*} \right)}
\]

This matrix was randomly generated subject to the conditions of sparsity and row-stochasticity. 
Using this random noise model $P_{\rm rand}$, we performed several runs of the SPSA to record its iteration complexity. We observe that the mean-valued run converges after approximately 60 iterations, see Fig. \ref{fig:my_spsa}.

\begin{figure}[!htb]
    \centering
    \includegraphics[scale=0.5]{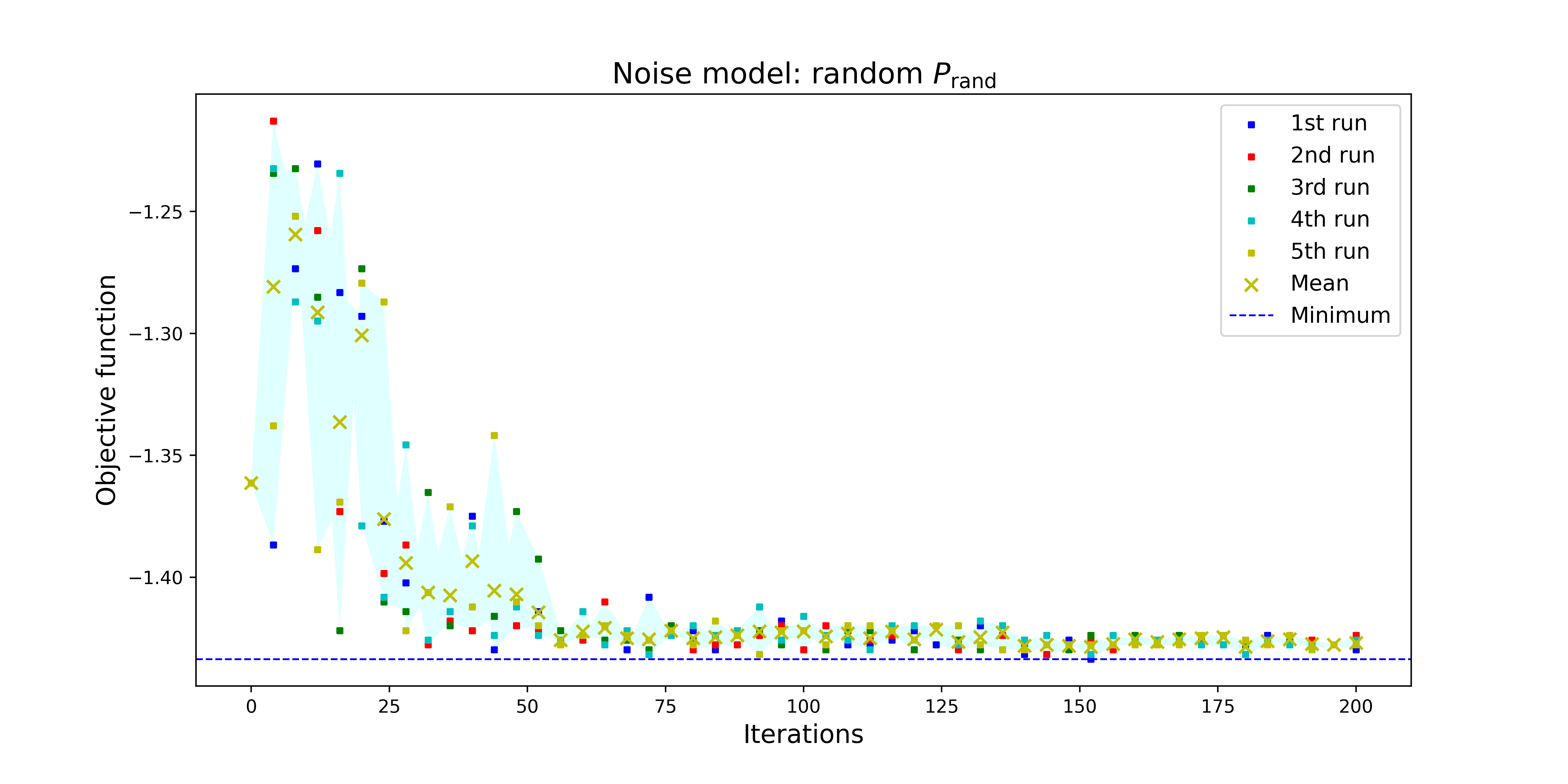}
    \caption{The (noisy) objective values of SPSA for the random noise model $P_{\rm rand}$ for a 4-regular graph. Although this noise model is not physical, we see that convergence is already achieved after approximately 50 iterations.}
    \label{fig:my_spsa}
\end{figure}

\newpage
\subsection*{A realistic noise model}

Next, we consider a realistic biased noise model given by the matrix $P$:

\[ 
P=\scriptscriptstyle{\left(
\begin{smallmatrix*}[r]
x& \tfrac{1-x}{4}  &{}\tfrac{1-x}{4}{}  & {}\tfrac{1-x}{4}{}& {}\tfrac{1-x}{4}{}&&&&&&&&&&&\\
{}\tfrac{1-y}{4}{} &{}y{}&&&& {}\tfrac{1-y}{4}{}&&&{}\tfrac{1-y}{4}{}&&{}\tfrac{1-y}{4}{}&&&&&\\
{}\tfrac{1-y}{4}{}&& {}y{}&&&& {}\tfrac{1-y}{4}{}&&& {}\tfrac{1-y}{4}{}  & {}\tfrac{1-y}{4}{}&&&&&\\
{}\tfrac{1-y}{4}{}&&&{}y{}&&&&& {}\tfrac{1-y}{4}{}& {}\tfrac{1-y}{4}{}&{}\tfrac{1-y}{4}{}&&&&&\\
{}\tfrac{1-y}{4}{}&&&&{}y{}&{}\tfrac{1-y}{4}{}&{}\tfrac{1-y}{4}{}&{}\tfrac{1-y}{4}{}&&&&&&&&\\
& {}\alpha\tfrac{1-x}{2}{}&&& {}\alpha\tfrac{1-x}{2}{} & {}y{}&&&&&&& {}\tfrac{1-y}{2}{} & {}\tfrac{1-y}{2}{}&&\\
&&{}\alpha\tfrac{1-x}{2}{} && {}\alpha\tfrac{1-x}{2}{} && {}y{}&&&&&{}\alpha\tfrac{1-y}{2}{} && {}\alpha\tfrac{1-y}{2}{}&&\\
&&&{}\alpha\tfrac{1-x}{2}{} & {}\alpha\tfrac{1-x}{2}{}&&&{}y{}&&&& {}\alpha\tfrac{1-y}{2}{} & {}\alpha\tfrac{1-y}{2}{} &&&\\
& {}\alpha\tfrac{1-x}{2}{} && {}\alpha\tfrac{1-x}{2}{} &&&&&{}y{}&&&& {}\alpha\tfrac{1-y}{2}{}&&{}\alpha\tfrac{1-y}{2}{}&\\
&& {}\alpha\tfrac{1-x}{2}{}& {}\alpha\tfrac{1-x}{2}{}&&&&& &{}y{}&& {}\alpha\tfrac{1-y}{2}{}&&&{}\alpha\tfrac{1-y}{2}{}&\\
& {}\alpha\tfrac{1-x}{2}{} &{}\alpha\tfrac{1-x}{2}{} &&&&&&&&{}y{}&&&{}\alpha\tfrac{1-y}{2}{} & {}\alpha\tfrac{1-y}{2}{}&\\
&&&&&&{}\alpha\tfrac{1-x}{3}{}&{}\alpha\tfrac{1-x}{3}{}&&{}\alpha\tfrac{1-x}{3}{}&&{}y{}&&&&{}\alpha(1-y){} \\
&&&&&{}\alpha\tfrac{1-x}{3}{}&& {}\alpha\tfrac{1-x}{3}{}&{}\alpha\tfrac{1-x}{3}{}&&&&{}y{}&&&{}\alpha(1-y){} \\
&& &&& {}\alpha\tfrac{1-x}{3}{} & {}\alpha\tfrac{1-x}{3}{} &&&&{}\alpha\tfrac{1-x}{3}{}&&&{}y{}&& {}\alpha(1-y){} \\
&&&&&&&& {}\alpha\tfrac{1-x}{3}{}& {}\alpha\tfrac{1-x}{3}{} & {}\alpha\tfrac{1-x}{3}{} &          && &{}y{}& {}\alpha(1-y){} \\
&&&&&&&&&&& {}\beta\tfrac{1-x}{4}{} & {}\beta\tfrac{1-x}{4}{} & {}\beta\tfrac{1-x}{4}{} & {}\beta\tfrac{1-x}{4}{} & {}y{}      \\ 
\end{smallmatrix*} \right)}
\]

In this model, $x = P(1|0)$, $y=P(0|1)$, and $\alpha,\beta$ are normalization factors such that the sum of each row is one, preserving the probability. 
Our model only takes into account single bit flips in the classical measurement. For example, it assigns probabilities as described for states $\ket{\alpha_1 \alpha_2 \alpha_3 \alpha_4}$ where any (but only one) of the $\alpha_i$ flips its bit, but assigns zero probability for two or more bit flips. 

Using this model and taking into account the conditional probabilities of Table \ref{table:ibmq}, we performed several runs of the SPSA using the probability matrix induced by 1-qubit probability matrix of three \texttt{IBMQ} devices, \texttt{Armonk}, \texttt{Lima}, and \texttt{Bogota}.
In Figure \ref{fig:my_spsa}, we present the multiple noisy trajectories of the objective function values for five runs, together with the mean (golden crosses) and one stadard deviation added to and subtracted from the mean (shaded region) for the five runs.

\begin{figure}[!tbp]
    \centering
    \includegraphics[scale=0.4]{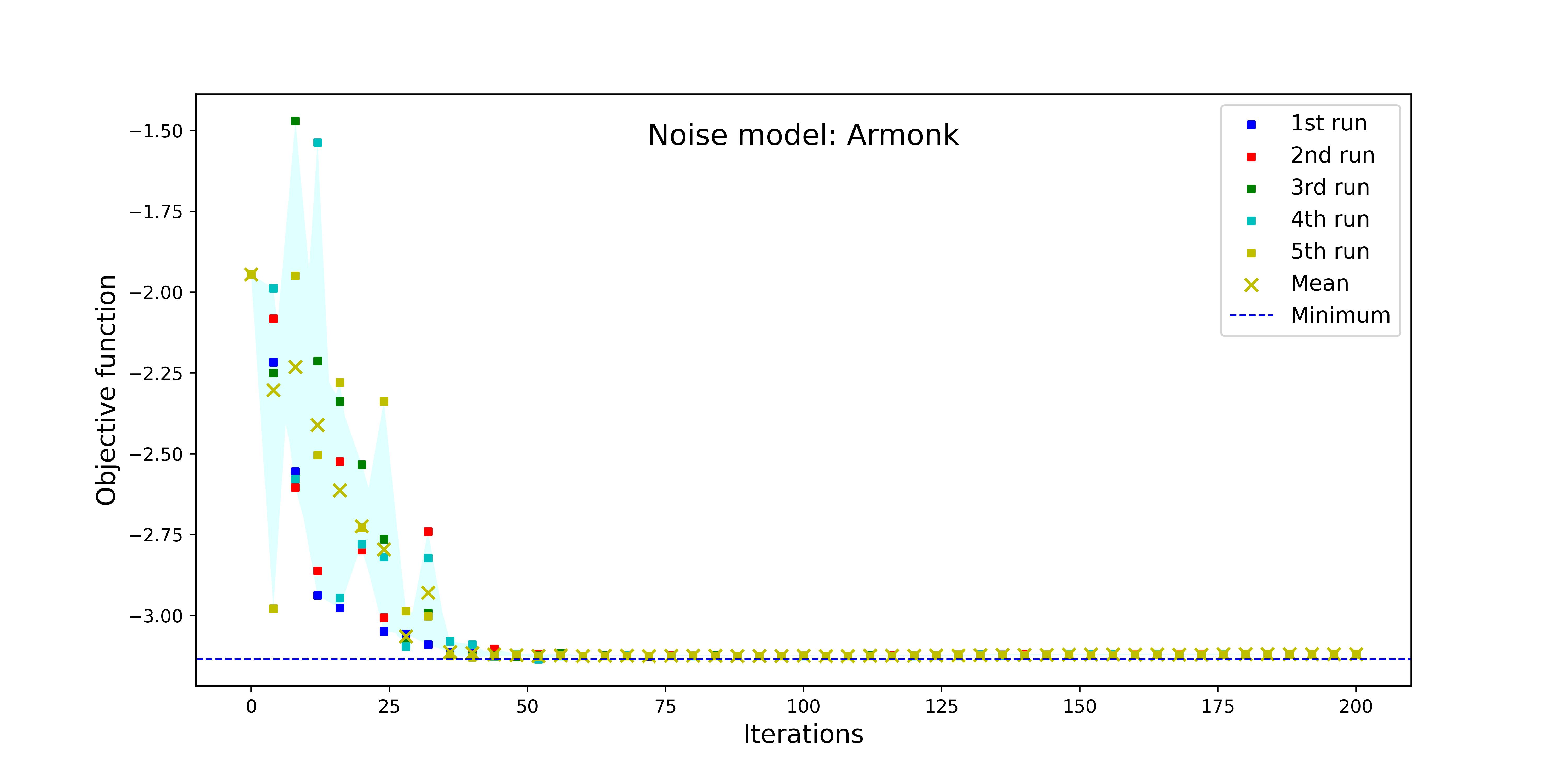}
    \includegraphics[scale=0.4]{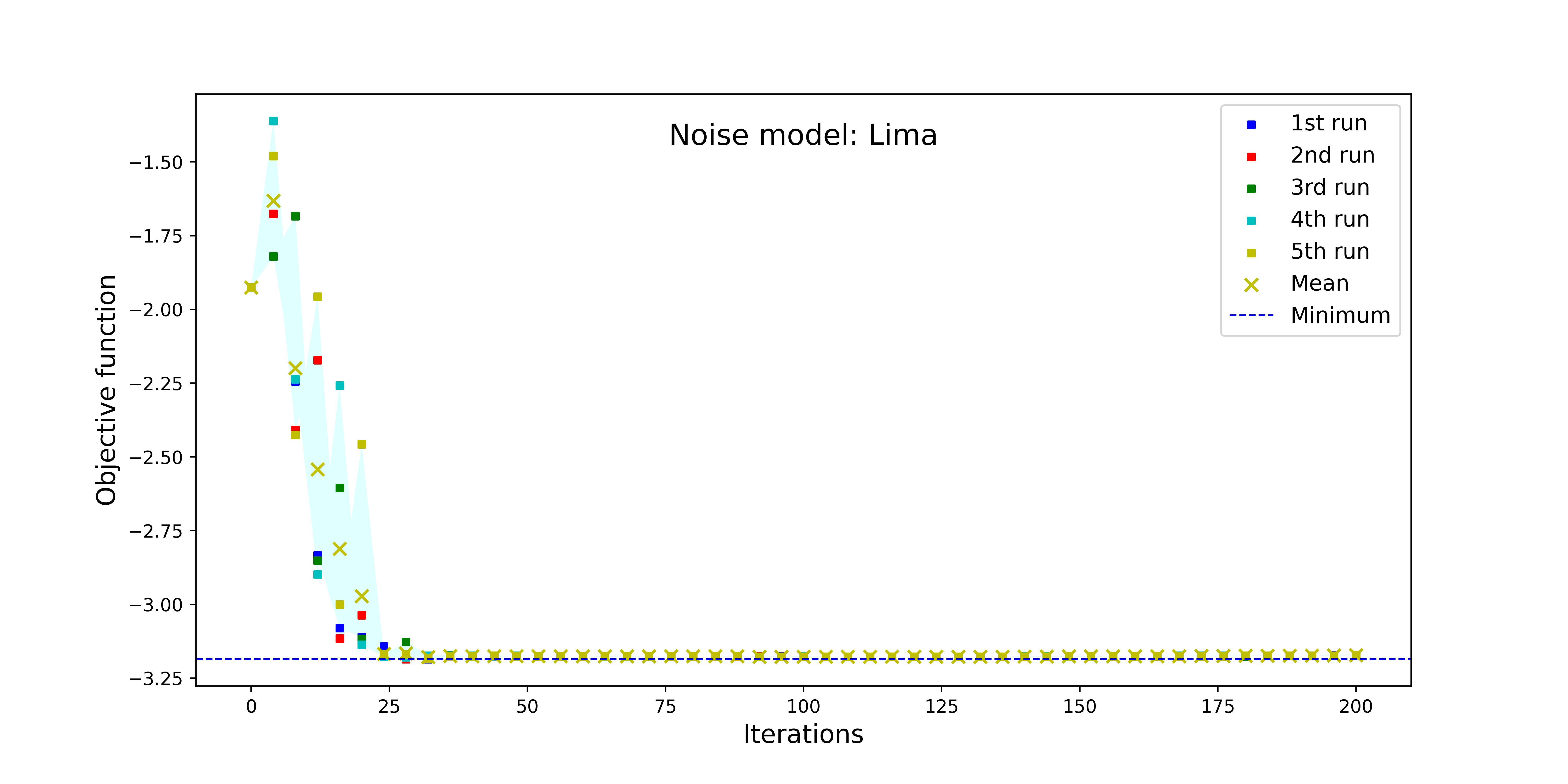}
    \includegraphics[scale=0.4]{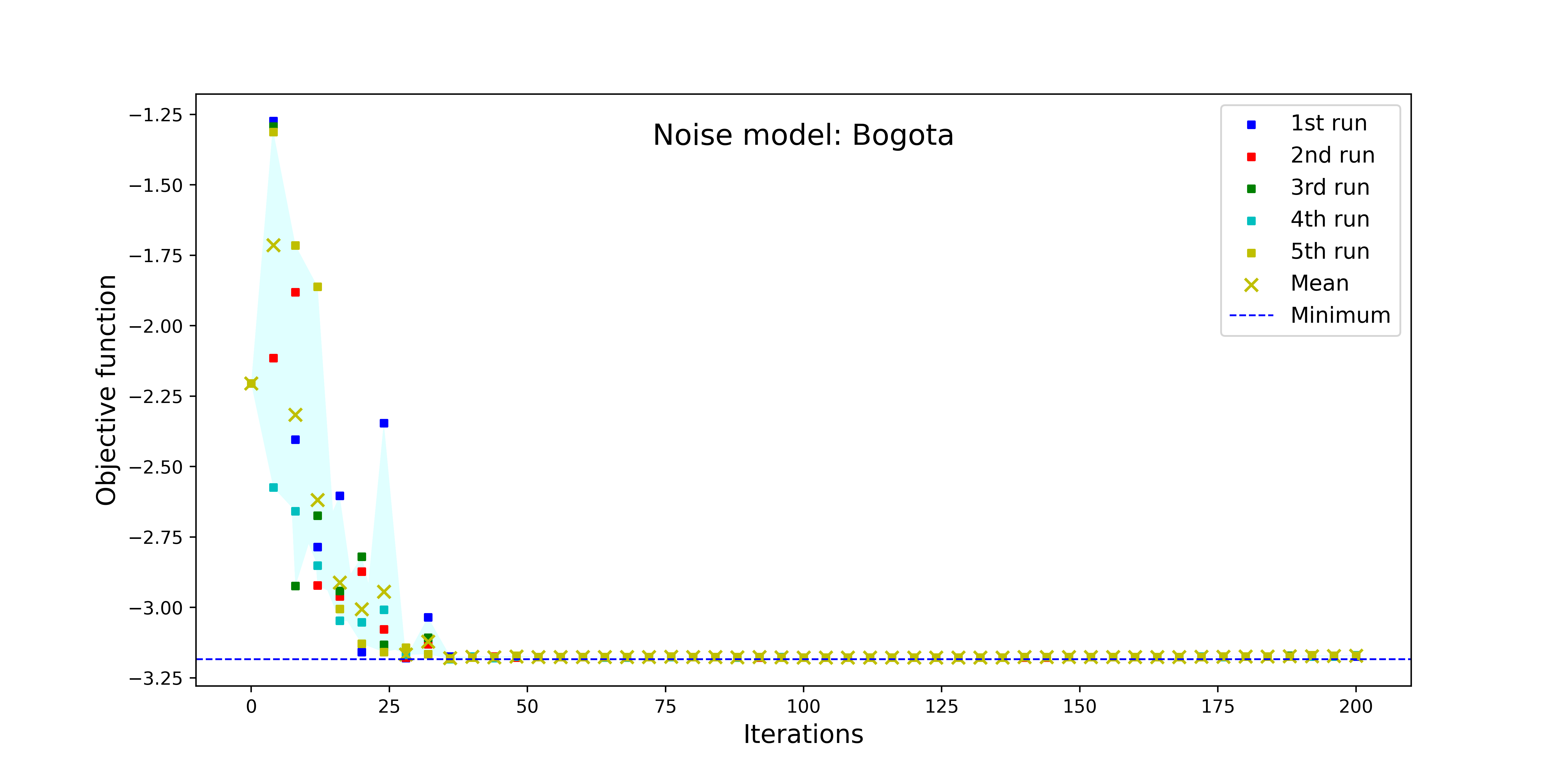}
    \caption{The iteration complexity of SPSA for different IBMQ devices.}
    \label{fig:my_spsa}
\end{figure}

In this setting, consider four different manifestations of the bias:
\begin{itemize}
    \item bias in the read-out model. In particular, a fixed difference $\tilde{b}(c) := |c-x|$ input into our choice $P$ of the model of the read-out error of a 4-qubit device, which defines dividing the magnitude of the SPSA perturbation parameter $c$ across the off-diagonal elements of a row of $P$ above. 
    \item 
    the empirical distribution of objective function values obtained by running the circuit of Fig. \ref{fig:qaoa} for QAOA on the 4-vertex MAXCUT instance described above with a particular read-out model with a particular choice of $\tilde{b}$. 
    In Fig. \ref{fig:objective_data}, we display the empirical distributions obtained by running 1 iteration of QAOA for four different choices of $\tilde{b}$ in the read-out model: 0.01, 0.02, 0.03, 0.04, with 1024 runs per each value of $\tilde{b}$.
   \item an estimate $\hat b$ of the bias in the empirical distribution of the objective function values, as obtained by performing parameter estimation of a mixture model \eqref{eq:biasas} composed of a normal and a uniform distribution on a set of samples of the QAOA run for a representative value of $\bs\theta$.  
In Fig. \ref{fig:objective_data}, clock-wise from the top-left, the estimates  $\hat{b}$ are $0.31224, 0.32779, 0.33484, 0.36894$ based on $\tilde{b}(c)$ of 0.01, 0.02, 0.03, and 0.04, respectively.
For comparison, see 
Fig.~\ref{fig:objective_data2} for the case of no read-out error, but realistic fidelity of the gates. 
\item the impact of the bias on the rate of convergence. 
In Fig. \ref{fig:my_spsa}, for the same values of $c,\tilde{b}$ and estimated bias $\hat b$ as in Fig. \ref{fig:objective_data}, we illustrate the trajectory of objective values by SPSA iteration while considering the variable noise model given by $P$.
\end{itemize}

\begin{figure}[!tbp]
    \centering
    \includegraphics[width=0.49\columnwidth]{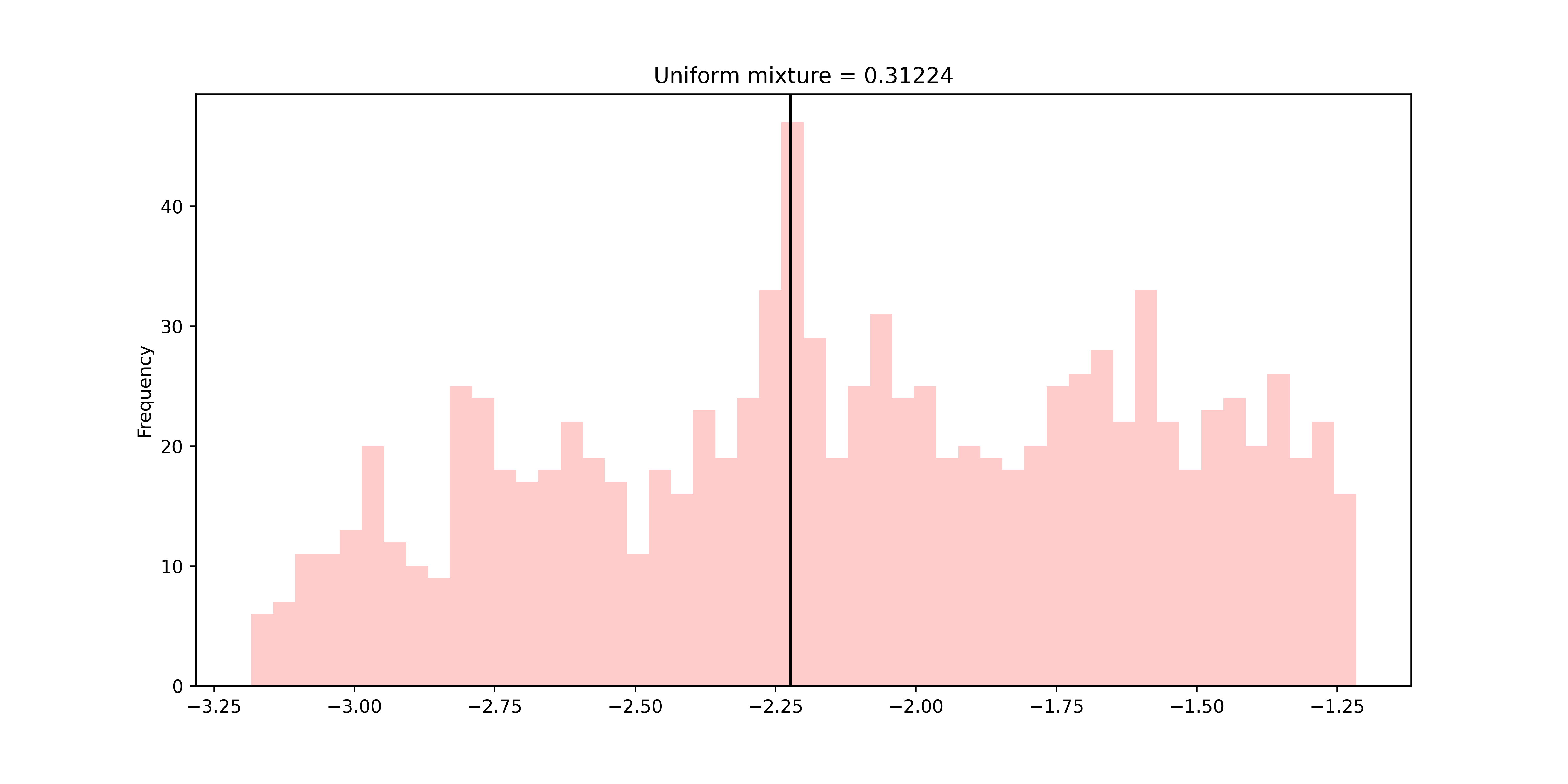} 
    \includegraphics[width=0.49\columnwidth]{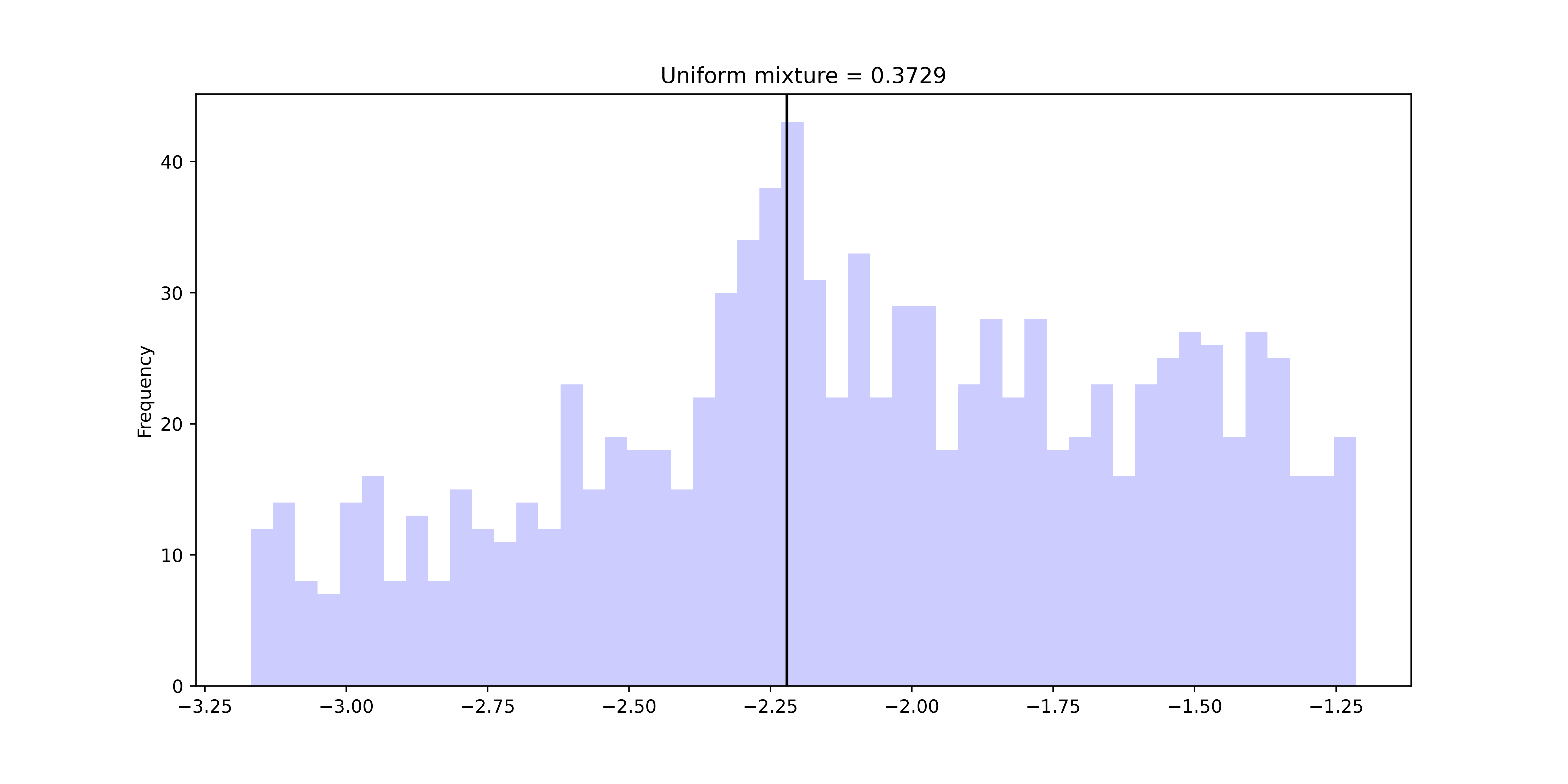}\\
    \includegraphics[width=0.49\columnwidth]{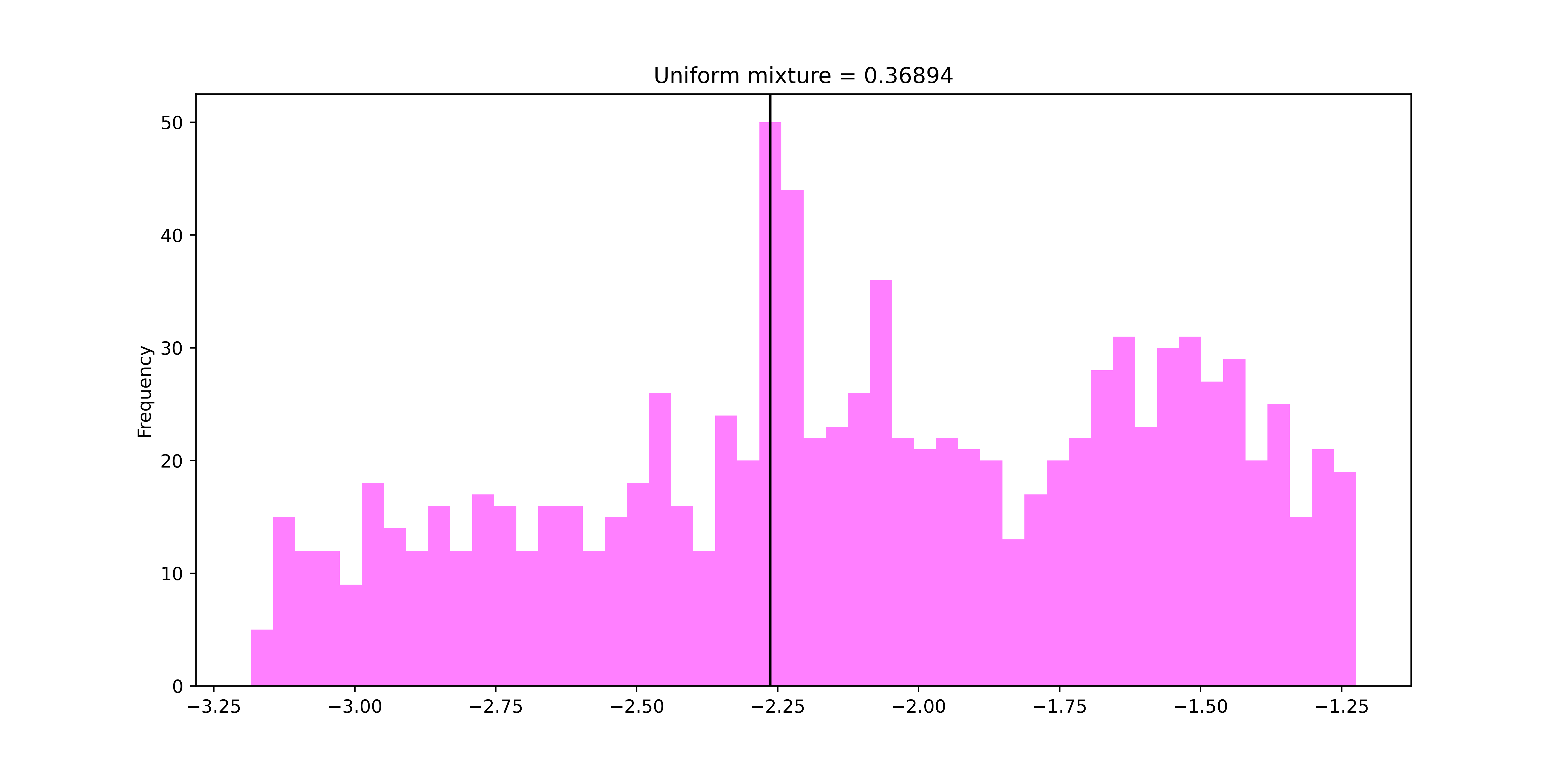}    
    \includegraphics[width=0.49\columnwidth]{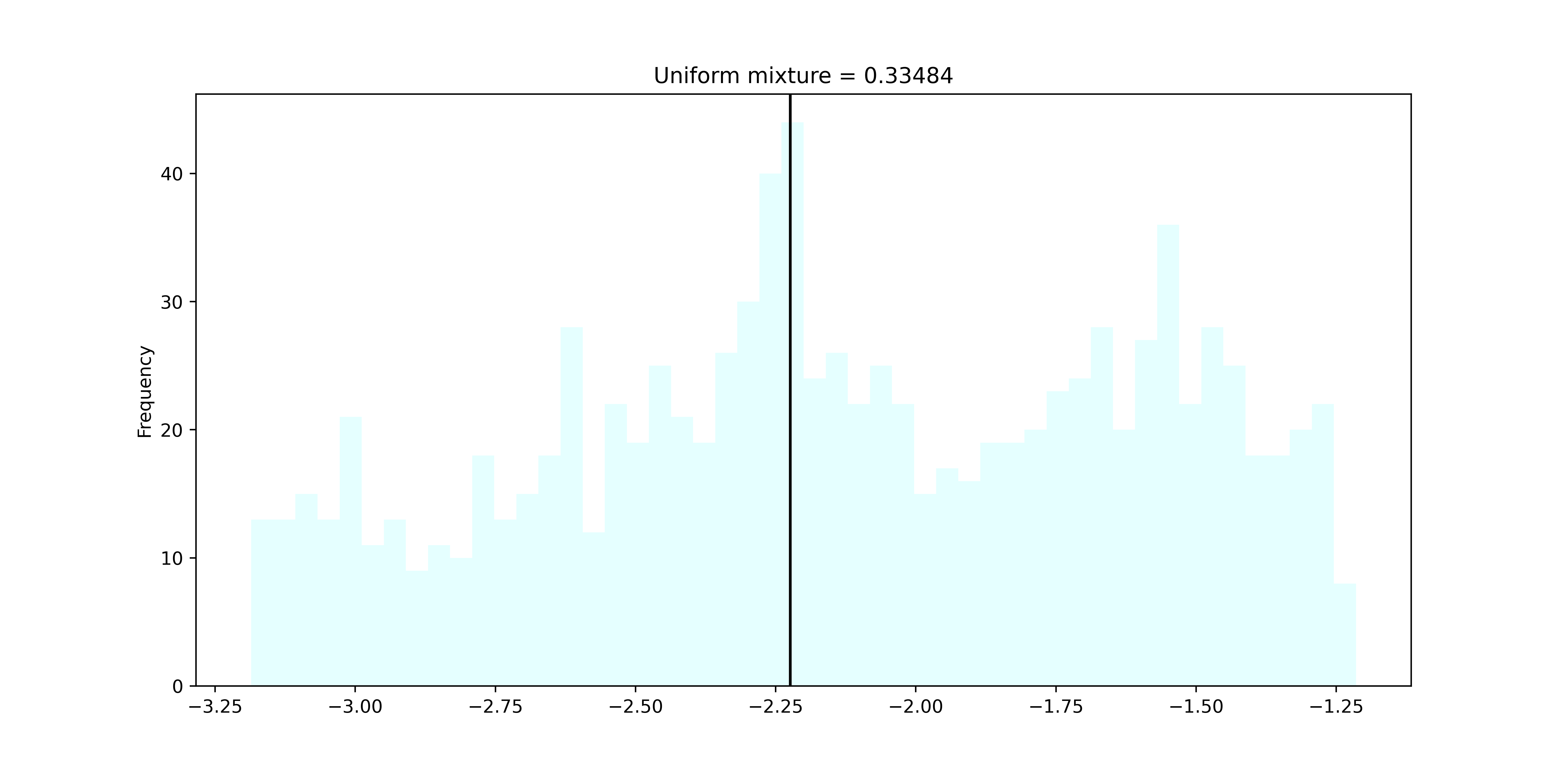}
    \caption{Four empirical distributions of the objective function values for four values of $\tilde{b}$. Clock-wise from the top-left: $\tilde{b}$: 0.01, 0.02, 0.03, 0.04, respectively, that yield $\hat{b} = 0.31224, 0.32779, 0.33484, 0.36894$. 
    The modes, indicated by the vertical black lines, are located at $-2.224,-2.221,-2.2635,-2.219$.
    The colours match those of Fig. \ref{fig:my_spsa}, except for 30\% transparency. 
    For comparison, see Fig. \ref{fig:objective_data2}, which suggests that the mean of the normally-distributed 
    component of the mixture is $-1.95$.}
    \label{fig:objective_data}
\end{figure}

\begin{figure}[!tbp]
    \centering
    \includegraphics[scale=0.5]{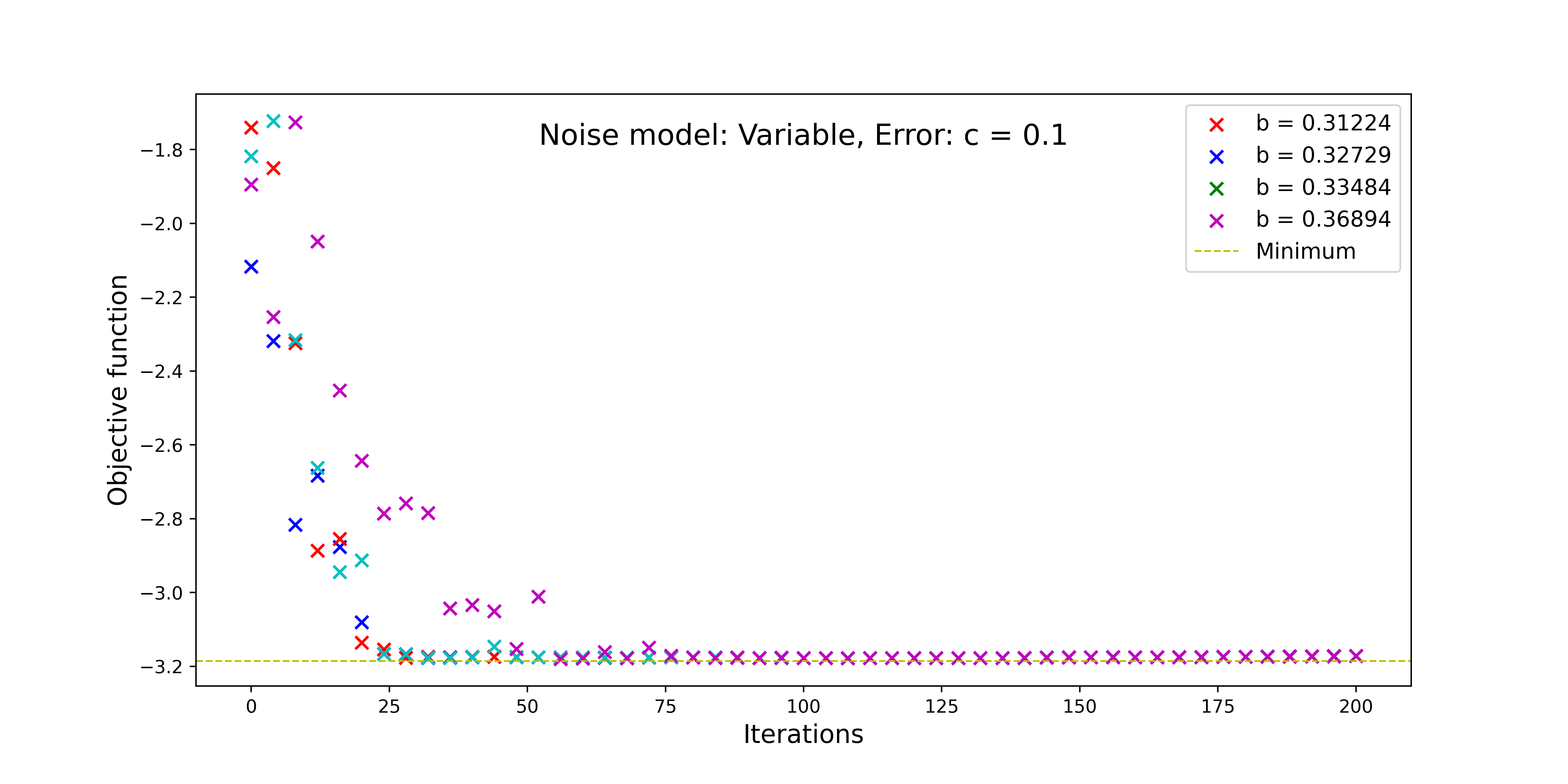}
    \caption{The iteration complexity of SPSA for four different values of bias $\hat b$ with error $c = 0.01$. }
    \label{fig:my_spsa}
\end{figure}

\begin{figure}[!htb]
    \centering
    \includegraphics[scale=0.50]{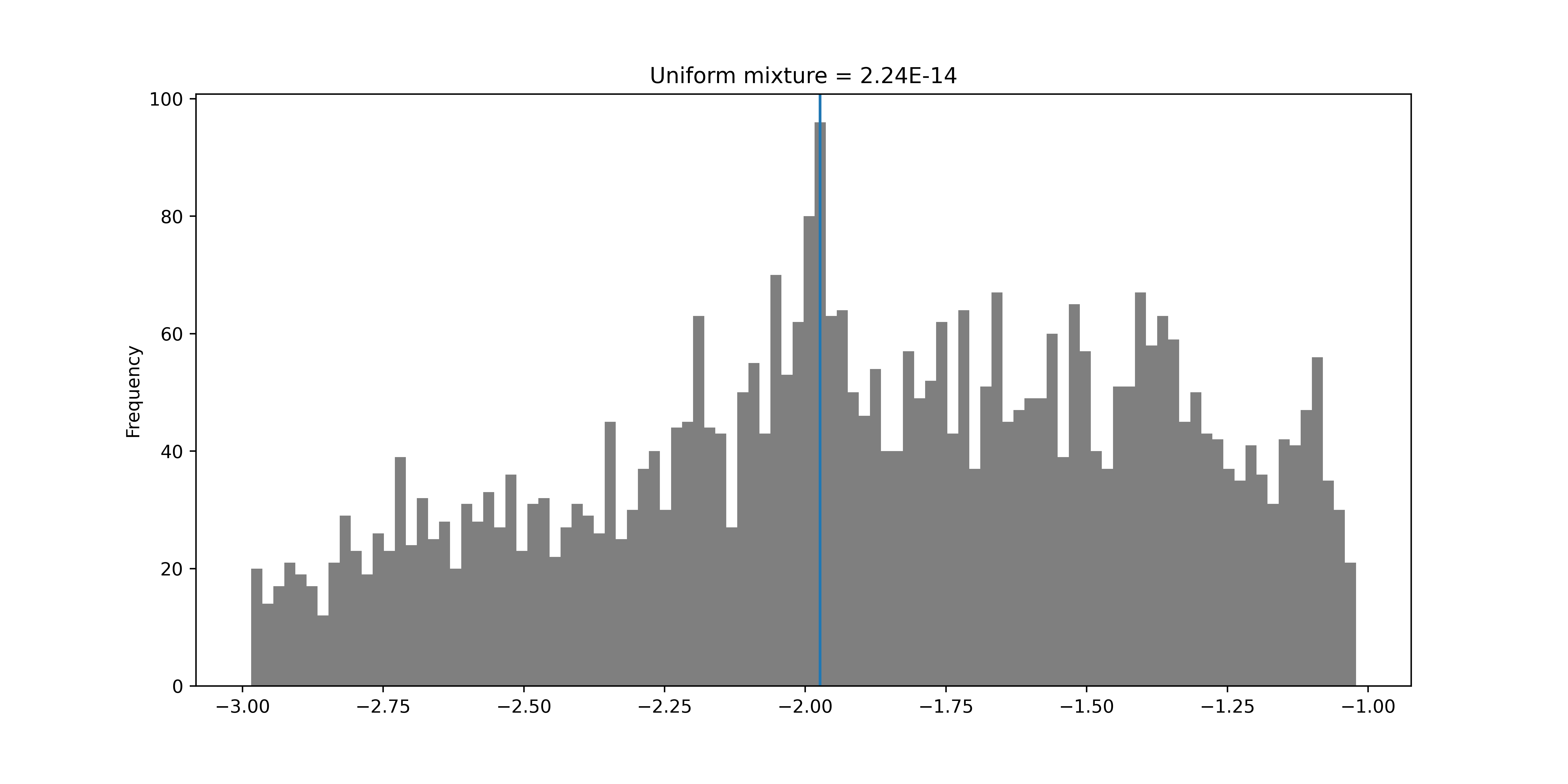}
    \caption{The empirical distribution of objective function values from 4096 evaluations in the ideal read-out scenario, i.e., where $P = \textrm{diag}(\mathds{1}_{16})$. 
    Our estimate of the bias is $2.24\times 10^{-14}$ with an empirical observation for the mean of the normal distribution at $-1.95$ and the mode of the mixture at $-1.9737$.}
    \label{fig:objective_data2}
\end{figure}

The empirical results in Fig.\ref{fig:my_spsa} seem to corroborate the analytical results of Thm.~\ref{th:spsa}. As expected from Thm.~\ref{th:spsa} and Fig.\ref{fig:boundsurface1}, the iteration complexity grows with the bias. 

{
\subsection{Comparison Between SPSA, 2PFA, and PSR}
Now we shall compare SPSA, the two point function approximation and the PSR as zero order gradient approximation schemes. For the experiment we chose 10 random starting points, and then started 20 trials for all three algorithms across parameter values $c=1e-1,1e-2,1e-3$ and the stepsize $\alpha=1e-1,1e-2,1e-3$. We are interested in studying the average final output value, the number of iterations until reaching convergence to the stationary point or noise dominated region, and the variability of the runs with respect to these criteria. To perform this analysis, we consider four different regimes with respect to small or large $c$ and small or large $\alpha$. Below we show figures that are representative of the results obtained. Note that the perturbation is not a parameter in PSR. We used the Lima noise model for all experiments. 
\subsubsection{Large Stepsize, Large Perturbation}
With this regime, we would expect fast convergence, however, to a fairly high variance noise-dominated region. In addition, monotonic decrease should be consistently seen, but perhaps not as steep as otherwise, as we expect high perturbation gradient estimates to be robust but imprecise. 

\begin{figure}[!tbp]
    \centering
    \includegraphics[scale=0.4]{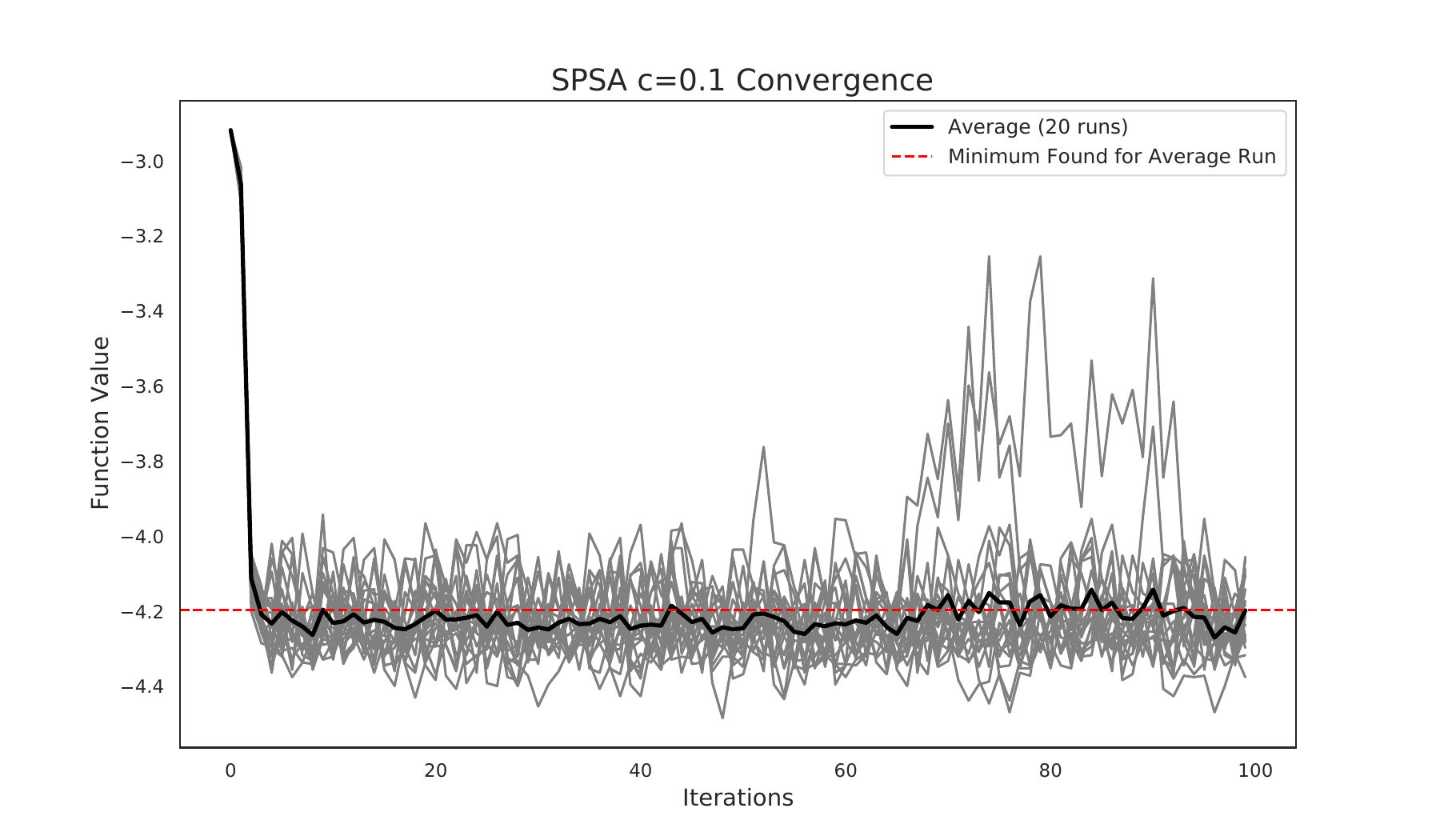}    \includegraphics[scale=0.4]{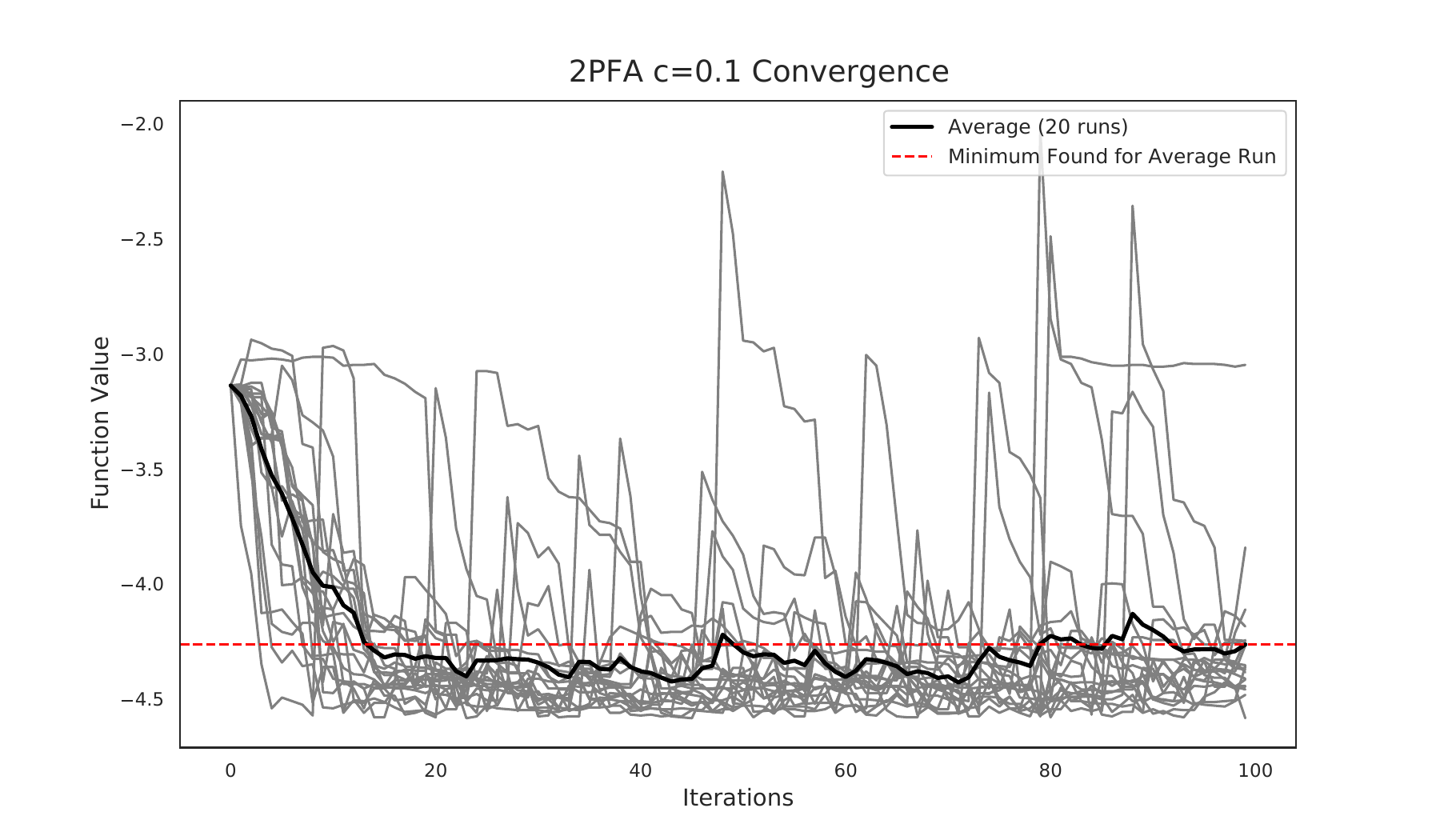}
\includegraphics[scale=0.4]{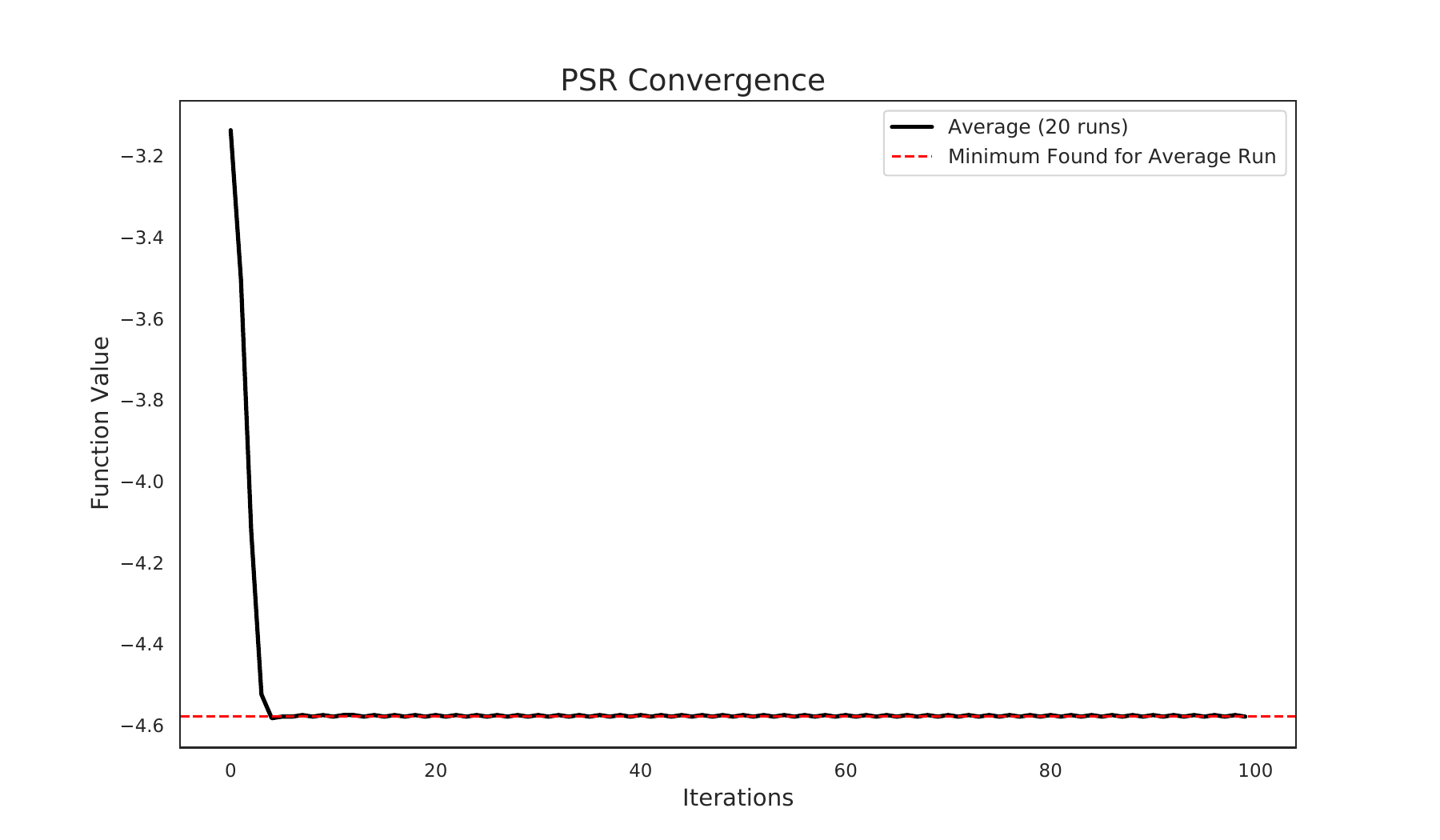}
    \caption{Results for SPSA and 2PFA for a representative run with $c=0.1$ and $\alpha=0.1$ and for PSR with $\alpha=0.1$.}
    \label{fig:bigcbigacompare}
\end{figure}
In Figure~\ref{fig:bigcbigacompare} we observe:
\begin{itemize}
    \item Impressively, even with an imprecise estimate with a high value of $c$, the speed of convergence of SPSA is the same as what can be considered the exact gradient of PSR. However, as expected the convergence is to a relatively large noise dominated region with a high average bias in the former case. By comparison 2PFA is both slower to converge in terms of iterations, while also being significantly less robust, with many late iterations even surpassing the objective value of the starting point.   
\end{itemize}


\subsubsection{Large Stepsize, Smaller Perturbation}
This regime indicates aggressive optimization towards a minimum, together with a precise but noisier estimate of the gradient, with the intention of finding a stationary point with lower bias. 
\begin{figure}[!tbp]
    \centering
    \includegraphics[width=0.49\columnwidth]{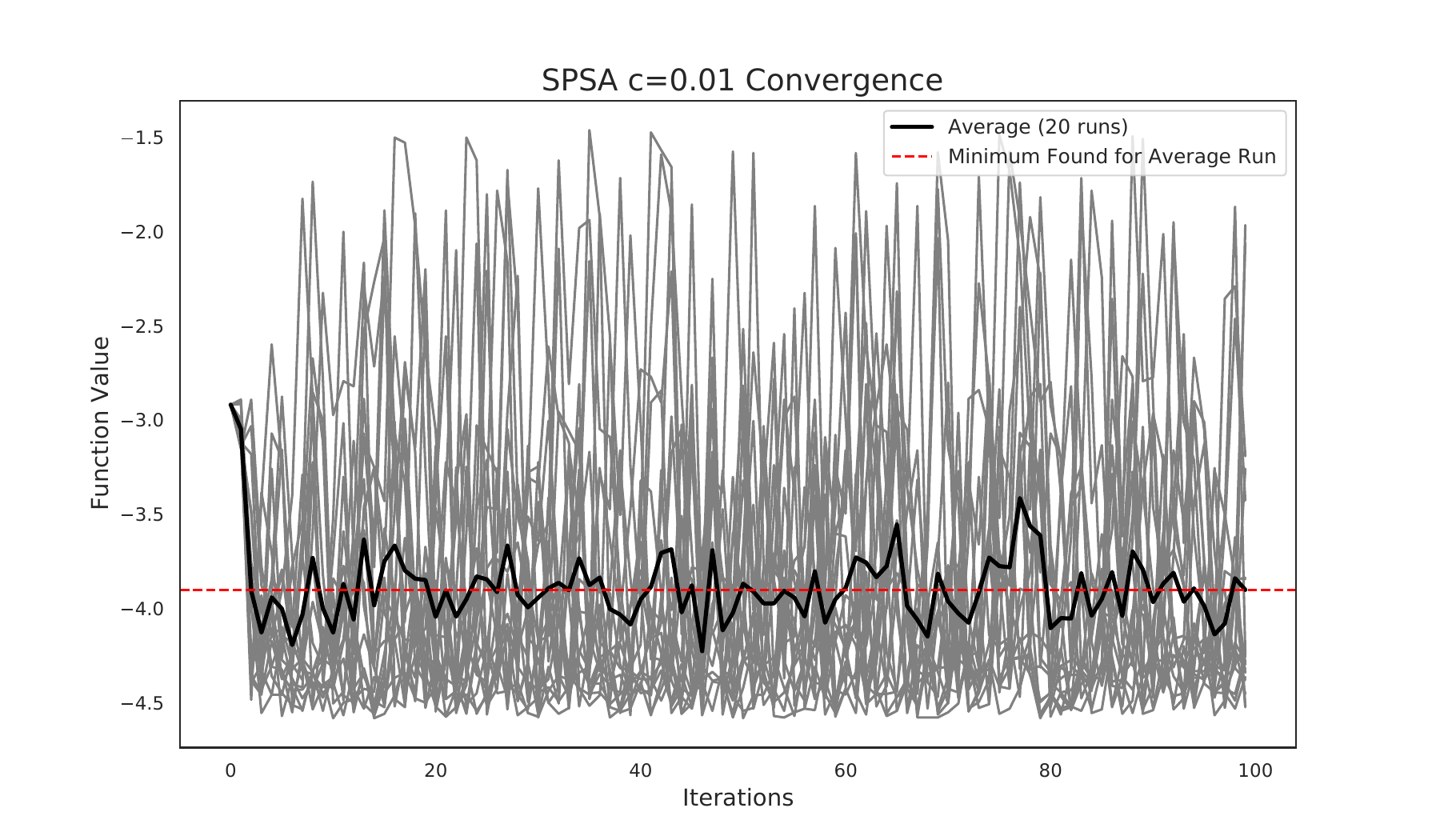}    \includegraphics[width=0.49\columnwidth]{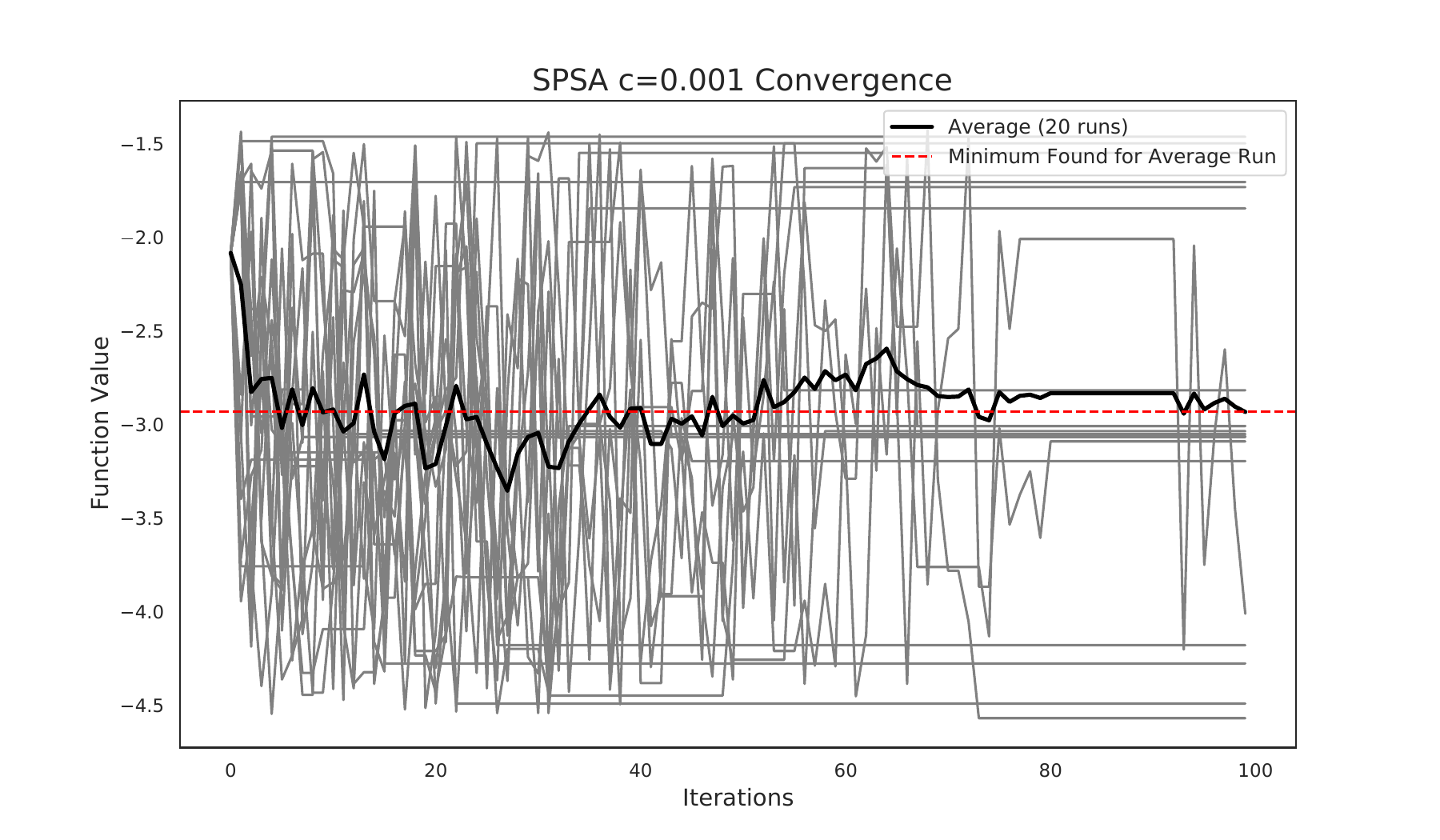} \\
\includegraphics[width=0.49\columnwidth]{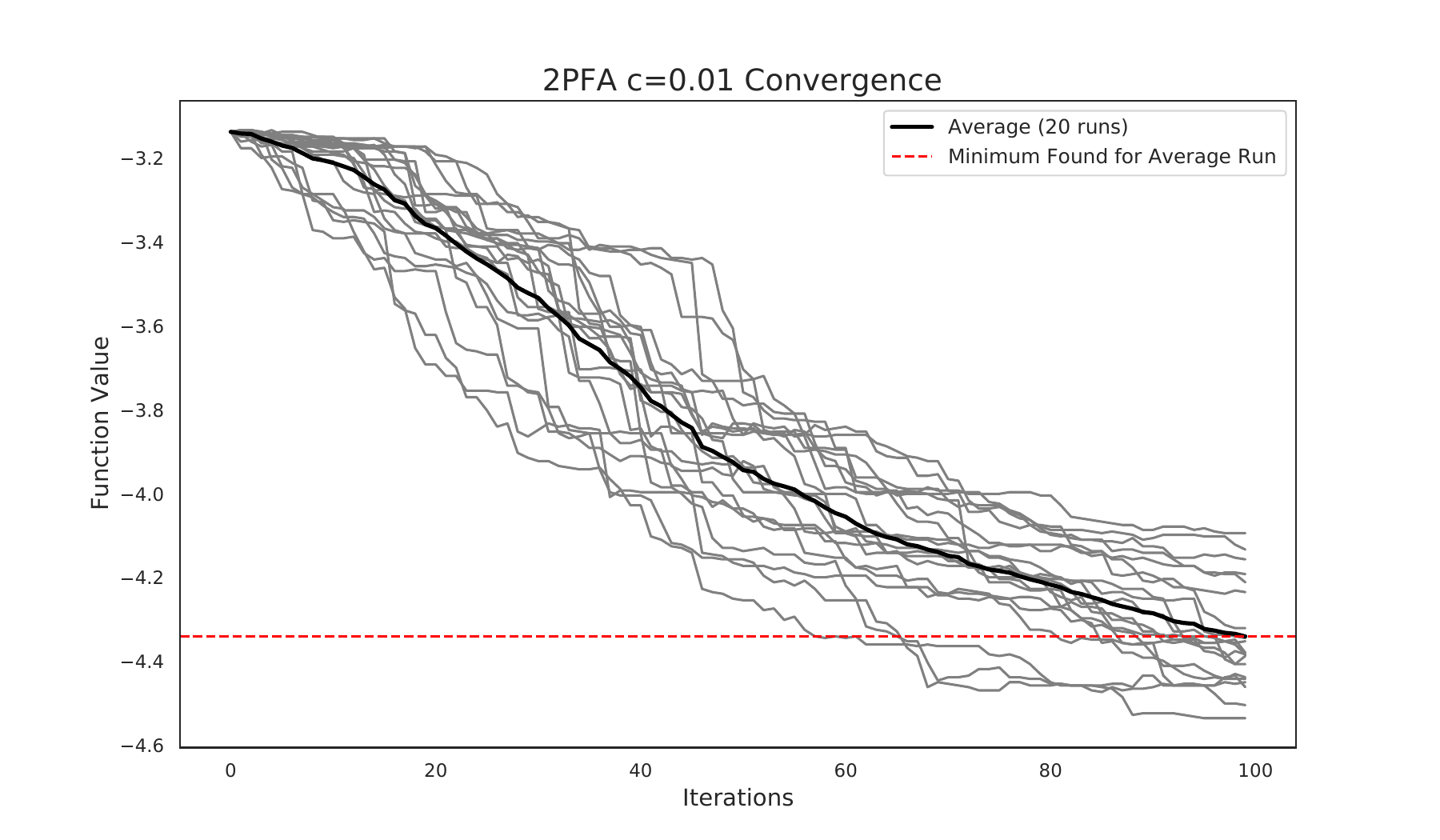}
\includegraphics[width=0.49\columnwidth]{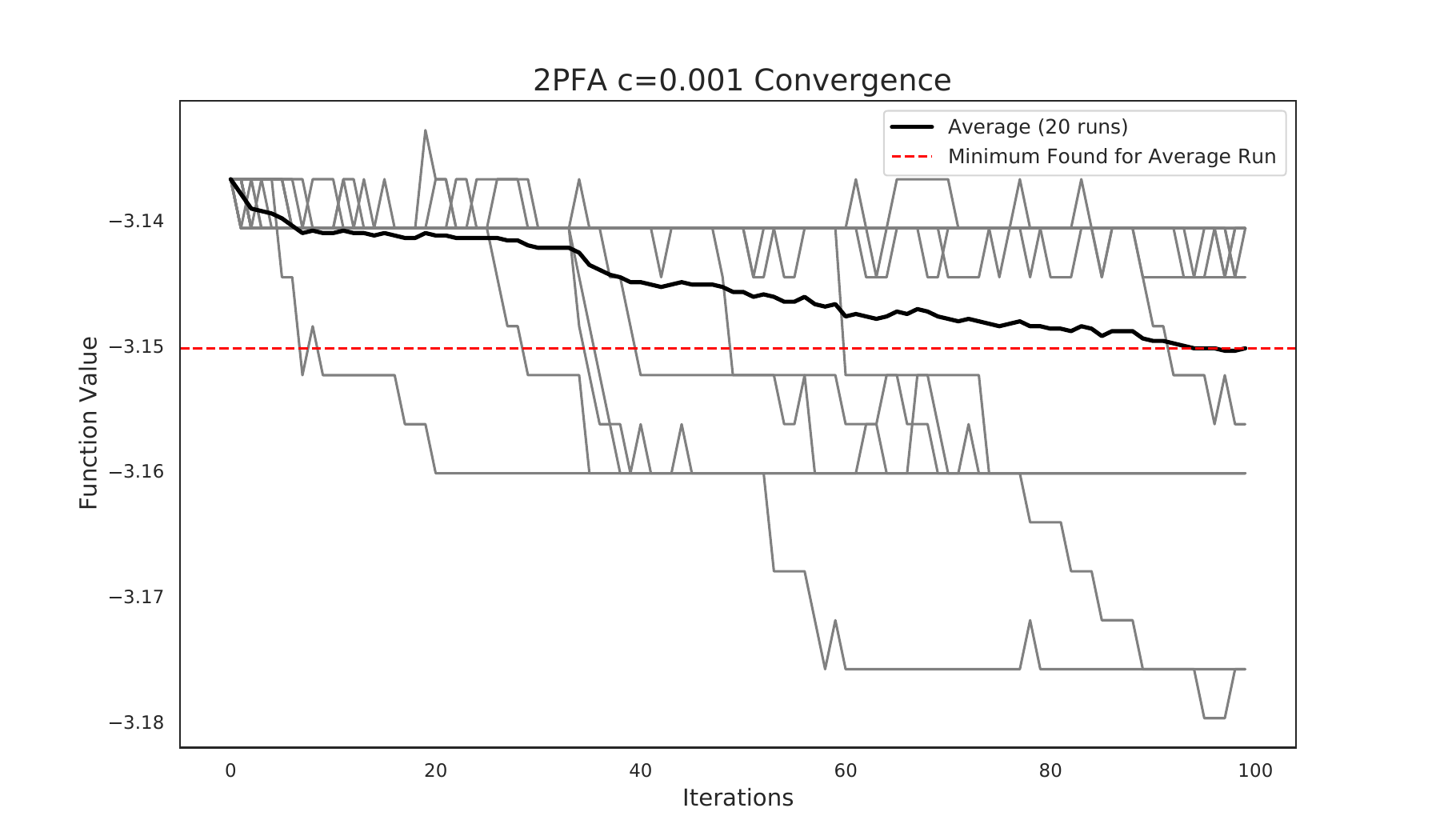}
    \caption{Results for SPSA and 2PFA for a representative run with $c=0.01,0.001$, and $\alpha=0.1$.}
    \label{fig:smallcbigacompare}
\end{figure}
The plots in Figure~\ref{fig:smallcbigacompare} are interesting and communicate some important subtleties in regards to the two zero order methods. We observe that SPSA is very noisy for even moderately small $c$ when the stepsize is large. However, 2PFA exhibits a steady and reliable, if slow, convergence with the moderately small perturbation to a fairly good objective value. Meanwhile it is completely useless with a higher variance estimate using $c=0.001$.

\subsubsection{Small Stepsize, Large Perturbation}
This regime is the ``safest'' in terms of obtaining a decent objective value reliably. A small stepsize means a smaller noise dominated region near optimality and more locally steps. A large perturbation results in a relatively low variance and but imprecise and high bias gradient.  
\begin{figure}[!tbp]
    \centering
    \includegraphics[width=0.49\columnwidth]{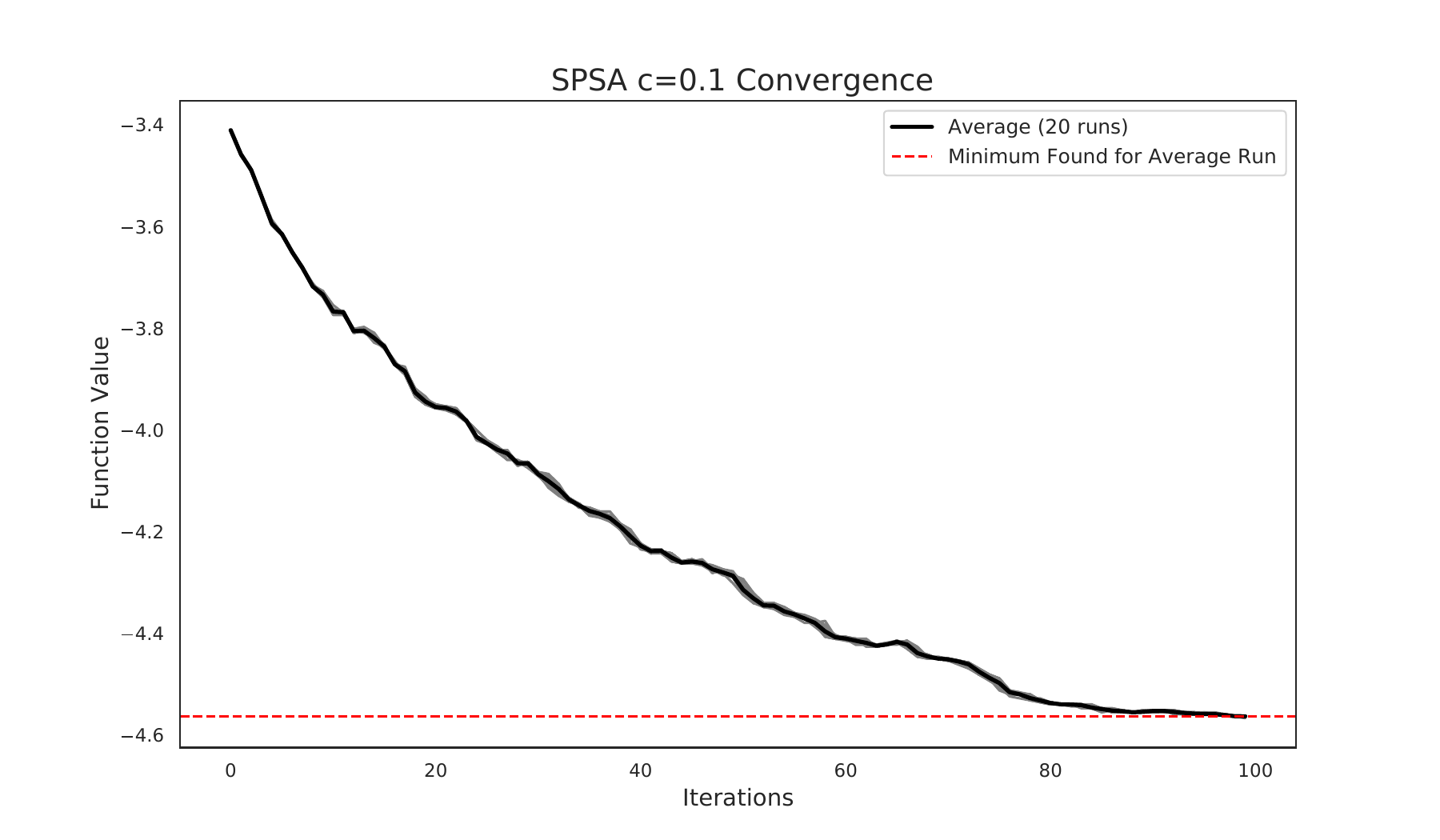}    \includegraphics[width=0.49\columnwidth]{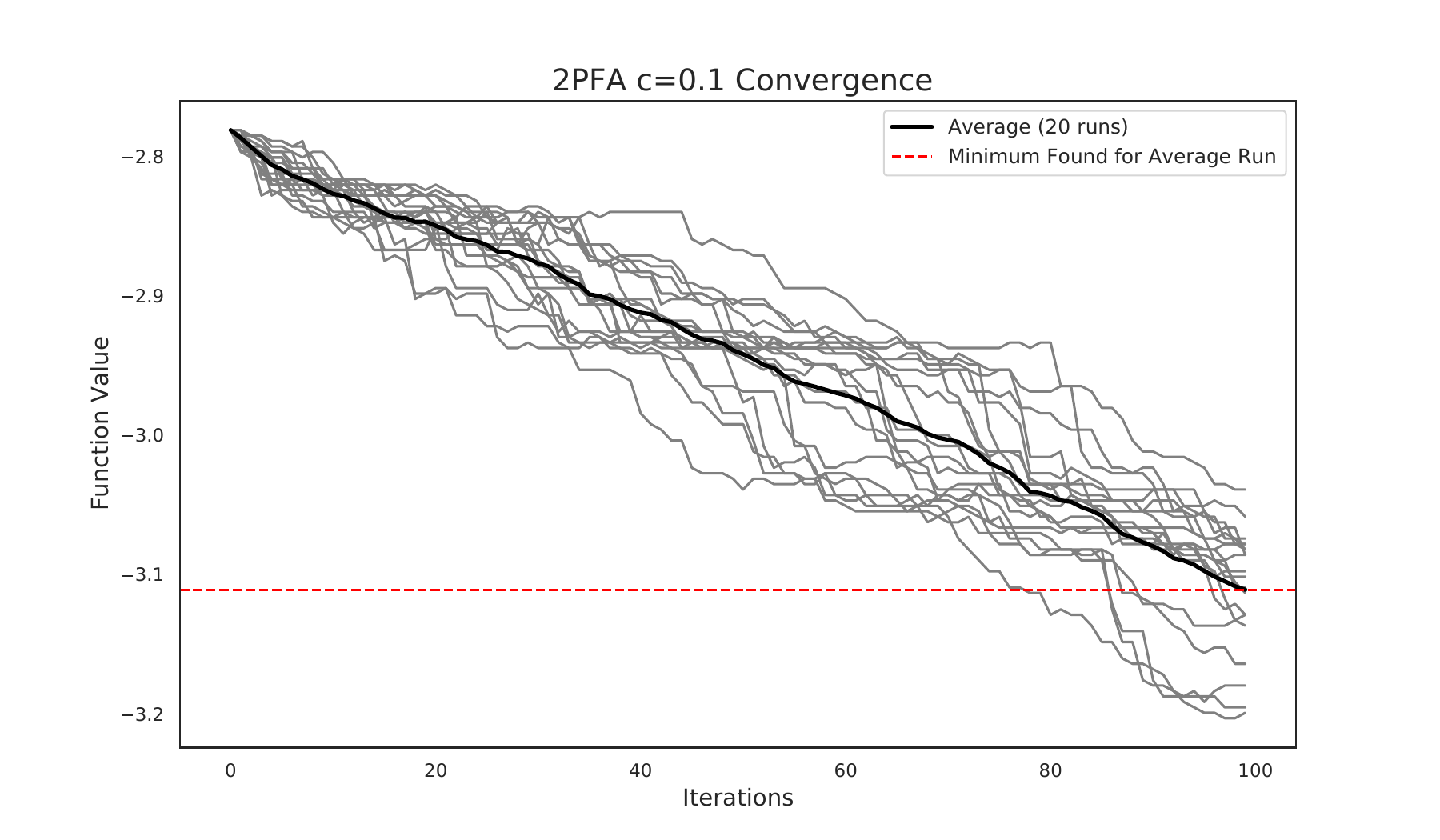}\\
\includegraphics[width=0.49\columnwidth]{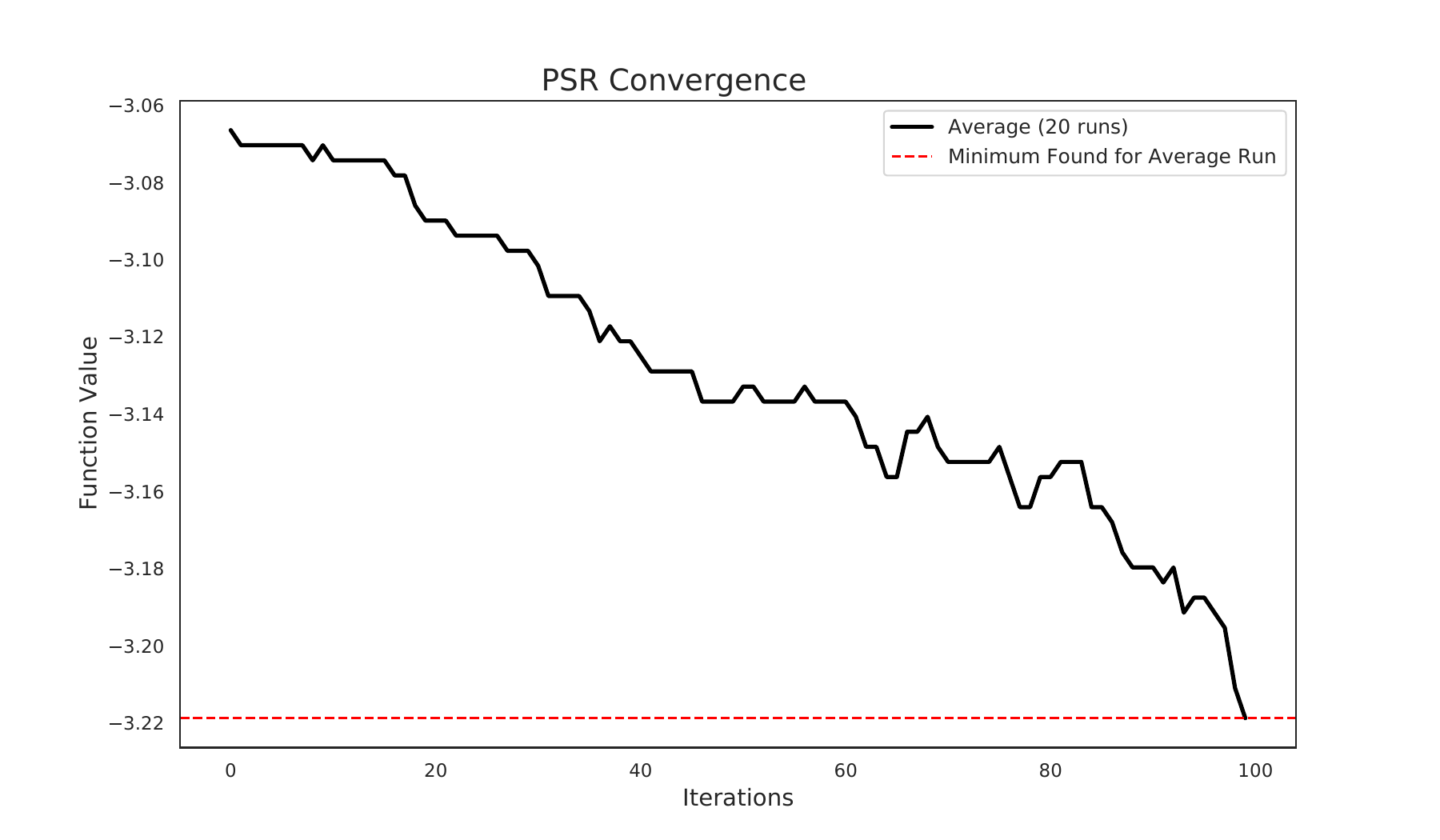}
    \caption{Results for SPSA and 2PFA for a representative run with $c=0.1$ and $\alpha=0.001$ and PSR with $\alpha=0.001$.}
    \label{fig:bigcsmallacompare}
\end{figure}
From Figure~\ref{fig:bigcsmallacompare} we observe, remarkably, that SPSA converges faster than even the PSR in the safe regime. Otherwise, all three algorithms exhibit steady decrease with the iterations, with 2PFA being the noisiest of the algorithms.
\subsubsection{Medium Stepsize, Small Perturbation}
This regime indicates a steady decrease, but with a precise gradient estimate, thus obtaining a final objective value lower in mean and higher in variance. 
\begin{figure}[!tbp]
    \centering
    \includegraphics[width=0.49\columnwidth]{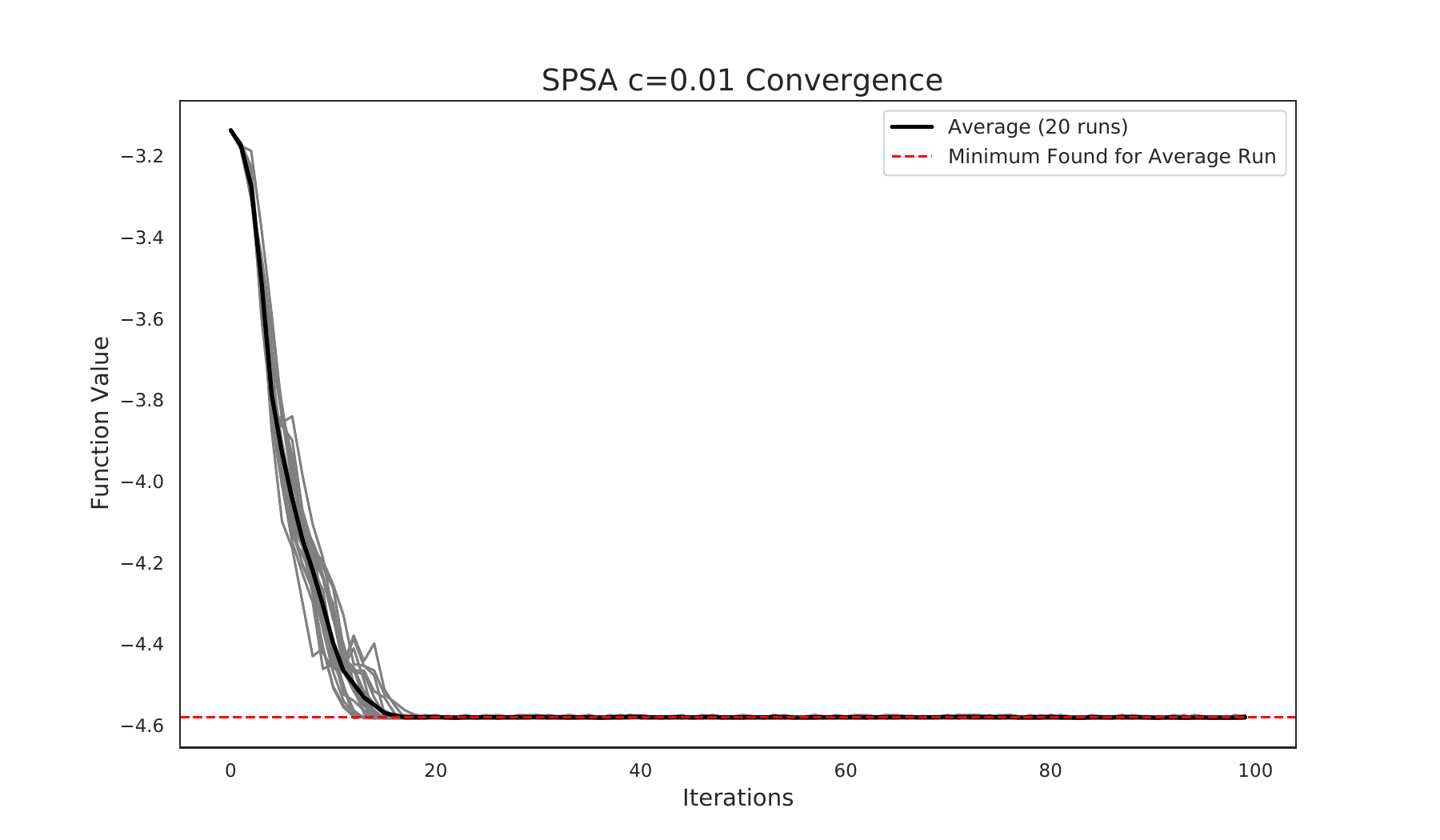}    \includegraphics[width=0.49\columnwidth]{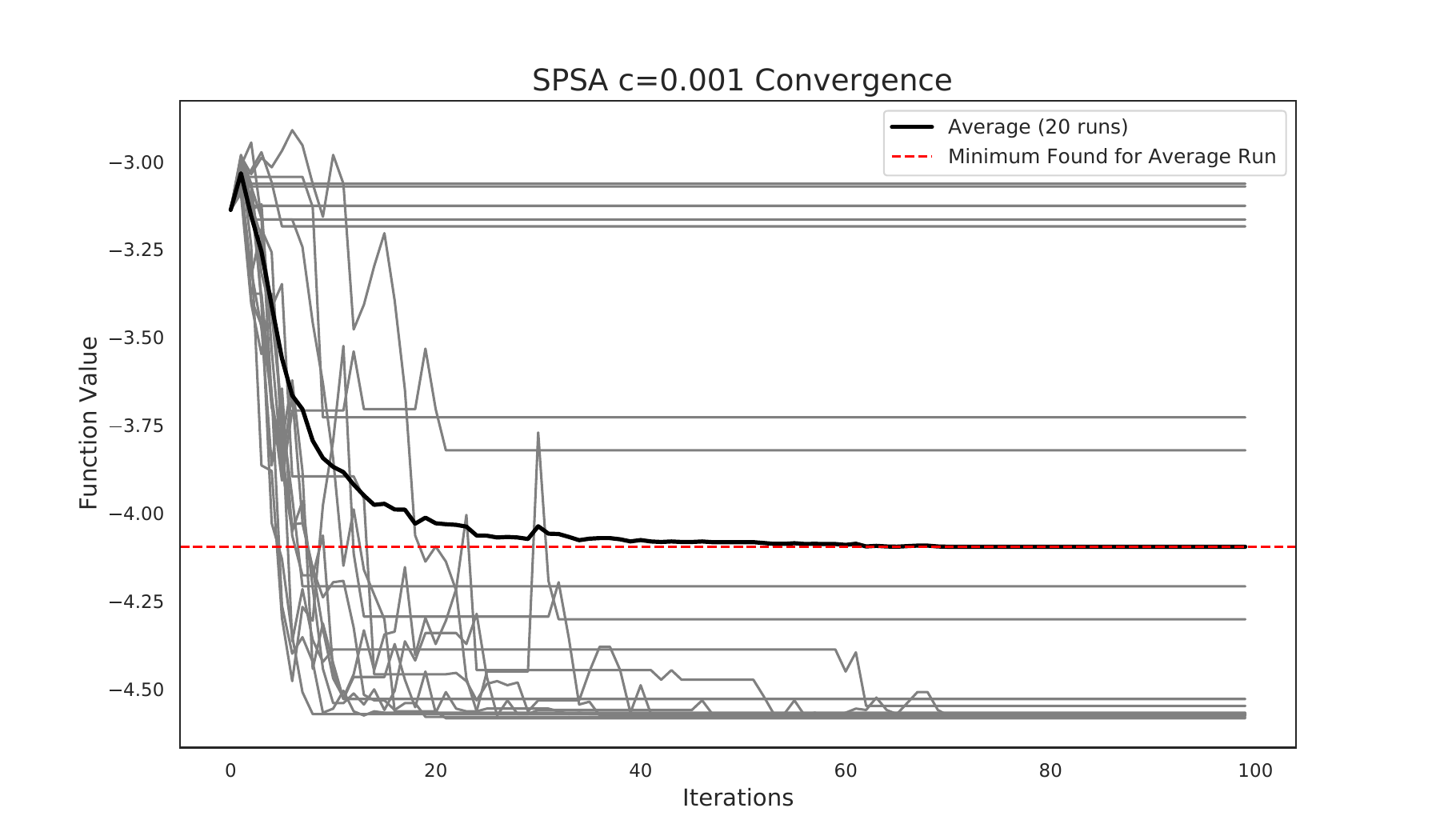}\\
\includegraphics[width=0.49\columnwidth]{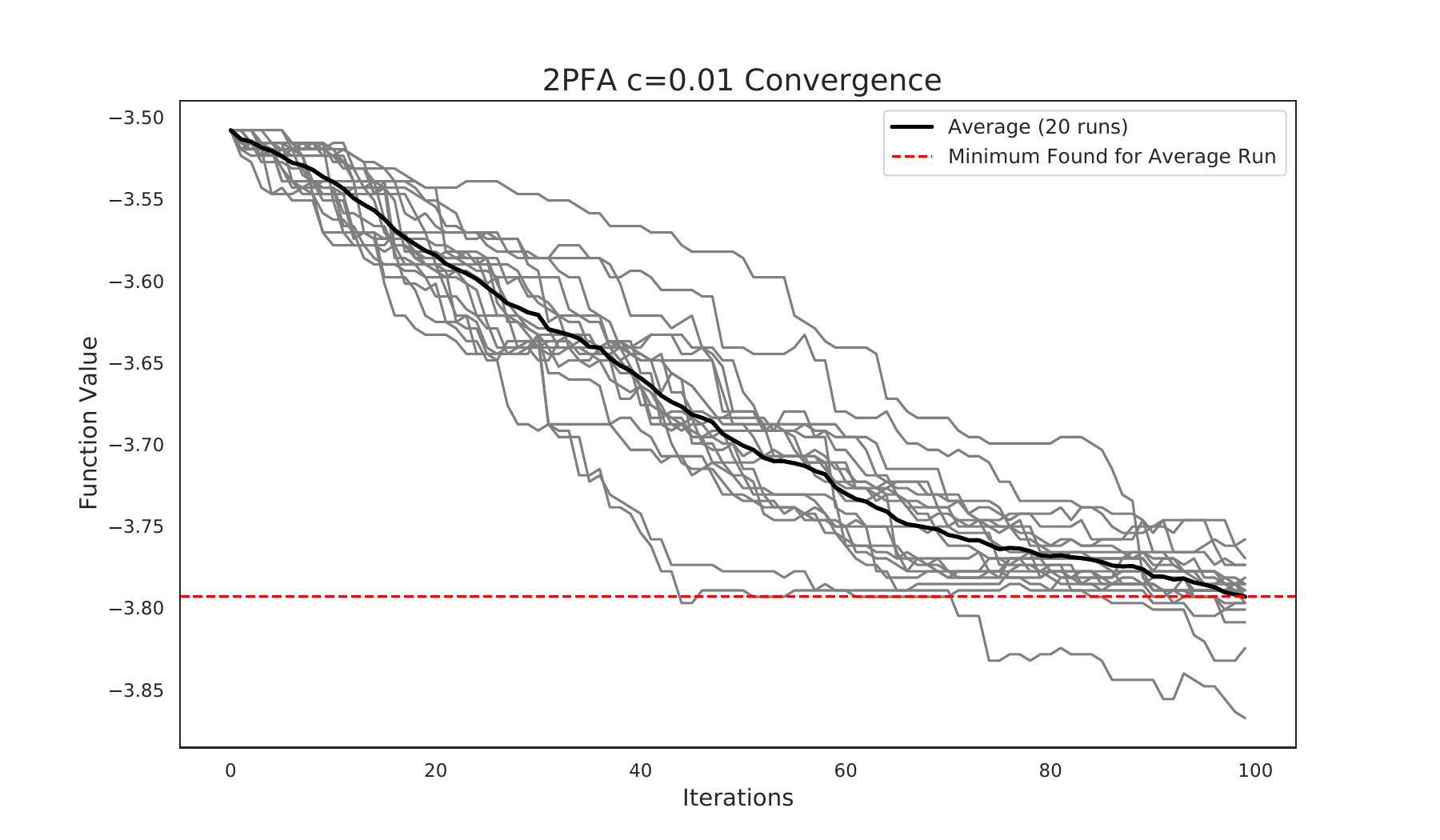}
\includegraphics[width=0.49\columnwidth]{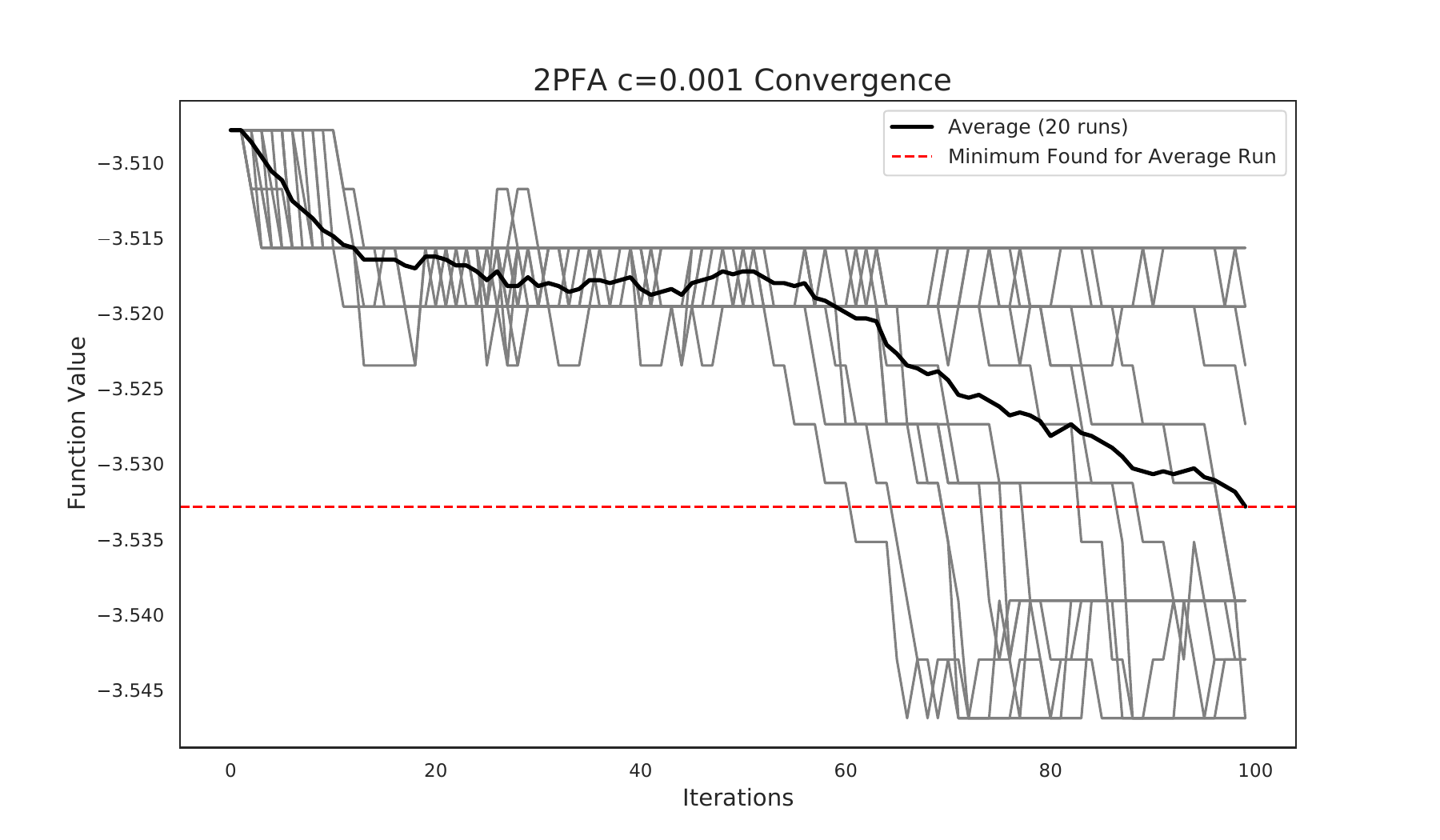}
    \caption{Results for SPSA and 2PFA for a representative run with $c=0.01,0.001$ and $\alpha=0.1$.}
    \label{fig:smallcsmallacompare}
\end{figure}
Finally, Figure~\ref{fig:smallcsmallacompare} indicates that with a small perturbation, the variance of the two point function approximation estimates practically eliminates any potential advantage of the steady reliable decrease of a smaller stepsize. For a more moderate perturbation, the convergence is both slow and noisy, albeit at least relatively steady. By contrast SPSA clearly has an optimal threshold, but with a small value of $c$, some iterates still obtain a good quality objective value. The top Figure indicates particularly strong performance in speed and reliability, together with the lowest bias seen among the comparisons. 
}


\section{A Discussion and Conclusion}\label{sec:conc}

We have analyzed the convergence guarantees of variational quantum algorithms in the presence of bias in the quantum part. Let us interpret these results. 

In regards to the asymptotic bias in the approximate measure of stationarity, we can notice that, predictably, the greater the bias in the function evaluation, the greater the error lower bound in the ultimate optimality criterion. A coarser discretization for estimating the gradient appears to stabilize the effect of the function evaluation bias but carries the trade-off in introducing additional asymptotic optimality bias in the overall rougher accuracy in gradient estimation. In addition, as typical for zero-order methods, the asymptotic bias scales poorly as a function of the problem dimension.

With respect to the iteration complexity, we can see that the overall rate is not affected by the presence of bias. The rate of convergence is the same as in a standard stochastic gradient analysis. That said, for the two terms in the convergence guarantee that do drop to zero with more iterations, exhibit a larger constant with greater bias.

More broadly, this is a small step towards a better understanding of the performance of variational quantum algorithms. Until recently, these have been seen as heuristics with no guarantees whatsoever. This and recent related work \cite[e.g.]{Egger2021warmstartingquantum} suggests that their performance can be characterized in some detail.

\paragraph*{Acknowledgement}
This work is a collaboration between the Fidelity Center for Applied Technology, Fidelity Labs, LLC., and the Czech Technical University in Prague. 

\bibliographystyle{quantum}
\bibliography{refs}
\FloatBarrier



\clearpage
\appendix
\section{Appendix: Proofs of Theoretical Results}
\label{app:proofs}

To begin with, we prove a bound on the error associated with the two point function approximation technique:
\begin{lemma}\label{lem:accuracy}
Under Assumption~\ref{as:boundbias}, the two point estimate~\eqref{eq:twopointNesterov} satisfies, for all $\bs\theta$
\[
\bs g(\bs\theta,\xi^1,\xi^2)=\nabla f_{c}(\bs\theta)+\bs r(\xi^1,\xi^2)+\bs b_c(\bs\theta),\,\mathbb{E}[\bs r(\xi^1,\xi^2)]=\bs 0,\,\|\bs b_c(\bs\theta)\|\le \frac{2b}{c}
\]
\end{lemma}
\begin{proof}
\[
\begin{array}{l}
\mathbb{E}[\bs g] = \mathbb{E}_{\bs\Delta} \left[
\frac{\mathbb{E}_{\xi^1}[F(\bs\theta+c\bs \Delta,\xi^1)]-\mathbb{E}_{\xi^2}[F(\bs\theta,\xi^2)]}{c}\bs \Delta\right] \\
\quad = \mathbb{E}_{\bs \Delta} \left[
\frac{f(\bs\theta+c\bs \Delta)+\mathbb{E}_{\xi^1}[n(\bs\theta+c\bs \Delta,\xi^1)]-f(\bs \theta)-\mathbb{E}_{\xi^2}[n(\bs \theta,\xi^2)]}{c}\bs \Delta\right] \\ 
\quad = \nabla f_{c} \bs \theta)+\mathbb{E}_{\bs \Delta} \left[
\frac{\mathbb{E}_{\xi^1}[b(\bs \theta+c\bs \Delta,\xi^1)]-\mathbb{E}_{\xi^1}[b(\bs \theta,\xi^2)]}{c}\bs \Delta\right]
\end{array}
\]
where we have used the definition of $\nabla f_{c}$ as given in~\eqref{eq:smoothedfunc} and the functional form of the bias given in~\eqref{eq:biasas}.

Now we can bound the second term,
\[
\begin{array}{l}
\left\|\mathbb{E}_{\bs \Delta} \left[
\frac{\mathbb{E}_{\xi^1}[b(\bs \theta+c\bs \Delta,\xi^1)]-\mathbb{E}_{\xi^1}[b(\bs \theta,\xi^2)]}{c}\bs \Delta\right]\right\|  \le \frac{2}{(2\pi)^{\frac{n}{2}}}\int_0^{\infty}\frac{\|b(\cdot)\|_{\infty}+\|b(\cdot)\|_{\infty}}{c}|\bs \Delta| d\bs \Delta \le 2b/c
\end{array}
\]
where we used the symmetry of the Gaussian. 
\end{proof}

Now we can show the main result, which we restate for reference. The proof follows standard arguments, e.g.,~\cite[Theorem 3.2]{ghadimi2013stochastic}.
\begin{theorem}
Consider the two point function approximation algorithm with the iteration~\eqref{eq:saiter} for a problem satisfying Assumptions~\ref{as:boundbias} and~\ref{as:lipgrad}. 

\begin{enumerate}
    \item 
Consider a fixed budget of $K$ iterations. Let $\alpha_k\equiv\alpha/\sqrt{K}$ with $\alpha\le L/6$. It holds that the iterates satisfy:
\begin{equation}\label{eq:totalconvfixed0}
   \min\limits_{k\in[K]} \mathbb{E}\|\nabla f(\bs \theta_k)\|^2 \le 4\frac{f_c(\bs \theta_k)-f^*_{c}}{\alpha\sqrt{K}} +2\frac{L\alpha(\sigma^2+\frac{8b^2}{c^2})}{\alpha\sqrt{K}} +\left[\frac{2b^2}{c^2}+\frac{c^2}{2}L(p+3)^3\right]
\end{equation}

\item Consider a diminishing stepsize of $\alpha_k=\alpha/k^{\gamma}$ with $\gamma\in (1/2,1]$ and $\alpha\le L/2$. Then it holds that, taking,
\[
\bs \theta_R\sim \bs \theta_k\text{ w.p. } \alpha_k/\sum\limits_{l=1}^K \alpha_l
\]
we have,
\begin{equation}\label{eq:totalconvdim0}
   \mathbb{E}\|\nabla f(\bs \theta_R)\|^2 \le 4\frac{f_c(\bs \theta_k)-f^*_{c}}{\alpha K^\gamma} +2\frac{L\alpha(\sigma^2+\frac{8b^2}{c^2})}{\alpha K^{\gamma}} +\left[\frac{2b^2}{c^2}+\frac{c^2}{2}L(p+3)^3\right]
\end{equation}
\end{enumerate}
\end{theorem}
\begin{proof}
We begin with the standard Descent Lemma on $f_{c}(\bs \theta)$, recalling that $f_{c}$ has a Lipschitz gradient with the Lipschitz constant bounded by $L$ and proceed with expressing the step as the composition given in Lemma~\ref{lem:accuracy}.
\[
\begin{array}{l}
f_c(\bs \theta_{k+1}) \le f_c(\bs \theta_k)-\alpha_k \langle \nabla f_c(\bs \theta_k),\bs g_k\rangle +\frac{L}{2}\alpha_k^2 \|\bs g_k\|^2 \\
\quad = f_c(\bs \theta_k)-\alpha_k\|\nabla f_c (\bs \theta_k)\|^2-\alpha_k \langle \nabla f_c(\bs \theta_k),\bs r(\xi^1_k,\xi^2_k)+\bs b_c(\bs \theta_k)\rangle +\frac{L}{2}\alpha_k^2 \|\bs g_k\|^2
\end{array}
\]
Taking conditional expectations with respect to the Filtration $\mathcal{F}_k$ corresponding to $(\xi_k^1,\xi_k^2)$, and summing up, we get, with $f^*_c$ a lower bound for $f_c\bs \theta)$,
\begin{equation}\label{eq:sumdesc}
    \sum\limits_{k=1}^K \alpha_k \mathbb{E}\left[\|\nabla f_c(\bs \theta_k)\|^2|\mathcal{F}_k\right] \le f_c(\bs \theta_k)-f^*_{c} - \sum\limits_{k=1}^K \mathbb{E}\left[\alpha_k \langle \nabla f_c(\bs \theta_k),{\bs r}(\xi^1_k,\xi^2_k)+{\bs b}_c(\bs \theta_k)\rangle|\mathcal{F}_k\right]+\frac{L}{2}\sum\limits_{k=1}^K \alpha_k^2 \mathbb{E}\left[\|\bs g_k\|^2|\mathcal{F}_k\right]
\end{equation}
Now, we have, from Lemma~\ref{lem:accuracy} and Assumption~\ref{as:boundbias}, and noting that cross terms involving $r(\xi^1_k,\xi^1_k)$ and deterministic quantities are zero by $\mathbb{E}[r(\xi^1_k,\xi^2_k)|\mathcal{F}_k]=0$,
\[
\mathbb{E}[\|\bs g_k\|^2|\mathcal{F}_k]\le 2\|\nabla f_c(\bs \theta_k)\|^2+\mathbb{E}[\|r(\xi^1_k,\xi^2_k)\|^2|\mathcal{F}_k]+2\|b_c(\bs \theta_k)\|^2 \le 2\|\nabla f_c(\bs \theta_k)\|^2+\sigma^2+\frac{8b^2}{c^2}
\]
Furthermore, using that $\mathbb{E}[r(\xi^1_k,\xi^2_k)|\mathcal{F}_k]=0$ and Young's inequality,
\[
-\mathbb{E}[\langle \nabla f_c(\bs \theta_k),r(\xi^1_k,\xi^2_k)+b_c(\bs \theta_k)\rangle|\mathcal{F}_k] = -\langle \nabla f_c(\bs \theta_k),b_c(\bs \theta_k)\rangle \le \frac{\|\nabla f_c(\bs \theta_k)\|^2}{2}+\frac{\|b_c\|^2}{2} \le \frac{\|\nabla f_c(\bs \theta_k)\|^2}{2}+\frac{2b^2}{c^2}
\]
Plugging in the previous two bounds into~\eqref{eq:sumdesc} and then
taking total expectations and using the expectation tower property,
\[
\sum\limits_{k=1}^K \frac{\alpha_k}{2}\left(1-2\alpha_k L\right) \mathbb{E}\|\nabla f_c(\bs \theta_k)\|^2 \le f_c(\bs \theta_k)-f^*_{c} +\frac{L(\sigma^2+\frac{8b^2}{c^2})}{2}\sum\limits_{k=1}^K \alpha_k^2 +\frac{2b^2}{c^2}\sum\limits_{k=1}^K \alpha_k
\]
Note that we have, from~\cite[(3.9)]{ghadimi2013stochastic},
\[
\|\nabla f(\bs \theta)\|^2-\frac{c^2}{2}L(p+3)^3 \le 2 \|\nabla f_c\bs \theta)\|^2
\]
and so finally,
\begin{equation}\label{eq:totalconv}
   \sum\limits_{k=1}^K \frac{\alpha_k}{2}\left(1-2\alpha_k L\right) \mathbb{E}\|\nabla f(\bs \theta_k)\|^2 \le f_c(\bs \theta_k)-f^*_{c} +\frac{L(\sigma^2+\frac{8b^2}{c^2})}{2}\sum\limits_{k=1}^K \alpha_k^2 +\left[\frac{2b^2}{c^2}+\frac{c^2}{2}L(p+3)^3\right]\sum\limits_{k=1}^K \alpha_k
\end{equation}

We can now consider budget and diminishing step-size regimes. In the first case, we have a fixed budget of $K$ iterations and take $\alpha_k:=\frac{\alpha}{\sqrt{K}}$ with $\alpha\le L/4$. The above expression becomes
\begin{equation}\label{eq:totalconvfixed}
   \frac{1}{4}\sqrt{K} \alpha\min\limits_{k\in[K]} \mathbb{E}\|\nabla f(\bs \theta_k)\|^2 \le f_c(\bs \theta_k)-f^*_{c} +\frac{L\alpha(\sigma^2+\frac{8b^2}{c^2})}{2} +\alpha\left[\frac{2b^2}{c^2}+\frac{c^2}{2}L(p+3)^3\right]\sqrt{K}
\end{equation}
from which we derive the first result. 
\end{proof}

Now we first re-state then prove Theorem~\ref{th:spsa}, the main convergence result for SPSA under biased noisy function evaluations.
\begin{theorem}
Consider the SPSA algorithm with the iteration~\eqref{eq:saiter} and gradient estimate~\eqref{eq:spsa} for a problem satisfying Assumptions~\ref{as:boundbias} and~\ref{as:lipgrad}. In addition, let the SPSA algorithm satisfy the conditions of Lemma~\ref{lem:spsaacc} together with there existing $B_3>0$ such that $[\bs \Delta_k]_i^{-1}\le B_3$ almost surely for all $k$. Then it holds that,
\begin{enumerate}
    \item 
Consider a fixed budget of $K$ iterations. Let $\alpha_k\equiv\alpha/\sqrt{K}$ with $\alpha\le L/2$. It holds that the iterates satisfy:
\begin{equation}\label{eq:totalconvspsafixed0}
   \min\limits_{k\in[K]} \mathbb{E}\|\nabla f(\bs \theta_k)\|^2 \le \frac{4(f_c(\bs \theta_k)-f^*_{c})}{\alpha\sqrt{K}} +\frac{2L\alpha\left[\frac{\sigma^2 B_3^2}{c_k^2}+2b_0 c_k^2+\frac{2b^2 B_3^2}{c^2_k}\right]}{\sqrt{K}} +\left[b_0 c_k^2+\frac{b^2 B_3^2}{c^2_k}\right]
\end{equation}

\item Consider a diminishing stepsize of $\alpha_k=\alpha/k^{\gamma}$ with $\gamma\in (1/2,1]$ and $\alpha\le L/2$. Then it holds that, taking,
\[
\bs \theta_R\sim \bs \theta_k\text{ w.p. } \alpha_k/\sum\limits_{l=1}^K \alpha_l
\]
we have,
\begin{equation}\label{eq:totalconvspsadim0}
   \mathbb{E}\|\nabla f(\bs \theta_R)\|^2 \le 4\frac{f_c(\bs \theta_k)-f^*_{c}}{\alpha K^\gamma} +2\frac{L\alpha\left[\frac{\sigma^2 B_3^2}{c_k^2}+4b_0 c_k^2+\frac{4b^2 B_3^2}{c^2_k}\right]}{ K^{\gamma}} +\left[b_0 c_k^2+\frac{b^2 B_3^2}{c^2_k}\right]
\end{equation}
\end{enumerate}
\end{theorem}

\begin{proof}
\label{secondproof}
Consider $\hat{\bs g}_k$ to be the SPSA estimate with zero bias in function evaluations and $\bs g_k$ the actual estimate. Recalling Lemma~\ref{lem:spsaacc}, we have that
\[
\bs g_k -\nabla f(\bs\theta_k) = \bs g_k-\hat{\bs{g}}_k+\hat{\bs g}_k-\nabla f(\bs\theta_k) =
\bs g_k-\hat{\bs{g}}_k+\bs r(\xi^1_k,\xi^2_k,\bs\Delta_k)+\bs b_0(\bs\theta_k)
\]
Now using the characterization of the function evaluation~\eqref{eq:biasas} as well as the definition of $\bs{g}_k$, we get,
\[
[\bs g_k]_i-[\hat{\bs{g}}_k]_i = \frac{\hat{r}_k+\hat{b}_k}{2 c_k[\bs \Delta_k]_i} 
\]
with $\mathbb{E}[\hat{r}_k]=0$ and $\|\hat b_k\|\le 2 b$ a.s.. Now, recalling the assumptions on $\bs\Delta_k$ we have that,
\[
\left|[\bs g_k]_i-[\hat{\bs{g}}_k]_i\right|\le |\bar{r}_k|+|b^1_k|
\]
with $\mathbb{E}|\bar{r}_k|\le \frac{\sigma B_3}{c_k}$ and $|b^1_k|\le bB_3/c_k$ almost surely. Putting these quantitites together we get that,
\[
\bs g_k -\nabla f(\bs\theta_k) = \bs r_s(\xi^1_k,\xi^2_k,\bs\Delta_k)+\bs b_s(\bs\theta_k)
\]
with $\mathbb{E}\bs r_s=0$ and $\|\bs b_s(\bs\theta_k)\|_{\infty}\le b_0c_k^2+\frac{b B_3}{c_k}$.

Now we can proceed by similar lines as in the proof of Theorem~\ref{th:twopoint}, noting that in this case, although the gradient error estimate is more complex, the accuracy is with respect to the original function as opposed to a smoothed variant, thus permitting a more simplified exposition of the main convergence result. 

Indeed, starting with the standard Descent Lemma on $f(\bs \theta)$,
\[
\begin{array}{l}
f(\bs \theta_{k+1}) \le f(\bs \theta_k)-\alpha_k \langle \nabla f(\bs \theta_k),\bs g_k\rangle +\frac{L}{2}\alpha_k^2 \|\bs g_k\|^2 \\
\quad = f(\bs \theta_k)-\alpha_k\|\nabla f (\bs \theta_k)\|^2-\alpha_k \langle \nabla f(\bs \theta_k),\bs r_s(\xi^1_k,\xi^2_k)+\bs b_s(\bs \theta_k)\rangle +\frac{L}{2}\alpha_k^2 \|\bs g_k\|^2
\end{array}
\]
Summing up, we get, with $f^*$ a lower bound for $f(\bs \theta)$,
\begin{equation}\label{eq:sumdescspsa}
    \sum\limits_{k=1}^K \alpha_k \mathbb{E}\left[\|\nabla f(\bs \theta_k)\|^2|\mathcal{F}_k\right] \le f(\bs \theta_k)-f^* - \sum\limits_{k=1}^K \alpha_k \mathbb{E}\left[\langle \nabla f(\bs \theta_k),\bs r_s(\xi^1_k,\xi^2_k)+\bs b_s(\bs \theta_k)\rangle|\mathcal{F}_k\right]+\frac{L}{2}\sum\limits_{k=1}^K \alpha_k^2 \mathbb{E}\left[\|\bs g_k\|^2|\mathcal{F}_k\right]
\end{equation}
Continuing, from Lemma~\ref{lem:spsaacc} and Assumption~\ref{as:boundbias},
\[
\mathbb{E}[\|\bs g_k\|^2|\mathcal{F}_k]\le 2\|\nabla f(\bs \theta_0)\|^2+\mathbb{E}[\|\bs r_s(\xi^1_k,\xi^2_k)\|^2|\mathcal{F}_k]+2\|\bs b_s(\bs \theta_k)\|^2 \le 2\|\nabla f(\bs \theta_k)\|^2+\frac{\sigma^2 B_3^2}{c_k^2}+4b_0 c_k^2+\frac{4b^2 B_3^2}{c^2_k}
\]
and, using that $\mathbb{E}[\bs r_s(\xi^1_k,\xi^2_k)|\mathcal{F}_k]=0$,
\[
\begin{array}{l}
-\mathbb{E}[\langle \nabla f(\bs \theta_k),\bs r_s(\xi^1_k,\xi^2_k)+\bs b_s(\bs \theta_k)\rangle|\mathcal{F}_k] = -\langle \nabla f(\bs \theta_k),\bs b_s(\bs \theta_k)\rangle \\ \qquad \le \frac{\|\nabla f(\bs \theta_k)\|^2}{2}+\frac{\|\bs b_s\|^2}{2} \le \frac{\|\nabla f(\bs \theta_k)\|^2}{2}+b_0 c_k^2+\frac{b^2 B_3^2}{c^2_k}
\end{array}
\]
Incorporating these two bounds into~\eqref{eq:sumdescspsa} and taking total expectations, we obtain,
\begin{equation}\label{eq:totalconvspsa}
\sum\limits_{k=1}^K \frac{\alpha_k}{2}\left(1-2\alpha_k L\right) \mathbb{E}\|\nabla f(\bs \theta_k)\|^2 \le f(\bs \theta_0)-f^* +\frac{L\left[\frac{\sigma^2 B_3^2}{c_k^2}+4b_0 c_k^2+\frac{4b^2 B_3^2}{c^2_k}\right]}{2}\sum\limits_{k=1}^K \alpha_k^2 +\left[b_0 c_k^2+\frac{b^2 B_3^2}{c^2_k}\right]\sum\limits_{k=1}^K \alpha_k
\end{equation}

We can now consider budget and diminishing step-size regimes. In the first case, we have a fixed budget of $K$ iterations and take $\alpha_k:=\frac{\alpha}{\sqrt{K}}$ with $\alpha\le L/4$. The above expression becomes
\begin{equation}\label{eq:totalconvfixed}
   \frac{1}{4}\sqrt{K} \alpha\min\limits_{k\in[K]} \mathbb{E}\|\nabla f(\bs \theta_k)\|^2 \le f_c(\bs \theta_k)-f^*_{c} +\frac{L\alpha^2 \left[\frac{\sigma^2 B_3^2}{c_k^2}+2b_0 c_k^2+\frac{2b^2 B_3^2}{c^2_k}\right]}{2} +\alpha\left[b_0 c_k^2+\frac{b^2 B_3^2}{c^2_k}\right]\sqrt{K}
\end{equation}
\end{proof}

\section{Bias in VQA and Its Effects on Convergence Properties}\label{sec:bias}

\paragraph*{Noise models in quantum circuits}\label{sec:noise}


In the circuit model of quantum computation, an ideal quantum gate performs a unitary transformation on the quantum state, which is described by density matrices $\rho \in \mathbb{S}_+^d$. The unitary transformation acts on an input state $\rho_{\rm{in}}$ as $U\rho_{\rm in}U^\dagger = \rho_{\rm out}$. 
In practice, real quantum hardware is very noisy because of its unavoidable interaction with the environment. This corresponds to an \emph{open quantum system} whose dynamics is never purely unitary.

In the most general setting, the evolution of states in open quantum systems can be described as $\rho_{\rm out}= \mathcal{E}(\rho_{\rm in})$ by a linear operator $\mathcal{E}: \mathbb{S}_+^d \to \mathbb{S}_+^d$ via the (nonunique) Kraus sum representation
\begin{align}\label{eq:Kraus}
    \mathcal{E}(\rho) \coloneqq \sum_{k=1}^\ell E_k \rho_{\rm in} E_k^\dagger,
\end{align}
where $\{E_k\}_{k=1}^\ell$ are called Kraus operators and satisfy $\sum_{k=1}^\ell \hat{E}_k^\dagger \hat{E}_k  = \mathds{1}$,\, $\ell \leq d^2-1$.

Let us consider some concrete examples of how Kraus operators appear in the noise typical noise channels of the gate model of quantum computing with the standard Pauli gates $\{ X,Y,Z\}$.
\begin{enumerate}
    \item[\emph{(1)}] The \emph{bit flip} with probability $p$ is described as $\mathcal{E}_{}(\rho)=(1-p) \rho+ p X \rho X^\dagger$ with the Kraus operators $\{\mathds{1}, X \}$.
    \item[\emph{(2)}] The \emph{dephasing channel} is described as $\mathcal{E}_{}(\rho)=(1-p) \rho+p Z \rho Z^\dagger$ with the Kraus operators being $\{\mathds{1}, Z \}$.
    \item[\emph{(3)}] The \emph{depolarizing channel} is described as $\mathcal{E}_{}(\rho)=(1-p) \rho+\frac{p}{4}(\rho+X \rho X^\dagger+Y \rho Y^\dagger+Z \rho Z^\dagger)=\frac{p}{2} \mathds{1} +(1-p) \rho$.
    \label{eq:depolarizing}
\end{enumerate}

\paragraph*{Bias in Quantum Circuits}\label{sec:biascirc}
Errors are endemic in the measurements of quantum states after a circuit run. A browse through any contemporary textbook or broad survey (e.g.,~\cite{NielsenChuang,benenti2019principles,LaGuardia2020}) indicates the degree to which research is being undertaken in the design and implementation of error-correcting codes as well as in hardware design for minimizing noise, attributes that are absolutely essential for the commercial implementation of fault-tolerant error corrected quantum computers as well as the NISQ computers. In fact, it is widely accepted that the latter can handle only shallow quantum circuits due to the presence of noise, which, with increasing depth, makes the output of any measurement progressively less reliable. 


Of course, the presence of noise itself does not indicate bias, and even what would be considered a lot of noise could be unbiased but with a high variance. However, at first glance, it is clear that bias is inherently present in the computation of noisy quantum circuits. There are two sources of bias in the noise:
\begin{enumerate}
    \item[\emph{(1)}] Bias due to certain noise channels that favor certain operation errors, e.g., Pauli errors, more frequently compared to other operations.
    \item[\emph{(2)}] Bias in the measurement, wherein we obtain positive counts for states that should appear with probability 0.  
    \item[\emph{(3)}] Bias in the estimators used for reconstructing quantum states (quantum state tomography) or quantum processes (quantum process tomography). 
\end{enumerate}

As far as point \emph{(1)} is concerned, consider, for example, a qubit system with Hamiltonian proportional to $Z$ where a common noise model description for this system is
that of the dephasing channel, which occurs at rates
much greater than the rates of other noise channels. Specifically, consider the case where the energy eigenstate basis of the state of a qubit before the execution of the circuit matches the computational basis $\{\ket{0}, \ket{1}\}$. The noise induced by the dephasing channel is much higher, for example, than the noise induced by bit flips. For example, in ~\cite{tannu2019mitigating} it was reported that simply measuring a zero state had $82\%$ fidelity (and the state of all ones $62\%$) on a state-of-the-art quantum device. If the circuit and the observable operator are designed so that, in the theoretically ideal scenario, there is one eigenstate that is measured with probability one, then the expectation \emph{should be} that eigenstate. However, suppose that one is measuring one state most of the time and another state some other percentage of the time, in practice, for this circuit, that is, an evaluation $f(\bs \theta)\sim E_1$ with probability $p_1$ and $f(\bs \theta)\sim E_2>E_1$ with probability $p_2$. In such a basic scenario, if $E_1$ is intended, then clearly all the noise is one-directional and thus biased. Furthermore, the noise is not even consistent since we expect repeated trials of this to continue to exhibit one-sided noise. As a matter of fact, in several types of physical qubit implementations certain types of errors over-dominate others, for example in solid state spin qubits \cite{Shulman2012,Watson2018} and superconducting qubits \cite{Pop2014}. This is a first indication that the implemented circuits are inherently biased. 

However, in the context of the practical implementation of VQAs, point \emph{(3)} is what really matters. Specifically, we are interested in how the noise in the actual measurements affects the values of the observables. As an example, consider the depolarizing channel, 
\begin{align}
    \mathcal{E}(\rho) =p\mathds{1}+(1-p)\rho
\end{align}
which replaces each qubit (or the entire system) in the state $\rho$ with the totally mixed state with some probability $p > 0$. Subsequently, consider that an observable $O$ is applied to compute the energy, entropy, or entanglement, for instance. It is clear that if in all circuit runs the state $\rho$ is the input of $\mathcal{E}$, any one evaluation will not have an expectation $\braket{O} \coloneqq {\rm Tr}(\rho O)$, rather $\braket{\tilde O} \coloneqq {\rm Tr}(\mathcal{E}(\rho) O)$, and no matter how many trials are performed, the sample average will not equal this desired value, i.e., 
\begin{align}
    \braket{O} \neq \mathbb{E}(\tilde{O})\coloneqq \frac{1}{N}\sum\limits_{i=1}^N f_i,\quad f_i\sim \mathbb{E}_i(\tilde{O}),
\end{align}
where $f_i$ corresponds to a measurement. Said differently, the estimator $\tilde{O}$ of some observable $O$ is biased in the sense $\mathbb{E}[O-\tilde{O}] \neq 0$. It is known that, for example, in IBM machines, there exists an asymmetric read-out error \cite{Nachman2020} since
$p(\ket{0}|\ket{1}) > p(\ket{1}|\ket{0})$. 
Thus, it is clear that the noise arising from a classical measurement applied to a quantum circuit, as is commonly done in, e.g., VQA, is not unbiased, as far as the expectation of the (parametrized) estimator is concerned, and thus any analysis of algorithms with function evaluations of this sort must incorporate this for faithful modeling of the problem. 


Next, let us consider the bias in the measurement (2). To give an intentionally simple example, consider a Bell state in a 2-qubit register obtained performing a Hadamard gate followed by a CNOT gate. Figure~\ref{Fig:ibmsimplegate}
presents the corresponding measurements obtained in hardware (\texttt{ibmq\_santiago}) and in simulation (using the \texttt{qiskit} Aer simulator). 
Table~\ref{table:simpleexprobs} presents the same data numerically. 

\begin{figure*}[t]
    \centering
    \includegraphics[width = 0.48\columnwidth]{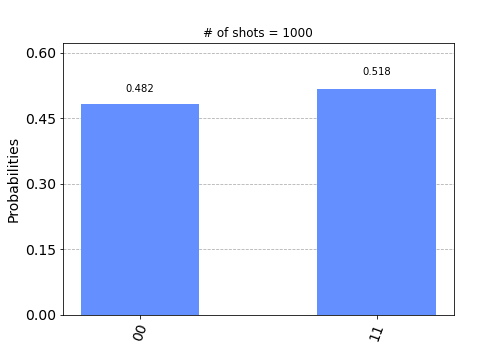}
    \includegraphics[width = 0.48\columnwidth]{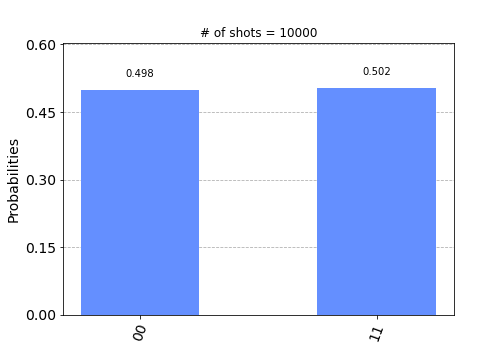}
    \includegraphics[width = 0.48\columnwidth]{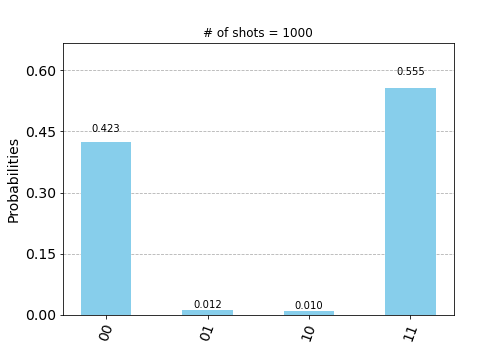}
    \includegraphics[width = 0.48\columnwidth]{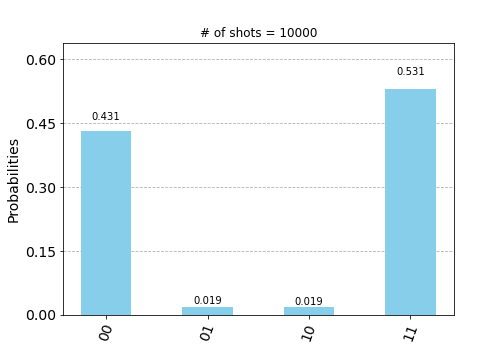}
    \caption{Bias in the measurement of a Bell state: The top two plots present the measurements of the circuit of Fig. \ref{fig:HCNOT}, in a noise-free simulation. 
    The bottom two plots present the corresponding measurements from IBM's 5-qubit device \texttt{ibmq\_santiago}. The fact that states $\ket{01}$ and $\ket{10}$ appear with non-zero probability suggests there is bias, which is independent of the number of shorts. 
    }
    \label{Fig:ibmsimplegate}
\end{figure*}

\begin{table}[tb]
\centering
\begin{tabular}{ccc} \Xhline{5\arrayrulewidth} 
State $\ket{\psi_f}$ & Analytical & Hardware   \\ \hline
$\ket{00}$    & 0.5 & 0.43  \\
$\ket{01}$    & 0 & 0.02  \\
$\ket{10}$    & 0  & 0.02  \\ 
$\ket{11}$    & 0.5  & 0.53  \\ \hline
\end{tabular} 
\caption{A numerical comparison of the probabilities for the Bell state of Figure~\ref{Fig:ibmsimplegate}. 
}
\label{table:simpleexprobs}
\end{table}
Errors that result in inaccurate sample averages of particular states are, of course, well-known issues that are endemic to near-term devices, and a host of error mitigation strategies exist to reduce the magnitude of said errors. 

However, regardless of the quantitative improvement in the frequency of errors, there would remain some nonzero appearance of the mixed states $\ket{01}$ and $\ket{10}$. This in itself makes it the case that there are observables for which, when they define the objective of a VQA circuit, there will always be bias, resulting in a fundamental discrepancy between the ideal distribution $\Xi$ and the observed distribution of objective evaluations $\chi$. 

As a simple example, consider the operation of simply projecting onto the $\ket{01}$ state. This would correspond to taking the observable $O \in \Herm(\mathbb{C}^4)$ as the Hermitian matrix,
\[ O =
\begin{pmatrix}
0 & 0 & 0 & 0 \\
0 & 1 & 0 & 0 \\
0 & 0 & 0 & 0 \\
0 & 0 & 0 & 0
\end{pmatrix}.
\]
Define $F(\bs\theta,\xi)$ as the output of $\braket{\psi_0|U^\dagger(\bs \theta)O U(\bs \theta)|\psi_0}=\braket{\psi_f|O |\psi_f}$. Note that there is no parameter $\bs\theta$ that enters the expression for $\ket{\psi_0}$ or the gates $U$ or $O$, and so 
$F$ is constant with respect to $\bs\theta$, and only the noise affects the outcome of its evaluation.

We have that $F(\cdot,\cdot)$ under the ideal distribution $\Xi$ and the actual distribution $\chi$ are,
\[
F(\bs \theta,\cdot) \sim_{\Xi} 
0  \text{ w.p. }1,\quad
F(\bs \theta,\cdot) \sim_{\chi} \left\{ \begin{array}{lr}
0 & \text{w.p.}\approx 0.98,\\
1 & \text{w.p.}\approx 0.02\end{array}\right.,\,
\]
thus with $f(\bs\theta)\equiv 0$ as the desired expectation, we have,
\begin{align}
F(\bs\theta,\xi)= f(\bs\theta)+n(\bs\theta,\xi)= f(\bs\theta)+r(\bs\theta,\xi)+b(\bs\theta)
\end{align}
with $b(\bs\theta)=(\mathbb{E}_{\chi}-\mathbb{E}_{\Xi})[F(\bs\theta,\cdot)]=0.02$ the bias and $r(\bs\theta,\xi)=0$ the additional statistical variation. Thus, it can be clearly seen that even in the simple example, bias in the function evaluation is present, suggesting that the discrepancy in the mean and limiting sample average between the ideal and observed values of noisy VQA function evaluations is an endemic issue, and classic optimization convergence analyses that do not take this into account are not directly applicable to studying the minimization of VQA circuits.
Note that increasing the sample size does not diminish the error percentage, indicating that the procedure is not statistically consistent.

\begin{figure}[t]
    \centering
    \quad\quad
    \includegraphics[scale=1.2]{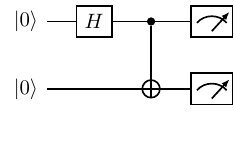} \qquad 
    \caption{An intentionally simple quantum circuit preparing a Bell state.}
   \label{fig:HCNOT}
\end{figure}


\section{Gate definitions}
\label{sec:gatedef}

For the sake of completeness, let us reiterate the definition of gates used in Figure \ref{fig:qaoa}:


\begin{align*}
 U_2(\phi,\lambda) & = \frac{1}{\sqrt{2}}\left(\begin{array}{cc}
1 & -e^{i \lambda} \\
e^{i \phi} & e^{i(\phi+\lambda)}
\end{array}\right),  \qquad
 R_z(\lambda) = \left(\begin{array}{cc}
e^{-i \frac{\lambda}{2}} & 0 \\
0 & e^{i \frac{\lambda}{2}}
\end{array}\right), \\
R(\theta,\phi) & = \left(\begin{array}{cc}
\cos \frac{\theta}{2} & -i e^{-i \phi} \sin \frac{\theta}{2} \\
-i e^{i \phi} \sin \frac{\theta}{2} & \cos \frac{\theta}{2}
\end{array}\right).
\end{align*}



\end{document}